\documentclass[journal,twoside,web]{IEEEtran}
\usepackage{generic}
 \usepackage{cite}
 \usepackage{amsmath,amssymb,amsfonts}
\usepackage{algorithmic}
\usepackage{graphicx}
\usepackage{textcomp}
\usepackage{amsmath} % assumes amsmath package installed
\usepackage{amsfonts}
\usepackage{algorithm}
\usepackage{algorithmic}
\newtheorem{assumption}{Assumption}
\newtheorem{theorem}{Theorem}
\newtheorem{lemma}{Lemma}
\newtheorem{proposition}{Proposition}
 \newtheorem{corollary}{Corollary}
 \newtheorem{definition}{Definition}

\def\BibTeX{{\rm B\kern-.05em{\sc i\kern-.025em b}\kern-.08em
    T\kern-.1667em\lower.7ex\hbox{E}\kern-.125emX}}
\begin{document}
\title{Regret-Guaranteed Safe Switching with Minimum Cost: LQR Setting with Unknown Dynamics}
\author{Jafar Abbaszadeh Chekan and C\'edric Langbort
\thanks{This paragraph of the first footnote will contain the date on 
which you submitted your paper for review. It will also contain support 
information, including sponsor and financial support acknowledgment. For 
example, ``This work was supported in part by the U.S. Department of 
Commerce under Grant BS123456.'' }
\thanks{J.~A.~Chekan and C.~Langbort (emails: jafar2 \& langbort@illinois.edu) are with the Coordinated
  Science Laboratory and the Department of Aerospace
  Engineering at the University of Illinois at Urbana-Champaign (UIUC).}
}

\maketitle

\begin{abstract}
Externally Forced Switched (EFS) systems represent a subset of switched systems where switches occur deliberately to meet an external requirement. However, fast switching can lead to instability, even when all closed-loop modes are stable. In this study, our focus is on an EFS scenario with \textit{unknown system dynamics}, where the next mode to switch to is revealed by an external entity in real-time as the switch occurs. The challenge is to track the revealed sequence while (1) minimizing accumulated cost in a regretful sense and (2) ensuring that the norm of the system's state does not grow excessively-a property we refer to as 'the safety of switching.' Achieving the latter involves requiring the closed-loop system to remain in each revealed mode for some minimum dwell time, which must be learned online. We propose an algorithm based on the principles of Optimism in the Face of Uncertainty. This algorithm jointly establishes confidence sets for unknown parameters, devises a feedback policy, and estimates a minimum dwell time for each revealed mode from data. By precisely estimating dwell-time error, our strategy yields an expected regret of $\mathcal{O}(|M| \sqrt{ns})$, where $ns$ and $|M|$ denote the total switches and mode count, respectively. We benchmark this approach against scenarios with known parameters.

\end{abstract}


\section{Introduction}
\label{Introduction}
\IEEEPARstart{G}{reat} strides have been made over the past few years in the application of learning-based and data-driven control techniques to both linear and nonlinear systems whose parameters are unknown \cite{boffi2021regret, kakade2020information, chen2021black, dean2020sample, cohen2018online, mania2019certainty, cohen2019learning}. One class of plants for which progress in this direction is currently not as mature, however, is switched systems (i.e., systems consisting of a number of continuous-time modes the transition between which is governed by some discrete protocol\cite{zhu2015optimal, zhao2008stability, lin2009stability, liberzon1999stability}) due to the complexities brought about in the learning process by the combination of both continuous and discrete aspects.

In this work, we are primarily interested in developing such learning-based tools for unknown switched systems in the so-called externally forced switching (EFS) scenarios, in which either the time of the switch, the sequence of modes, or both are dictated to the system by an external entity \cite{zhu2015optimal}. In particular, we focus on a setting where only the next mode is dictated, and the decision regarding when to switch to that mode is made by the system itself. More precisely, the next mode to switch to is revealed in an online fashion, as the switch actually occurs, and the problem is to follow the revealed sequence while (1) minimizing accumulated cost in a regret sense and, (2) ensuring that the norm of system's state does not grow inordinately -- a property we refer to as `the safety of switching'. This safety concern removes the option of quick switching (which would appear rational if the goal were solely cost minimization) since rapid switching can result in state explosion, even in situations where the individual closed-loop modes are stable \cite{liberzon2003switching, hespanha1999stability, zhao2011stability}.  A typical approach to avoid such state explosion is to ensure that the system remains in each mode for a minimal specified duration, referred to as the minimum dwell time, whose computation usually requires system model \cite{lin2009stability, liberzon1999stability}. The central difficulty, then, in the absence of an a priori knowledge of the modes' parameters or even of the full sequence of switches (which would allow for computation of an average dwell-time), is to compute minimum dwell times for all modes dynamically and directly from data. Our goal in this paper is to design a strategy that involves specifying the switch sequence and feedback gain design in a scenario where the system's dynamics is unknown, and there is only noisy access to the state of the sub-systems (modes). We refer to this setting as "unknown modes-unknown sequence" ($\bar{M}\bar{S}$). To address this problem we first analyze the problem in two distinct scenarios: one where both the mode and sequence are known ($MS$), and another where the mode is known but the sequence is concealed ($M\bar{S}$). The $M\bar{S}$ scenario, which employs the same information disclosure mechanism as our proposed problem, serves as a reasonable benchmark for comparing the designed algorithm for the $\bar{M}\bar{S}$ setting. We demonstrate that achieving optimal solutions for the $M\bar{S}$ scenario requires more than just knowledge of the next mode. This leads us to propose a practical approach for managing the $M\bar{S}$ setting that involves implementing the policy obtained through solving Discrete Algebraic Riccati Equation (DARE) and crafting a corresponding mode-dependant dwell time strategy. Afterward, we apply an Optimism in the Face of Uncertainty (OFU) based algorithm to tackle the challenge in the $\bar{M}\bar{S}$ setting. Our proposed strategy aims to design minimum dwell time and feedback policy fully online, relying solely on state measurements, to enable switching according to the given sequence with a guaranteed regret comparing to a reference strategy for $M\bar{S}$ setting. Our main contributions are as follows 
\begin{itemize}
    \item We propose a computationally efficient algorithm to estimate dwell time in an unknown dynamic setting.
    \item We prove that the estimated dwell time is close to the true but unknown dwell time by providing an upper bound for the estimation error.
    % \item We have demonstrated that the expected time for switching in a sequence with $ns$ switches among $|\mathcal{M}|$ modes, while minimizing the cost, is approximately $\mathcal{O}\big(|\mathcal{M}| \sqrt{ns}\big)$.
    \item  We prove that our proposed algorithm has regret of $\mathcal{O}(|\mathcal{M}|\sqrt{ns})$ compared to the case when the parameters of all modes are known. 
\end{itemize}

At a high level, our work differs from the existing literature in terms of setup (i.e., (1) what is (not) known a priori about the modes and switching sequence, and (2) whether the switching signal is itself actionable), closed-loop goals (mere stability vs. performance guarantees) and methods used (direct vs. no identification of unknown parameters). For example, \cite{kenanian2019data} use a set of random observations from different trajectories of the switched system to perform probabilistic stability analysis even in the absence of knowledge of the modes and switching rule.  The authors of \cite{rotulo2022online} propose a data-driven control approach for an unknown discrete-time linear system that switches between a finite number of modes, with the switching signal being fully unknown. They design control rules, which can automatically adapt to changes in the actuation mode, based solely on data and without explicit identification. Likewise, in \cite{dai2018moments} and \cite{dai2022convex}, stabilizing controllers are designed for linear switched systems without explicit plant identification. These controllers rely solely on experimental data obtained from an experiment that thoroughly excites all the modes. The designed switched state feedback controller is robust in the sense that it guarantees stability for any switching sequence. In the work presented by \cite{eising2023data}, the authors put forth an online switching controller that accomplishes both mode detection and stabilization. This is achieved solely by utilizing a finite set of trajectories for each mode, which are provided offline.

In contrast with these works, which consider solely stabilization as their goal, Du et al., in  \cite{du2022data}, consider an unknown Markov jump system within a Linear Quadratic (LQ) framework and propose an algorithm that jointly learns the unknown parameters of the modes and the Markov transition matrix that governs the evolution of the mode switches, while achieving guaranteed regret.  The authors of \cite{li2023online} also consider regret minimization as part of their performance objective but, in a manner most closely related to our setup, also allow for the switching time between modes to be part of the control design and to be determined online. Their algorithm employs a multi-armed bandit like analysis and, by relying on cost measurements, switches online between policies drawn from a finite candidate pool in an effort to identify the best controller. This switching adheres to a minimum dwell-time policy of the kind mentioned earlier, but uses a conservative bound for this dwell-time, which impacts the performance of the closed-loop system. This is also the case in our own previous work \cite{chekan2022learn}, where we introduced a projection-equipped algorithm for unknown switching over-actuated systems.

At the technical level, our approach builds upon existing model-based RL algorithms, which come with guaranteed regret and use system identification as a core phase.  We use the principle of Optimism in the Face of Uncertainty (OFU), as originally introduced in \cite {campi1998adaptive} and later strengthened in various ways in \cite{abbasi2011regret,ibrahimi2012efficient,lale2022reinforcement, chekan2021joint, cohen2019learning}, to obtain sublinear regret and/or accelerate the stabilization of the closed-loop system.  More precisely, we exploit the relaxed-semidefinite programming (SDP) formulation of \cite{pmlr-v97-cohen19b} in our new, switched-system, context.

\textbf{Paper Structure.} The remainder of the paper is organized as follows: Section \ref{sec:probStat} presents the problem statement. In Section \ref{eq:knownsetting}, we discuss the solution when all system parameters are known. Section \ref{sec:prelim} provides a brief review of constructing confidence sets (system identification) and the relaxation of the standard LQR-SDP formulation for control and dwell-time design. The proposed algorithm and its main steps are presented in Section \ref{sec:solution}. Section \ref{sec:guarantee} summarizes the regret bound guarantees and stability analysis, including dwell-time design. Detailed analysis and proofs are provided in the Appendix.

\textbf{Notations.}
$\|M\|_*=tr(\sqrt{M^\top M})$ is Frobenius norm of matrix $M$ and $\underline{\lambda}(M)$ and $\overline{\lambda}(M)$ are its minimum and maximum eigenvalues. For a set $\mathcal{S}$, $|\mathcal{S}|$ denotes its cardinality.

\section{Problem Statement and Formulation} \label{sec:probStat}

Consider a time-invariant switched LQR system
\begin{align}
x_{t+1} &=A^{\sigma(t)}_{*}x_{t} + B^{\sigma(t)}_{*}u_{t}+\omega_{t+1}\label{eq:dyn_atttt} \\
c^{\sigma(t)}_{t} &=x_{t}^\top Q^{\sigma(t)}x_{t} + u_{t}^\top R^{\sigma(t)}u_{t}\label{eq:obs}
\end{align}
where $\mathcal{M}$ is an index set and $\sigma:\{0\}\cup \mathbb{N}\rightarrow \mathcal{M}$ is right-continuous piecewise constant switching signal that specifies at any time $t$ which mode is active. $A^{i}_{*}$, $B^{i}_{*}$, $\forall i\in \mathcal{M}$ are matrices of each individual modes which are initially unknown to the learner and $Q^{i}$, $R^{i}$ are known cost matrices.

We have the following assumption for the process noise which is standard in controls community (see \cite{abbasi2011regret, cohen2019learning, lale2022reinforcement}).

\begin{assumption}
\label{Assumption 1}
There exists a filtration $\mathcal{F}_{t}$ such that

$(1.1)$ $\omega_{t+1}$ is a martingale difference, i.e., $\mathbb{E}[\omega_{t+1}|\;\mathcal{F}_{t}]=0$

$(1.2)$ $\mathbb{E}[\omega_{t+1}\omega_{t+1}^\top|\;\mathcal{F}_{t}]=\bar{\sigma}_{\omega}^2I_{n}=:W$ for some $\bar{\sigma}_{\omega}^2>0$;

$(1.3)$ $\omega_{t}$ are component-wise sub-Gaussian, i.e., there exists $\sigma_{\omega}>0$ such that for any $\gamma \in \mathbb{R}$ and $j=1,2,...,n$
\begin{align*}
\mathbb{E}[e^{\gamma\omega_{j}(t+1)}|\;\mathcal{F}_{t}]\leq e^{\gamma^2\sigma_{\omega}^2/2}.
\end{align*}
\end{assumption}

Let $\mathcal{I} = \{i_0, i_1, \ldots, ..., i_n\}$ represent the sequence of modes in which the system switches. This sequence is not known in advance and is revealed gradually, meaning that only the next mode to switch to, denoted as $i_{k+1} \in \mathcal{M}$, is disclosed after the current switch, $i_k \in \mathcal{M}$, has occurred. We refer to the time interval between two subsequent switches as an 'epoch.' The termination of the sequence is uncertain and is randomly announced after the "last" mode $i_n$ is revealed. The objective is to develop an algorithm that guarantees actuation according to the set $\mathcal{I}$ with the minimum cost, conditioned on maintaining control over the expected growth of the state norm, as defined below.

\begin{definition} \label{def:underControl}
Let $\mathfrak{T}_{i_{k+1}}^+$ and $\mathfrak{T}_{i_{k}}^+$ represent time sequences immediately following two subsequent switches between any two arbitrary modes $i_{k}$ and $i_{k+1}$ both in $\mathcal{I}$. We define the expected state norm growth as being $(\bar{\alpha}, \bar{\beta})$-under control, where $0 < \bar{\alpha} < 1$ and $\bar{\beta} > 0$, if the following condition holds:

\begin{align}
&\mathbb{E}[x_{\mathfrak{T}_{i_{k+1}}^+}^\top  x_{\mathfrak{T}_{i_{k+1}}^+}| \mathcal{F}_{\mathfrak{T}_{i_{k+1}}-1}] \leq \bar{\alpha} \mathbb{E}[x_{\mathfrak{T}_{i_k}^+}^\top x_{\mathfrak{T}_{i_k}^+}| \mathcal{F}_{\mathfrak{T}_{i_k}-1}] + \bar{\beta}\sigma_{\omega}^2 \label{eq:stategrowthpp}
\end{align}
\end{definition}

In Definition \ref{def:underControl}, it is important to note that $\bar{\alpha}$ is a user-defined parameter, whereas $\bar{\beta}$ is contingent upon the specific policy class for control design and cost function associated with each mode, which will be introduced later. 

To minimize accumulated cost, switching fast seems rational. However, even when the mode characteristics are known it is well known that fast switching between modes can result in increases in a system's state's norm, which jeopardize global stability even if each mode is itself stable \cite{zhu2015optimal, zhao2008stability, lin2009stability, liberzon1999stability}. Accordingly, switching fast can violate state norm growth control in the sense of Definition \ref{def:underControl}. The issue can be alleviated by remaining in each mode for a computable minimum duration. Hence, the key challenge in addressing this problem boils down to the simultaneous design of policies to be applied during different epochs and determining the appropriate timing for switches.

To lay the foundation for developing an algorithm to tackle the proposed problem, known as Unknown-Modes and Unknown-Sequence ($\bar{M}\bar{S}$), it is crucial to thoroughly examine the problem within the settings of (1) Known-Modes and Known-Sequence (${M}{S}$) and (2) Known-Modes and Unknown-Sequence ($M\bar{S}$). This analysis not only enhances our understanding of the problem but also establishes a baseline for defining regret, with its upper-bound serving as a metric to evaluate our proposed algorithm.

Expanding upon the solution for the $M\bar{S}$ setup presented in Section \ref{eq:knownsetting}, we propose an algorithm in Section \ref{sec:solution} to address the $\bar{M}\bar{S}$ setting. Our algorithm relies on high probability estimates of parameters, which are obtained through the system identification procedure outlined in Section \ref{sec:ident}.

By introducing the notation ${\Theta_*^{i}}=(A_*^i,B_*^i)^\top$ we can express (\ref{eq:dyn_atttt}) equivalently as follows:
\begin{align}
	x _{t+1} ={\Theta^i_{*}}^\top z_{t}+\omega_{t+1}, \quad z_t=\begin{pmatrix} x_t \\ u_t \end{pmatrix}. \label{eq:dynam_by_theta} 
\end{align}

for some mode $\sigma(t)=i$.
This formulation, denoted by (\ref{eq:dynam_by_theta}), will be widely used henceforth to represent the dynamics model.

For control design purpose we restrict the policy to belong to the compact sets $\mathcal{S}(\Theta_*^i)$'s specified by certain $\kappa^i_{c}>0$ and $0<\gamma_{c}^i<1$, that denote some stabilizing policies $K_i$ of system parameterized with $\Theta_i^*$ as follows

\begin{align}
    \nonumber \mathcal{S}(\Theta_*^i)=\{&K_i\in \mathbb{R}^{(n_i+m_i)\times n_i}|\\
   & \|K_i\|\leq \kappa^i_{c}, \;  \rho (A_*^i+B_*^iK_i)<1-\gamma_{c}^i\}. \label{eq:policyClass}
\end{align}

Given the policy class (\ref{eq:policyClass}), the following lemma explicitly defines the parameter $\bar{\beta}$ in Definition \ref{def:underControl}. 
\begin{lemma} \label{lem:betaDef}
With the policy class (\ref{eq:policyClass}) employed for control design across all modes $i\in \mathcal{M}$, within the framework of the proposed switching scenario explained in Definition \ref{def:underControl}, if $\mathfrak{T}_{i_{k+1}}-\mathfrak{T}_{i_{k}}\geq \tau_{i_k i_{k+1}}$ for $\tau_{i_k i_{k+1}}\geq 1$ then the expected state norm is $(\bar{\alpha}, \bar{\beta})$-under control with a user-defined parameter $\bar{\alpha}$ and $\bar{\beta}$ defined by
    \begin{align}
        \bar{\beta}:= \frac{{4\alpha^*_1}^2}{{\alpha^*_0}^2}\frac{{\kappa^*_{c}}^4(1+{\kappa^*_{c}}^2)^2}{{\gamma^*_c}^2} \label{eq:betaDef}
    \end{align}
     where $\tau_{i_k i_{k+1}}$ is $\bar{\alpha}$-dependant and $\kappa^*_c=\max_{i\in |\mathcal{M}|}\kappa^i_c$ and $\gamma^*_c=\max_{i\in |\mathcal{M}|}\gamma^i_c$ and $\alpha^*_0 I\preceq Q^i, R^i \preceq \alpha^*_1 I$ for all $i\in \mathcal{M}$.
\end{lemma}

\section{Solution for $MS$ and $M\bar{S}$ Setups} \label{eq:knownsetting}

In this section, we focus on addressing the problem of efficient and safe switching in two distinct scenarios: $MS$ and $M\bar{S}$. For each scenario, we first introduce formulations and then provide an upper-bound for performance gap between their solutions.

% \begin{definition} \label{def:underControl}
% The state norm of the system is under control, i.e., the switch is safe, if the following inequality holds for any $t\geq 0$ 
% % \begin{align}
% % E[x_t^\top P_{K_{i(t)}} x_t| \mathcal{F}_{t-1}] \leq x_0^\top P_{K_{0}} x_0 + \frac{1}{1-\mathcal{X}} \frac{W}{\bar{\eta}}\label{eq:stategrowth}
% % \end{align}

%  	\begin{align}
% 	\nonumber &\mathbb{E}[x^\top (\tau_{q}^+)x(\tau_{q}^+)| \mathcal{F}_{\tau_{q}-1}]\leq\\
%  \nonumber & \frac{2 \mathcal{X} \nu_i }{\alpha^j_0 \sigma_{\omega}^2}\mathbb{E}[x^\top (\tau_{q-1}^+)x(\tau_{q-1}^+)|\mathcal{F}_{\tau_{q-1}-1} ]+ \frac{2 \nu_i}{\alpha_0^j} 
% 	\end{align}

% \begin{align}
% E[x_t^\top x_t| \mathcal{F}_{t-1}] \leq \frac{2\bar{\nu}}{\bar{\alpha}_0 \sigma_{\omega}^2} \mathcal{X}^{swn(t)}x_0^\top x_0 + \frac{2\bar{\nu}}{\bar{\alpha}_0 (1-\mathcal{X})}+\kappa_*^4 \sigma_{\omega}^2 \label{eq:stategrowth}
% \end{align}
% % Here, $K_0$ represents a stabilizing policy applied at $t=0$, $\bar{\eta}:=\min_{i\in \mathcal{M}}\alpha_0^i\sigma_{\omega}^2/\nu_i$, and $\mathcal{X}$ is a fixed user-defined parameter satisfying $0<\mathcal{X}<1$.
% where $swn(t)$ is total number of mode switches by time $t$,  $\bar{\alpha}_0:=\min_{i\in \mathcal{M}}\alpha_0^i$, and $\mathcal{X}\in (0,\; 1)$ is a fixed user-defined contraction parameter, $\bar{\nu}=\max_{i\in \mathcal{M]}}\nu_i$, and $\kappa_*=\max_{i\in \mathcal{M}}\frac{2\nu_i}{\alpha_o^i \sigma_{\omega}^2}$.
% \end{definition}

\subsection {$MS$ setup}

Consider a scenario where the switch sequence $\mathcal{I}$ and the system dynamics for all subsystems are known. In this context, the central challenge is to determine the switch times $\bar{\mathcal{T}}=\{\mathfrak{T}_{i_0}, \ldots, \mathfrak{T}_{i_{ns-1}}\}$ with $\mathfrak{T}_{i_0}=0$, with the starting time point established as $\mathfrak{T}_{i_0}=0$, along with the set of feedback gains $\bar{\mathcal{K}}=\{\{\mathcal{K}_{i_0}\},...,\{\mathcal{K}_{i_{ns-1}}\}\}$ to be applied at specific time instances with the aim of minimizing the cumulative cost. The subset $\{\mathcal{K}_{i_j}\}:=\{K_{\mathfrak{T}_{i_j}}, K_{\mathfrak{T}_{i_j}+1},..., K_{\mathfrak{T}_{i_{j+1}}-1}\}$, whose members belong to $\mathcal{S}(\Theta_*^{i_j})$, represents the sequence of feedback gains to be implemented during the time steps within the epoch following the $j$-th switch. It is note worthy that the cardinally of this sequence is $\mathfrak{T}_{i_{j+1}}-\mathfrak{T}_{i_{j}}$ which is also the decision variable. The formulation of this problem is given by Problem 0, as follows:

\subsubsection{$MS$, (Program 0)} \label{thm:(KMKS_lastEpochGuar0}

\begin{align}
 &\bar{\mathcal{T}}^*_0, \bar{\mathcal{K}}^*_0=\operatorname*{argmin}_{\bar{\mathcal{K}}, \mathcal{T}}  \sum _{k=0}^{ns-1}\sum _{t={\mathfrak{T}}_{i_k}}^{\mathfrak{T}_{i_{k+1}}-1} c_{t}^{i_k} \label{eq:Optimization0}\\
  \nonumber  &\textbf{s.t.}\\
  & \nonumber x_{t+1} =A^{\sigma(t) }_{*}x_{t} + B^{\sigma(t) }_{*}u_{t}+\omega_{t+1}\\
 & \nonumber \sigma(t)=i_k\; \text{for}\; \mathfrak{T}_{i_k}\leq t< \mathfrak{T}_{i_{k+1}}\\
 & \mathbb{E}[x_{\mathfrak{T}_{i_{k+1}}^+}^\top  x_{\mathfrak{T}_{i_{k+1}}^+}| \mathcal{F}_{\mathfrak{T}_{i_{k+1}}-1}] \leq \bar{\alpha} \mathbb{E}[x_{\mathfrak{T}_{i_k}^+}^\top x_{\mathfrak{T}_{i_k}^+}| \mathcal{F}_{\mathfrak{T}_{i_k}-1}] + \bar{\beta}\sigma_{\omega}^2 \label{eq:StateBoundedness}\\
 &\mathfrak{T}_{i_{k+1}}-\mathfrak{T}_{i_k}\geq 1 \label{eq:condSta}
\end{align}

The constraint (\ref{eq:condSta}) imposes the requirement of visiting each mode $i_k \in \mathcal{I}$. We denote the optimum solution of Program 0 as $\bar{\mathcal{T}}_0^{*}=\{\mathfrak{T}_{i_0}^{0_*}, \ldots, \mathfrak{T}^{0_*}_{i_{ns-1}}\}$ and $\bar{\mathcal{K}}_0^{*}=\{\{\mathcal{K}^{0_*}_{i_0}\},...,\{\mathcal{K}^{0_*}_{i_{ns-1}}\}\}$.

Although Problem 0 benefits from having full access to switch sequence information $\mathcal{I}$, solving it remains challenging due to its intricate nature as a complex combinatorial optimization problem. The other challenge arises from the nonconvexity of the set $\mathcal{S}(\Theta_*^{i_k})$, for any $i_k\in \mathcal{I}$ as demonstrated in \cite{fazel2018global} through a counterexample. This nonconvexity significantly complicates the computational resolution of (\ref{eq:Optimization0}). Moreover, the constraint (\ref{eq:StateBoundedness}), i.e., the boundedness of the expected state norm as per Definition \ref{def:underControl}, does not explicitly depend on the design variables $\bar{\mathcal{K}}$ and $\mathcal{T}$ which adds a layer of complexity to Problem 0. Problem 1, as formulated below, aims to explicitly incorporate this constraint into the decision variables by establishing a fixed policy for each epoch. 
Also, thanks to the fixed policy for epochs, Problem 1 slightly reduces the combinatorial complexity associated with solving the $MS$ Problem 0. Consequently, solving the problem simplifies to determining the epoch's duration and the policy to be employed during that epoch. For the context of Problem 1, without loss of generality and for the sake of brevity, we will redefine the decision variables. By letting $\tau_{i_j}\in \mathbb{N}$, we will use $\mathcal{T}=\{\tau_{i_0}, \ldots, \tau_{i_{ns-1}}\}$, to represent the set of epochs' duration, and $\mathcal{K}=\{K_{i_0}, \ldots, K_{i_{ns-1}}\}$ to denote the set of feedback policies applied for each epoch. In this revised formulation, we have $\tau_{i_j i_{j+1}}=\mathfrak{T}_{i_{j+1}}-\mathfrak{T}_{i_{j}}$, and the fixed control gain for an epoch occurring within the time interval $[\mathfrak{T}_{i_{j}},\; \mathfrak{T}_{i_{j+1}})$ is $K_{i_j}$. Having the average expected cost of policy $K_{i_j}$, denoted by $J(\Theta^{i_j}_*,K_{i_j})$ and computed by  

\begin{align*}
   J(\Theta^i_*,K_i)=\lim_{T\rightarrow \infty}\frac{1}{T}\mathbb{E}[\sum_{t=0}^{T-1}x^\top_t(Q^i+R^i{^\top} K_iR^i)x^\top_t]
\end{align*}

the expected accumulated cost during actuating in mode $i$ for time with length of $\tau_{i}$ will be $\tau_{i} J(\Theta_*^{i}, K_{i})$.

Now let $P_{K_i}$ be a certain positive semidefinite (PSD) matrix associated with the fixed $K_i$ to be applied in an epoch, which is defined as follows: 

 \begin{align}
     P_{K_i}=Q^i+K_i^\top R^iK_i+(A^i_*+B^i_*K_i)^\top P_{K_i}(A^i_*+B^i_*K_i). \label{eq:lyapunovEq}
 \end{align}

Then following proposition outlines a sufficient condition to achieve the boundedness of the state in accordance with Definition \ref{def:underControl}.

\begin{proposition}
    Let $\mathfrak{T}_{i_{k+1}}^+$ and $\mathfrak{T}_{i_{k}}^+$ represent time sequences immediately following two subsequent switches from the mode $i_{k}$ to $i_{k+1}$, both within $\mathcal{I}$. The sufficient condition for (\ref{eq:StateBoundedness}), with $\bar{\beta}$ as defined in (\ref{eq:betaDef}), is as follows: 

    \begin{align}
        &{\tau}_{i_k i_{k+1}}:=\mathfrak{T}_{i_{k+1}}-\mathfrak{T}_{i_{k}}\geq\\
       \nonumber  & -\frac{\ln \bar{\rho}(K_{i_k}, K_{i_{k+1}})+\ln \bar{\mathcal{X}}(K_{i_k}, K_{i_{k+1}})-\ln \mathcal{\bar{\alpha}}}{\ln \big(1-\bar{\eta} \big(K_{i_k})\big)} 
    \end{align}
  where 
 \begin{align}
	&\nonumber \bar{\eta} \big(K_{i_k}) := \frac{\underline{\lambda}\big(H_{K_{i_k}}\big)}{\overline{\lambda}\big(P_{K_{i_k}}\big)}, \; \bar{\rho}(K_{i_k}, K_{i_{k+1}}):=\frac{\overline{\lambda}\big(P_{K_{i_{k+1}}}\big)}{\underline{\lambda}\big(P_{K_{i_k}})}\\
& \bar{\mathcal{X}}(K_{i_k}, K_{i_{k+1}}):= \frac{\overline{\lambda}\big(P_{K_{i_k}}\big)}{\underline{\lambda}\big(P_{K_{i_{k+1}}})}, H_{K_{i_k}}=Q^{i_k}+K_{i_k}^\top R^{i_k} K_{i_k}.	\label{eq:operatorsDef}
\end{align}  

\end{proposition}

\subsubsection{$MS$ (Program 1)} \label{thm:(KMKS_epochbyepochGuar}

\begin{align}
  \nonumber  \mathcal{K}^*_1, \mathcal{T}^*_1=& \operatorname*{argmin}_{\mathcal{K}, \mathcal{T}} \sum_{k=0}^{n-1} \tau_{i_k i_{k+1}}J(\Theta_*^{i_k}, K_{i_k})\\
  \nonumber  &\textbf{s.t.}\\
  & \nonumber x_{t+1} =A^{\sigma(t) }_{*}x_{t} + B^{\sigma(t) }_{*}u_{t}+\omega_{t+1}\\
 & \nonumber \sigma(t)=i_k\; \text{for}\; \mathfrak{T}_{i_k}\leq t< \mathfrak{T}_{i_{k+1}}\\
 \nonumber  & \ln \bar{\rho}(K_{i_k}, K_{i_{k+1}})+ {\tau_{i_k}} \ln \bigg(1-\bar{\eta} \big(K_{i_k}\big)\bigg) +\\
& \ln \bar{\mathcal{X}}(K_{i_k}, K_{i_{k+1}}) \leq  \ln \mathcal{\bar{\alpha}} \quad   \quad \quad\quad k=0,1,..,ns-1 \label{eq:conststab}\\
 & \nonumber \mathfrak{T}_{i_{k+1}}=\tau_{i_k i_{k+1}}+\mathfrak{T}_{i_k}, \quad \mathfrak{T}_{i_0}=0\\
   & \tau_{i_k i_{k+1}}\geq 1. \quad \forall k \label{eq:constvis}
\end{align}

 The optimum duration and feedback gains outputted by Program 1 is denoted by $\mathcal{T}^*_1=\{\tau^{1_*}_{i_0 i_1}, \ldots, \tau^{1_*}_{i_{n-1}i_{n}}\}$ and $\mathcal{K}^*_1=\{K^{1_*}_{i_0}, \ldots, K^{1_*}_{i_{ns-1}}\}$.

\begin{proposition} \label{cor:mintimecond}

Problem 1 achieves its minimum in 

\begin{align}
    \nonumber  &{\tau}_{i_k i_{k+1}}:=\\
     &\max\bigg\{1, -\frac{\ln \bar{\rho}(K_{i_k}, K_{i_{k+1}})+\ln \mathcal{X}(K_{i_k}, K_{i_{k+1}})-\ln \bar{\alpha}}{\ln \big(1-\bar{\eta} \big(K_{i_k})\big)}\bigg\}.
\end{align} 
and for any arbitrary switch sequence $\mathcal{I}$, every feasible solution of Program 1 has a cost  $C_1(\mathcal{I})> n \min_ {j\in \mathcal{M}}J_*(\Theta_*^j, Q^j, R^i):=\mathcal{L}_1$.
\end{proposition}

Problems 0 and 1, apart from their computational complexity, require full knowledge of switch sequence that is unavailable in our proposed setting, where only the next mode is revealed to the system and the termination of the sequence is unknown. In the following section, our goal is to investigate the problem in the absence of this information.

\subsection {$M\bar{S}$ Setup}

When confronted with lack of switch sequence information and uncertainty regarding sequence termination, a robust approach with respect to the termination involves minimizing the costs on an epoch-by-epoch basis. The subsequent program outlines this approach.

\subsubsection{ $M\bar{S}$ (Problem 2)}\label{thm:minimum_average_dwell_known_K}
	Consider the switched system (\ref{eq:dyn_atttt}), that at time $\mathfrak{T}_{i_{k+1}}$ switches from mode $i_k$ to the mode $i_{k+1}$ according to the switch sequence $\mathcal{I}$ i.e., $\sigma(\mathfrak{T}_{i_{k+1}}^-)=i_k$ and $\sigma(\mathfrak{T}_{i_{k+1}}^+)=i_{k+1}$. Additionally, let $K_{i_k} \in \mathcal{S}(\Theta_*^{i_k})$ and $K_{i_{k+1}}\in \mathcal{S}(\Theta_*^{i_{k+1}})$ stand for the respective stabilizing linear controllers for these modes. Then the minimum mode-dependent dwell-time, $\tau^{2_*}_{i_k i_{k+1}}$ and policy for the epoch of actuating in mode $i_k\in \mathcal{I}$, $K^{2_*}_{i_k}$ is obtained by solving following program.
\begin{align}
  \nonumber   K_{i_k}^{2_*}, \tau_{i_k i_{k+1}}^{2_*}=& \operatorname*{argmin}_{K_{i_k}\in \mathcal{S}(\Theta_*^{i_k}), \tau_{i_k} \in \mathbb{Z}^{+} }  \tau_{i_k i_{k+1}} J(\Theta_*^{i_k}, K)\\
 &\textbf{s.t.} \label{eq:Optimization1}\\
 &\nonumber x_{t+1} =A^{i_k}_{*}x_{t} + B^{i_k}_{*}u_{t}+\omega_{t+1},\\
 \nonumber  & \tau_{i_k}\geq \bar{\tau}(K_{i_k},K_{i_{k+1}}) 
\end{align}

where
     \begin{align}
     \nonumber &\bar{\tau}(K_{i_k},K_{i_{k+1}}):=\\
     &\max\bigg\{1, -\frac{\ln \bar{\rho}(K_{i_k}, K_{i_{k+1}})+\ln \bar{\mathcal{X}}(K_{i_k}, K_{i_{k+1}})-\ln \bar{\alpha}}{\ln \big(1-\bar{\eta} \big(K_{i_k})\big)}\bigg\}\label{eq:tau_a_known0} 
     \end{align}
with operators  $\bar{\rho}(.,.)$, $\bar{\eta}(.)$ and $\bar{\mathcal{X}}(., .)$ defined by (\ref{eq:operatorsDef}).

The objective of Problem 2 is the product of the average expected cost $J(\Theta_*^{i_k}, K_{i_k})$ incurred by a policy $K_{i_k}$ and its associated dwell time $\tau_{i_k i_{k+1}}$. Noting Proposition \ref{cor:mintimecond}, Problem 2 achieves its minimum when $\tau_{i_k i_{k+1}}=\bar{\tau}(K_{i_k}, K_{i_{k+1}})$ where $K_{i_{k+1}}$ represents the policy of the next epoch. Then problem can be reformulated as follows:

\begin{align}
   K_{i_k}^{2_*}=& \operatorname*{argmin}_{K\in \mathcal{S}(\Theta_*^{i_k}) }  \bar{\tau}(K_{i_k},K_{i_{k+1}})  J(\Theta_*^{i_k}, K) \label{eq:Optimization2}\\
 \nonumber &\textbf{s.t.}\\
\nonumber  & x_{t+1} =A^{i_k}_{*}x_{t} + B^{i_k}_{*}u_{t}+\omega_{t+1}. \label{eq:reformule}
\end{align}

While formulation (\ref{eq:Optimization2}) seems less complex than Problem 2, it still contains two main hurdles. First,  $K_{i_{k+1}}$ is not know as it depends on the policy of the subsequent epoch whose index $i_{k+1}$ is yet to be revealed. In other words, $K_{i_k}^{2_*}$ is a function of unknown policy $K_{i_{k+1}}$. Secondly, solving this problem is computationally intractable because of the complexity of its objective function and the presence of the nonconvex set $\mathcal{S}(\Theta_*^{i_k})$.

To obtain feasible solution given the disclosed information, one possible approach is to enable the system to utilize optimal control feedback obtained by solving DARE at each epoch. This approach indeed specifies the policies to be applied in each epoch and not only eliminates the requirement for future switch sequences but also sidesteps the computational complexity associated with solving a complex optimization problem. The only remaining task is to design a minimum dwell time associated with the solution of DARE. 

The rationale for this approach is supported by employing the following insightful lemma, which provides an expression for the value of $J(\Theta_*, K)$ in terms of the optimal average expected cost $J_*(\Theta_, Q, R)$ and the optimal feedback gain $K_*(\Theta_*, Q, R)$.

\begin{lemma} (Lemma 3 of \cite{mania2019certainty}) \label{lem:deviationLemma}
In the classic Linear Quadratic Regulator (LQR) framework, characterized by the environment parameters $\Theta_*=(A_*, B_*)^\top$ and the cost matrices $(Q, R)$, the application of any arbitrary stabilizing static linear controller $K$ incurs an average expected cost of:

\begin{align}
&J(\Theta_*, K)=J_*(\Theta_*, Q, R)+ \\
 & \nonumber  \operatorname{Tr} (\Sigma_{xx}(K)(K-K_*(\Theta_*, Q,R))^\top \\
&(R+B_*^\top P_*B_*)(K-K_*(\Theta_*, Q,R))\label{eq:expexpcos2}
\end{align}
where $\Sigma_{xx}(K)$ represents the stationary distribution of the covariance matrix of the closed-loop system's state and computed by
\begin{align*}
    (A_*+B_*K)^\top \Sigma_{xx}(K) (A_*+B_*K)-\Sigma_{xx}(K)+\sigma_{\omega}^2 I=0
\end{align*}
\end{lemma}

The expression (\ref{eq:expexpcos2}) can be interpreted as the sum of the optimal average expected cost and an additional term that arises due to deviations from the optimal stationary policy $K_*(\Theta_*, Q, R)$. 

Let us consider that the policy for the next mode, $K_{i_{k+1}}$, is known in advance. This results in the minimum dwell time definition being solely a function of policy $K_{i_{k}}$. Even with this information, solving (\ref{eq:Optimization2}) to minimize $K_{i_{k}}$ is computationally burdensome. Furthermore, it is worth noting that the minimizer of $\tau_{i_k i_{k+1}}$ may not necessarily be the feedback gain obtained by solving DARE. However, such a policy would incur an average expected cost greater than that of the solution from DARE, as explained in Lemma \ref{lem:deviationLemma}. Therefore, deviating from the policy obtained by solving DARE does not necessarily reduce the accumulated expected cost. This rationale justifies adhering to the policy obtained by solving DARE. It makes us independent of the next epoch's policy and allows us to avoid the computational complexity of solving (\ref{eq:Optimization2}) while remaining reasonably close to the optimal policy.

The following program provides the minimum dwell time for the proposed alternative algorithm. 

\subsubsection{Approximate $M\bar{S}$ (Program 3)}\label{thm:minimum_average_dwell_known}
	Consider the switched system (\ref{eq:dyn_atttt}),  that at time $\mathfrak{T}_{i_{k+1}}$ switches from mode $i_k$ to the mode $i_{k+1}$, (i.e., $\sigma(\mathfrak{T}_{i_{k+1}}^-)=i_k$ and $\sigma(\mathfrak{T}_{i_{k+1}}^+)=i_{k+1}$). Let the corresponding dual solutions SDP solutions for these modes be $P_*({\Theta}^{i_k}_*, Q^{i_k}, R^{i_k})$, $P_*({\Theta}^{i_{k+1}}_*, Q^{i_{k+1}}, R^{i_{k+1}})$ and the designed control feedback for the mode $i_k$ be $K_*({\Theta}^{i_k}_*, Q^{i_k}, R^{i_k})$. Then minimum mode-dependent dwell-time for the epoch starting at $\mathfrak{T}_{i_{k}}$ is
     \begin{align}
     \tau_{*}^{{i_k}i_{k+1}}:=\max\bigg\{1, -\frac{\ln {\rho}_*^{{i_k}i_{k+1}}+\ln \mathcal{X}_*^{{i_k}i_{k+1}}-\ln \bar{\alpha}}{\ln \big(1-{\eta }_*^{i_k}\big)}\bigg\}\label{eq:tau_a_known} 
     \end{align}
     where
     \begin{align}
	\nonumber &\eta_*^{i_k} := \frac{\underline{\lambda}\big(H(\Theta_*^{i_k},Q^{i_k},R^{i_k})\big)}{\overline{\lambda}\big(P_*({\Theta}^{i_k}_*, Q^{i_k}, R^{i_k})\big)},\\
 \nonumber &\rho_*^{i_ki_{k+1}}:=\frac{\overline{\lambda}\big(P_*({\Theta}^{i_{k+1}}_*, Q^{i_{k+1}}, R^{i_{k+1}})\big)}{\underline{\lambda}\big(P_*({\Theta}^{i_k}_*, Q^{i_k}, R^{i_k})}\\
\nonumber & 
\mathcal{X}_*^{i_ki_{k+1}}:=\frac{\overline{\lambda}\big(P_*({\Theta}^{i_k}_*, Q^{i_k}, R^{i_k})\big)}{\underline{\lambda}\big(P_*({\Theta}^{i_{k+1}}_*, Q^{i_{k+1}}, R^{i_{k+1}})}\\
\nonumber &
H(\Theta_*^{i_k},Q^{i_k},R^{i_k})=Q^{i_k}+\\
&K_*^\top ({\Theta}^{i_k}_{*}, Q^{i_k}, R^{i_k})R^{i_k} K_*({\Theta}^{i_k}_{*}, Q^{i_k}, R^{i_k}).
\label{eq:HDefknown}
	\end{align}

 We denote the switch time sequence associated with the solution of Program 3 as $\bar{\mathcal{T}}_*=\{\mathfrak{T}^*_{i_0}, \ldots, \mathfrak{T}^*_{i_{n-1}}\}$ where $\mathfrak{T}^*_{i_0}=0$ by definition and $\mathfrak{T}^*_{i_{k+1}}=\mathfrak{T}^*_{i_{k+1}}+ \tau^{i_k i_{k+1}}_*$.

% Definition 1, in the sense of applying solution of Program 4, is rewritten as follows:

% \begin{align}
% \nonumber E[x_t^\top & P({\Theta}^{i(t)}_*, Q^{i(t)}, R^{i(t)}) x_t| \mathcal{F}_{t-1}] \leq\\
% &x_0^\top P({\Theta}^{i_0}_*, Q^{i_0}, R^{i_0}) x_0 + \frac{1}{1-\mathcal{X}} \frac{W}{\bar{\eta}}\label{eq:stategrowth4}
% \end{align}

Considering equation (\ref{eq:tau_a_known}), the undesirable switching scenario is when the subsequent mode switch to $i_{k+1}$ necessitates prolonging the duration in current mode $i_k$ to effectively suppress the state explosion in the sense of Definition \ref{def:underControl}. We refer to this type of scenario as a malignant switch, as it hinders a quick transition. To provide a precise definition of such a switch, we introduce the following rigorous definition.

\begin{definition}{(Malignant and Benign Switch)} 
 \label{def:malignant} 
 When system at time $\mathfrak{T}_{i_{k+1}}$ switches from mode $i_k$ to $i_{k+1}$, i.e.,  $\sigma(\mathfrak{T}_{i_{k+1}}^-)=i_k$ and  $\sigma(\mathfrak{T}_{i_{k+1}}^+)=i_{k+1}$, we call the switch malignant if
 \begin{align*}
 \ln \rho_*^{i_ki_{k+1}}+\ln \mathcal{X}_*^{i_ki_{k+1}}>\ln \bar{\alpha}
 \end{align*}
 benign otherwise.
  \end{definition}

In equation (\ref{eq:tau_a_known}), assigning a value of one to $\mathcal{T}_*^{ij}$ indicates the occurrence of benign switching.

Now, our focus is to assess the performance loss resulting from the approximation of Program 1 by Program 3. The following corollary provides an upper bound on the performance loss.

  \begin{proposition} \label{coro: compare}
For any arbitrary switch sequence $\mathcal{I}$ we have the following properties: 

\begin{enumerate}

    \item The solution of Program 3, $C_3(\mathcal{I})$ has the following lower bound:

\begin{align}
   &C_3(\mathcal{I})\leq \max_{i,j\in \mathcal{M}}\big[-ns \tau^{*}_{ij} J_*(\Theta_*^i, Q^i, R^i)\big]:=\mathcal{U}_3
\end{align}

\item The optimality gap between Problem 1 and 3 $C_3(\mathcal{I})-C_1(\mathcal{I})$ has the following upper-bound
\begin{align}
   C_3(\mathcal{I})-C_1(\mathcal{I})\leq \mathcal{U}_3-\mathcal{L}_1.
\end{align}
\end{enumerate}
\end{proposition}
% where

% \begin{align*}
%    \nonumber \tilde{\mathcal{X}}_*&=\max_{i,j\in \mathcal{M}} \frac{\overline{\lambda}\big(P({\Theta}^i_*, Q^i, R^i)\big)}{\underline{\lambda}\big(P({\Theta}^j_*, Q^j, R^j)} \\
%    \nonumber \tilde{\rho}_*&=\max_{i,j\in \mathcal{M}} \frac{\overline{\lambda}\big(P({\Theta}^j_*, Q^j, R^j)\big)}{\underline{\lambda}\big(P({\Theta}^i_*, Q^i, R^i)} \\
%    \tilde{\eta }_* &=\min_{i\in \mathcal{M}}\frac{\underline{\lambda}\big(H(\Theta_*^i,Q^i,R^i)\big)}{\overline{\lambda}\big(P({\Theta}^i_*, Q^i, R^i)\big)}
% \end{align*}
Conclusively, under the illustrated information disclosure mechanism and when all mode parameters are known, the justifiable strategy when feasibility and computational complexity are crucial factors, should be applying the policy obtained by solving DARE for each mode, and constraining the duration between two subsequent switches to be minimum dwell-time computed by Program 3, called minimum mode-dependent dwell time.

Now, before delving into the problem in the $\bar{M}\bar{S}$ setup, it is crucial to review some foundational concepts.

 \section{Preliminaries for the $\bar{M}\bar{S}$ Setup}
\label{sec:prelim}
We review the notation of strong stability introduced in \cite{cohen2019learning}, which will be used in the context of switched system as well.

\subsection{$(\kappa,\gamma)-$ strong stability \cite{cohen2019learning}}

\begin{definition} 
\label{Def:StrongStability}
Consider a linear plant parameterized by $A$ and $B$. The closed-loop system matrix $A+BK$ is $(\kappa, \gamma)-$ strongly stable for $\kappa>0$ and $0<\gamma<1$ if there exists $H\succ 0$ and $L$ such that $A_*+B_*K=HL{H}^{-1}$ and
\begin{enumerate}
    \item $\|L\|\leq 1-\gamma$
     \item $\|H\| \|{H}^{-1}\|\leq \kappa$ .
\end{enumerate}
Furthermore, we say a sequence $K$ of control gains is $(\kappa, \gamma)-$strongly stabilizing for a plant $(A, B)$ if $A+BK$ is $(\kappa, \gamma)-$strongly stable.
\end{definition}
As proved in \cite{cohen2018online} any stabilizing policy $K$ is indeed $(\kappa, \gamma)$-stabilizing for some $\kappa$ and $\gamma$ and the definition thus does not introduce any additional assumption. However, it simplifies the analysis. For the sake of completeness we provided the proof in Appendix.

\subsection{System Identification}
\label{sec:ident}

The preliminaries provided by this section are a summary of results provided by \cite{cohen2019learning} and \cite{abbasi2011regret}. However, we slightly modify the representation, especially for construction of confidence ellipsoid.

\begin{assumption}
\label{asum:Side}
$ $
\begin{enumerate}
\item There are known constants $\alpha^i_0, \alpha^i_1,\vartheta^i, \nu^i>0$ for all $i\in \mathcal{M}$ such that, $\alpha^i_0I\leq Q^i\leq \alpha^i_1 I$, $\alpha^i_0I\leq R^i\leq \alpha^i_1 I$, $\|\Theta _{*}^i\|\leq \vartheta^i$ and $\nu^i$ is an upper-bound for average expected cost $J_*(\Theta_*^{i}, Q^{i}, R^{i})$, defined in Section \ref{sec:RegretDef}.
% \begin{align*}
% & \alpha^i_0I\leq Q^i\leq \alpha^i_1 I,\; \; \; \alpha^i_0I\leq R^i\leq \alpha^i_1 I
% \\
% & \|\Theta _{*}^i\|\leq \vartheta^i,\; \; \; J_*^i\leq \nu^i .
% \end{align*}
\item There is an initial stabilizing policy $K^i_0$ for all subsystems.
\end{enumerate}
\end{assumption}

It is important to mention that Assumption \ref{asum:Side} is a standard assumption in the literature and can be found in \cite{cohen2019learning}. The second part of the assumption can be omitted by utilizing the results from \cite{lale2022reinforcement} and the authors' previous work \cite{chekan2022joint}. However, for the sake of consistency, we will adhere to the assumption presented in \cite{cohen2019learning}.

Now, consider the linear switched system (\ref{eq:dyn_atttt}), which operates according to the sequence $\mathcal{I}$ for an arbitrary time $t$. Furthermore, let us denote $n_i(t)$ as the total number of time steps the system has spent in mode $i\in \mathcal{M}$ up to time $t$. Then, we can express (\ref{eq:dyn_atttt}) as follows:

\begin{align}
X_{n_i(t)}&=Z_{n_i(t)} \Theta^i_{*}+W_{n_i(t)} \label{compactdyn}
\end{align}
where $W_{n_i(t)}$ is the vertical concatenation of $\omega_{1}^\top,...,\omega_{n_i(t)}^\top$ and $X_{n_i(t)}$ and $Z_{n_i(t)}$ are matrices constructed by rows $x^\top_{1}, ...,x^\top_{n_i(t)}$ and ${z}^\top_{0}, ...,{z}^\top_{n_i(t)-1}$ respectively. 

Using the measured data, the $l^{2}$-regularized least square estimate can be obtained as
 \begin{align}
 \hat{\Theta}^i_{n_i(t)} &=\operatorname*{argmin}_{\Theta^i} e(\Theta^i)={V^i}_{n_i(t)}^{-1}\big(Z_{n_i(t)}^\top X_{n_i(t)}+ \lambda_i \Theta^i_0\big). \label{eq:LSE_Sol12} 
 \end{align}
where $e(\Theta^i)$ is defined by
 \begin{align}
\nonumber e(\Theta^i)&=\lambda_i \operatorname{Tr} \big((\Theta^i-\Theta^i_0)^\top(\Theta^i-\Theta^i_0)\big)+\\
&\sum _{s=0}^{n_i(t)-1} \operatorname{Tr} \big((x_{s+1}-{\Theta^i}^\top z_{s})(x_{s+1}-{\Theta^i}^\top z_{s})^\top)\big) \label{eq:LSE_pr}
 \end{align}
 where $\Theta^i_0$ is an initial estimate and $\lambda_i$ is a regularization parameter. Furthermore,
 \begin{align*}
V^i_{n_i(t)}=\lambda_i I + \sum_{s=0}^{n_i(t)-1} z_{s}z_{s}^\top=\lambda_i I +Z_{n_i(t)}^\top Z_{n_i(t)},
 \end{align*}
is covariance matrix. By assuming martingale difference properties for dynamics and sub-Gaussianity for the process noise, and assuming having access to an initial estimate $\Theta^i_0$ such that $\|\Theta^i_*-\Theta^i_0\|_*\leq \epsilon_i$ for some $\epsilon_i >0$ a high probability confidence set around true but unknown parameters of system is constructed as follows
 \begin{align}
\nonumber \mathcal{C}^i_t(\delta):=\big\{&\Theta^i \in \mathbb{R}^{(n_i+m_i)\times n_i}|\\
&\operatorname{Tr}\big((\Theta^i-\hat{\Theta}^i_{n_i(t)})^\top V^i_{n_i(t)}(\Theta^i-\hat{\Theta}^i_{n_i(t)})\big) \leq r^i_{n_i(t)}\big\} \label{eq:confSet1_tighterghfff}
\end{align}
where 
\begin{align}
    r^i_{n_i(t)}=\bigg( \sigma_{\omega} \sqrt{2n \log\frac{n\det(V^i_{n_i(s)}) }{\delta \det(\lambda_i I)}}+\sqrt{\lambda_i}  \epsilon_i \bigg)^2. \label{radius_centralEl_realTime}
\end{align}
It is guaranteed that the true parameter of the system $\Theta^i_*$ belongs to the confidence ellipsoid $\mathcal{C}^i_t(\delta)$ with probability at least $1-\delta$ where $0<\delta<1$. The regularization parameters $\lambda_i$'s are a user-defined parameter which needs to be specified in a way that the stability of system is guaranteed and the regret scales appropriately. We will specify it in the stability analysis section. It is noteworthy that $\mathcal{C}^i_t(\delta)$, constructed based on $n_i(t)$ data points, preserves information about the parameters of mode $i$ up to time $t$.

\subsection{Primal and Dual Relaxed-SDP Formulation}
While referring the readers to Appendix \ref{sec:primaldualSDP} for more detail, LQR control problem with known system parameters can be reformulated in Semi-definite Programming (SDP) in the form of primal and dual problems. Given the constructed confidence set (\ref{eq:confSet1_tighterghfff})-(\ref{radius_centralEl_realTime}), by applying the perturbation lemma (see Appendix \ref{perturbationLemma_Stab}), the primal and dual formulations are relaxed to account for the estimation error.

The relaxed primal SDP is given as follows:

\begin{align}
\begin{array}{rrclcl}
\displaystyle  \min & \multicolumn{1}{l}{\begin{pmatrix}
	Q^i& 0 \\
	0 & R^i
	\end{pmatrix}\bullet \Sigma}\\
\textrm{s.t.} & \Sigma_{xx}\geq {\hat{\Theta}_{n_i(t)}^{i^{\top}}} \Sigma\hat{\Theta}^i_{n_i(t)}+W-\mu^i_{n_i(t)}\big(\Sigma\bullet {V^i}^{-1}_{n_i(t)}\big)I,\\
&\Sigma\succ 0. 
\end{array}\label{eq:RelaxedSDP}
\end{align}

where minimization is with respect to 
\begin{align*}
    \Sigma=\begin{pmatrix}
\Sigma_{xx} & \Sigma_{xu} \\
\Sigma_{ux} & \Sigma_{uu}
\end{pmatrix}
\end{align*}
and $\mu_{n_i(t)}^i\geq r^i_{n_i(t)}+\sqrt{r^i_{n_i(t)}}\vartheta^i \|V^i_{n_i(t)}\|^{1/2}$. 
We denote the optimal solution of this program as the operator $\Sigma(\mathcal{C}_t^i, Q^i, R^i)$, which operates on $Q^i$, $R^i$, as well as the confidence ellipsoid $\mathcal{C}_t^i(\delta)$ that fully determines $\hat{\Theta}^i_{n_i(t)}$, $V^i_{n_i(t)}$, and $\mu^{n_i(t)}_i$. For brevity, we choose to use $\mathcal{C}_t^i$ rather than $\mathcal{C}_t^i (\delta)$ within the operator expression.

The controller extracted from solving the relaxed SDP (\ref{eq:RelaxedSDP}) is deterministic and linear in state ($u=K(\mathcal{C}_t^i, Q^i, R^i)x$) where 
\begin{align}
K(\mathcal{C}_t^i, Q^i, R^i)=\Sigma_{ux}(\mathcal{C}_t^i, Q^i, R^i)\Sigma_{xx}^{-1}(\mathcal{C}_t^i, Q^i, R^i). \label{Feedback}
\end{align} 

The designed control possesses strong stabilizing characteristics in accordance with the definition provided in Definition \ref{Def:StrongStability}. 

The relaxed primal problem (\ref{eq:RelaxedSDP}) is mainly used for control design purpose. However for the stability analysis, minimum mode-dependent dwell time design and regret bound analyses, we need its dual program which is given as follows:
\begin{align}
\begin{array}{rrclcl}
 \max &\multicolumn{1}{l}{P\bullet W}\\
\textrm{s.t.} & \begin{pmatrix}
Q^i-P& 0 \\
0 & R^i
\end{pmatrix}+\\
& \hat{\Theta}^i_{n_i(t)} P\; {\hat{\Theta}^{i^\top}_{n_i(t)}}\succeq \mu^i_{n_i(t)} \|P\|_*{{V}^{-1}_{n_i(t)}}\\
&P \succeq 0 
\end{array}\label{eq:RedSDP_DUAL}
\end{align}

In the context of optimization with respect to $P$, we denote the optimal solution of the relaxed dual problem (\ref{eq:RedSDP_DUAL}) as $P(\mathcal{C}^i_t, Q^i, R^i)$.

The derivation of the relaxed primal and dual formulations follows the analysis provided by \cite{cohen2019learning}. However, the fundamental distinction between our formulation and that presented in \cite{cohen2019learning} lies in how we define $\mu_{n_i(t)}^i$. This discrepancy arises from our choice not to normalize the confidence set. 
  
\subsection{Regret Definition}
\label{sec:RegretDef}

Let $\mathcal{A}_*$ represent the algorithm that possesses knowledge of the mode parameters. This algorithm devises a feedback policy by solving the DARE, calculates the mode-dependent dwell time associated with the DARE solution through the resolution of Problem 3, and follows the sequence $\mathcal{I}=\{i_0, i_1,..., i_{n}\}$. Under these conditions, the expression for the optimal accumulated cost is given by:

\begin{align}
	J_{\mathcal{A}_*}(\mathcal{I})= \sum_{k=0}^{n-1}\tau^*_{i_ki_{k+1}} J_*(\Theta_*^{i_k}, Q^{i_k}, R^{i_k})
 \label{eq:optcost} 
\end{align}

where $\tau^*_{i_ki_{k+1}}$ is given by (\ref{eq:tau_a_known}) and $J_*(\Theta_*^{i}, Q^{i}, R^{i})=P_*(\Theta_*^{i}, Q^{i}, R^{i}) \bullet W$ is optimal average expected cost for the mode $i\in \mathcal{M}$. To emphasize, the operator $P_*(\Theta_*^{i}, Q^{i}, R^{i})$ represents the solution of the dual SDP and should not be confused with $P(\mathcal{C}^i_t, Q^i, R^i)$, which is the solution of the relaxed dual SDP.

Now, let an arbitrary algorithm $\mathcal{A}$ address the $\bar{M}\bar{S}$ problem by following switching with epoch duration $\mathcal{T}$ and applying episodic policies $\mathcal{K}$. Having the corresponding accumulated cost denoted by $J_{\mathcal{A}} (\mathcal{I})$, we define the regret for the algorithm $\mathcal{A}$ as follows:

\begin{align}
	R_{\mathcal{A}}(\mathcal{I}) =J_{\mathcal{A}} (\mathcal{I})-J_{\mathcal{A}_*}(\mathcal{I}).\label{eq:Reg} 
\end{align}

\section {Problem Solution for $\bar{M}\bar{S}$ setup} \label{sec:solution}
\subsection{Overview}

In this section, we introduce our algorithm, which aims to design the feedback gain and the minimum mode-dependent dwell time for the switched system (\ref{eq:dyn_atttt}-\ref{eq:obs}) when there is only access to noisy measurements of the state. The objective is to ensure that the system switches according to an externally revealed step-by-step sequence $\mathcal{I} = \{i_{0},; i_{1},; ...i_{n}\}$ with minimal cost while maintaining control over the expected state norm, as defined in Definition \ref{def:underControl}. we first introduce specific symbols to streamline our results and analysis. We denote the estimated minimum mode-dependent dwell time of the $k^{th}$ switch between $i_k$ and $i_{k+1}$ by $\tau_{es}^{i_k i_{k+1}}$, where the $0^{th}$ switch starts at time $t=0$ with mode $i_0$. This symbol provides information about the index of the switch and both consecutive modes. The estimated time of occurrence of the $k^{th}$ switch, denoted as $\mathfrak{T}^{es}_{i_k}$, is defined as $\mathfrak{T}^{es}_{i_{k}}=\tau_{es}^{i_{k-1} i_{k}} + \mathfrak{T}^{es}_{i_{k-1}}$, with $\mathfrak{T}^{es}_{i_{0}}=0$. It is directly deduced that $\mathfrak{T}^{es}_{i_{k}}=\sum_{q=0}^{k}\tau_{es}^{i_{q-1}j_{q}}$. We apply an OFU-based principle through which we build a high probability confidence ellipsoid for systems parameters, using which we design our feedback policy through solving relaxed primal problem (\ref{eq:RelaxedSDP}) and specify the associated minimum mode-dependent dwell time using solution of relaxed dual SDP (\ref{eq:RedSDP_DUAL}). The way we apply OFU principle is slightly different than of OSLO algorithm proposed by \cite{cohen2018online} in that we apply an off-policy version of it where we exploit the information learned by $\mathfrak{T}^{es}_{i_{k}}$ for an epoch occurring in interval $[\mathfrak{T}^{es}_{i_{k}},\; \mathfrak{T}^{es}_{i_{k+1}})$. The system commits to apply the designed feedback policy throughout the epoch whose length is indeed the designed minimum dwell time that is sufficient condition to guarantee state norm boundedness in the sense of Definition \ref{def:underControl}.

\subsection{Main Steps of Algorithm \ref{alg:OSL3}}

Given initial parameter estimates for all the subsystems denoted by $\Theta_0^i$ such that $|\Theta_0^i-\Theta_*^i|\leq \epsilon_i$ (as defined in Appendix \ref{sec:initializationa}), Algorithm \ref{alg:OSL3} designs minimum-dwell time and feedback gain to be implemented. For such initialization, we propose a provably efficient strategy, Algorithm \ref{alg:IExp} adopted from \cite{cohen2018online} with slight modifications on the duration of the warm-up phase. We leave the supporting arguments and theorem to Appendix \ref{sec:initializationa} for brevity. 

The process for building confidence set is performed by following the process specified in Section \ref{sec:ident}. For $k^{th}$ epoch starting at time $\mathfrak{T}^{es}_{i_{k}}$ we choose $\lambda_{i_k}$ such that
\begin{align}
\lambda_{i_k}\geq \frac{4\bar{\mu}_{i_k}\nu_{i_k}}{\alpha_0^{i_k}\sigma_{\omega}^2}. \label{eq:cond1}
\end{align}
where
\begin{align}
&{\mu}^{i_k}_{n_{i_k}(\mathfrak{T}^{es}_{i_{k}})}=r^{i_k}_{n_{i_k}(\mathfrak{T}^{es}_{i_{k}})}+ \label{eq:mubarAlg}\\
\nonumber &\sqrt{r^{i_k}_{n_{i_k}(\mathfrak{T}^{es}_{i_{k}})}}\vartheta_{i_k} (\lambda_{i_k}+\|\sum_{q=1}^{n_{i_k}(\mathfrak{T}^{es}_{i_{k}})} z^{i_k}_q {z^{i_k}_q}^\top\|)^{0.5}. 
\end{align}

The algorithm employed is of the off-policy type, utilizing information learned by the beginning of each epoch to compute feedback gain and the minimum mode-dependent dwell time to be applied during the epoch.

To clarify, let $\Pi_{t}$ represent the set of all confidence ellipsoids updated by time $t$, defined as follows:

\begin{align}
    \Pi_{t}=\{\mathcal{C}_{t}^j(\delta)|\; j=1,...,|\mathcal{M}|\} \label{eq:confsetssets}
\end{align}

Now, consider the moment $\mathfrak{T}^{es}{i{k}}$ when the $k^{th}$ switch from mode $i_k$ to $i_{k+1}$ occurs. During this epoch, the fixed applied control is designed based on the most recently updated corresponding confidence ellipsoid $\mathcal{C}^{i_k}_{\mathfrak{T}^{es}_{i_{k}}}(\delta)\in \Pi_{{\mathfrak{T}^{es}_{i_{k}}}}$ and through solving the relaxed primal problem (\ref{eq:RelaxedSDP}). Furthermore, given the revealed next mode $i_{k+1}$, the minimum mode-dependent dwell time $\tau_{es}^{i_k i_{k+1}}$ is computed using $\mathcal{C}^{i_k}_{\mathfrak{T}^{es}_{i_{k}}}(\delta)$ and $\mathcal{C}^{i_{k+1}}_{\mathfrak{T}^{es}_{i_{k+1}}}(\delta)$. Theorem \ref{thm:minimum_average_dwell_revised} illustrates how to compute this quantity.
By actuating with the designed fixed policy, the algorithm maintains exploration throughout. When the epoch concludes at $\mathfrak{T}^{es}_{i_{k+1}}$, the set of confidence ellipsoids is updated with $\mathcal{C}^{i_k}_{\mathfrak{T}^{es}_{i_{k+1}}}(\delta)$. Consequently, the set of recently updated confidence ellipsoids is refreshed as $\Pi_{\mathfrak{T}^{es}_{i_{k+1}}}$, where $\mathcal{C}^j_{\mathfrak{T}^{es}_{i_{k+1}}}(\delta)=\mathcal{C}^j_{\mathfrak{T}^{es}_{i_{k}}}(\delta)$ for all $j\neq i$.

% Having regularization parameter $\lambda_i$ and the measurements in hand, the system identification step follows the same steps of Section \ref{sec:ident}. The high probability confidence set for the mode $i$ is given as follows:
%  \begin{align}
% \nonumber&\mathcal{C}^i_{t}(\delta):=\big\{\Theta \in \mathbb{R}^{(n+m)\times n}|\\
% &\operatorname{Tr}\big((\Theta-\hat{\Theta}^i_{t})^\top V^i_{n_i(\tau_q)}(\Theta-\hat{\Theta}^i_{t})\big) \leq r^i_{\tau_q}\big\} \label{eq:confSet1_tighterghfff2}
% \end{align}
% where 
% \begin{align}
%     r^i_{\tau_q}=\bigg(\sigma_{\omega} \sqrt{2n \log\frac{n\det(V^i_{n_i(\tau_q)}) }{\delta \det(\lambda_i I)}}+1 \bigg)^2. \label{radius_centralEl_realTime2}
% \end{align}
% It is worth noting that the quantities $\hat{\Theta}^i_{t}$ and $V^i_{n_i(\tau_q)}$ are obtained using $n_i(\tau_q)$ measurement.  

\begin{algorithm} 
	\caption{Safe and Fast Switching Algorithm (SFSA)} \label{alg:OSL3}
	\begin{algorithmic}[1]
		\STATE \textbf{Inputs:} $\bar{\alpha}\in (0,\; 1)$, $\alpha^i_0,$$\,\sigma_{\omega}^2,$ $\,\vartheta^i, $$\,\nu^i>0,$ $\,\delta\in(0,\;1),$ $\,Q^i\,, R^i,\, \Theta^i_{0}$ $\forall i=1,..., |\mathcal{M}|$ 
		% \STATE Initialize $V^i_0=\lambda_i I$ for all $i=1,..., \mathcal{M}$
        \STATE Set $\mathfrak{T}^{es}_{i_{0}}=0$
		\STATE {Initialize the set $\Pi_{\mathfrak{T}^{es}_{i_{0}}}$ using the initial estimates $\Theta^i_{0}$} 
       \STATE Initialize first mode $i_0=i$
       \FOR {$k=0,1,...$}
         \STATE Receive the next mode $i_{k+1}$
         \STATE Use $\Pi_{\mathfrak{T}^{es}_{i_{k}}}$ and solve relaxed primal SDP (\ref{eq:RelaxedSDP}) for $\Sigma (\mathcal{C}_{\mathfrak{T}^{es}_{i_{k}}}^{i_k},Q^{i_k},R^{i_k})$ and compute feedback $K(\mathcal{C}_{\mathfrak{T}^{es}_{i_{k}}}^{i_k},Q^{i_k},R^{i_k})$ by (\ref{Feedback})
         \STATE Use $\Pi_{\mathfrak{T}^{es}_{i_{k}}}$ and solve dual Relaxed-SDP  (\ref{eq:RedSDP_DUAL}) for both $P(\mathcal{C}_{\mathfrak{T}^{es}_{i_{k}}}^{i_k},Q^{i_k},R^{i_k})$ and $P(\mathcal{C}_{\mathfrak{T}^{es}_{i_{k}}}^{i_{k+1}},Q^{i_{k+1}},R^{i_{k+1}})$
         \STATE Compute $\tau_{es}^{i_k i_{k+1}}$ by Theorem \ref{thm:minimum_average_dwell_revised} \STATE Set $\mathfrak{T}^{es}_{i_{k+1}}=\mathfrak{T}^{es}_{i_{k}}+\tau_{es}^{i_k i_{k+1}} $
         \FOR {$\mathfrak{T}^{es}_{i_{k}}\leq t <\mathfrak{T}^{es}_{i_{k+1}}$}
         \STATE Actuate in the mode $i_k$ by executing control $u_t^{i_k}=K(\mathcal{C}_{\mathfrak{T}^{es}_{i_{k}}}^{i_k},Q^{i_k},R^{i_k})x_t$
		\STATE Observe new state $x_{t+1}$, Save $(z^{i_k}_t,x_{t+1})$ into database of mode $i_k$ 
	  \ENDFOR
      \STATE update confidence ellipsoid of the mode $i_k$, $\mathcal{C}_{\mathfrak{T}^{es}_{i_{k+1}}}^{i_k}$ by following (\ref{compactdyn}- \ref{radius_centralEl_realTime}) using regularization parameter obtained by (\ref{eq:cond1}-\ref{eq:mubarAlg})
      \STATE update $\Pi_{\mathfrak{T}^{es}_{i_{k}}}$ and name it $\Pi_{\mathfrak{T}^{es}_{i_{k+1}}}$
      \ENDFOR
	\end{algorithmic}
\end{algorithm}

\section{Analysis and Guarantees} 
\label{sec:guarantee}

\subsection{Stability Analysis} 
\label{sec:stabReg}
In this section, we first demonstrate that the learning-based policy designed is $(\kappa, \gamma)$-strongly stabilizing within an epoch. The result we present are similar to those of \cite{cohen2019learning}, with the rationale detailed in its proof (see Appendix \ref{sec:initializationa}). The second part of this section introduces Theorem \ref{thm:minimum_average_dwell_revised}, which provides a formula for online computation of the minimum mode-dependent dwell-time. This computation is used in Algorithm \ref{alg:OSL3} and ensures control of state norm growth in the sense of Definition 1.

\subsubsection{Stability in an Epoch}

The following Theorem gives the stability properties of the designed OFU-based control applied on a mode $i$ within an epoch (i.e., interval between two subsequent switches in $\mathcal{I}$).

\begin{theorem} \label{Stability_thm17}

Suppose that at time $t$, a transition to an arbitrary mode $i$ takes place, and let $\mathcal{C}^i_{t}(\delta)\in \Pi_{t}$ indicate the confidence ellipsoid formed using $n_i(t)$ data points. Then defining $\kappa_i=\sqrt{2\nu_i /\alpha^i_0\sigma_{\omega}^2}$ and $\gamma_i=1/2\kappa_i^2$, it is evident that the policy $K(\mathcal{C}_t^i, Q^i, R^i)$ obtained through solving relaxed primal SDP exhibits $(\kappa_i,\gamma_i)$-strong stabilizing characteristics as per Definition \ref{Def:StrongStability}.
\end{theorem}

% Recalling Definition 1, a sequentially stabilizing policy keeps the states of system bounded in an epoch where the corresponding upper bound is in terms of the stabilizing policy parameters $\kappa_i$ and $\gamma_i$.

\subsubsection{Minimum Mode-Dependent Dwell-Time}

\begin{theorem}\label{thm:minimum_average_dwell_revised}
Consider the switched system described by (\ref{eq:dyn_atttt}), which undergoes a mode transition to mode $i$ at time $t$ and subsequently identifies the next mode index $j$ for the upcoming switch. Let $K(\mathcal{C}_t^{i}, Q^{i}, R^{i})$ denote the control feedback designed for mode $i$, and let $P(\mathcal{C}_t^{i}, Q^{i}, R^{i})$ and $P(\mathcal{C}_t^{j}, Q^{j}, R^{j})$ represent the corresponding solutions of relaxed dual SDPs. Then, the growth in state norms during the transition from mode $i$ to $j$ is under control in the sense of Definition \ref{def:underControl}, provided that the duration of actuation in mode $i$ lasts for at least the minimum dwell-time of $\tau^{ij}_{es}$ given by:

     \begin{align}
   \nonumber  &\tau^{ij}_{es}:=\max\\
    &\bigg\{1, -\frac{\ln \rho (\mathcal{C}^i_{t}(\delta), \mathcal{C}^j_{t}(\delta))+\ln \mathcal{X} (\mathcal{C}^i_{t}(\delta), \mathcal{C}^j_{t}(\delta))-\ln \bar{\alpha}}{\ln \big(1-\eta \big(\mathcal{C}^i_{t}(\delta)\big)\big)}\bigg\}\label{eq:tau_a_fast} 
     \end{align}
     where
     \begin{align}
	\eta \big(\mathcal{C}^i_{t}(\delta)\big):= \frac{\underline{\lambda}\big(\mathcal{H}(\mathcal{C}^i_{t}(\delta)\big)}{\overline{\lambda}\big(P(\mathcal{C}^i_{t}, Q^i, R^i)\big)}.\label{eq:EtaDef2}
	\end{align}
\begin{align}
\nonumber  \rho (\mathcal{C}^i_{t}(\delta), \mathcal{C}^j_{t}(\delta)):=\frac{\overline{\lambda}\big(P(\mathcal{C}^j_{t}, Q^j, R^j)\big)}{\underline{\lambda}\big(P(\mathcal{C}^i_{t}, Q^i, R^i)\big)},\\
\mathcal{X} (\mathcal{C}^i_{t}(\delta), \mathcal{C}^j_{t}(\delta)):=\frac{\overline{\lambda}\big(P(\mathcal{C}^i_{t}, Q^i, R^i)\big)}{\underline{\lambda}\big(P(\mathcal{C}^j_{t}, Q^j, R^j)\big)}
\label{eq:rhoDef}
\end{align} 
and
 	\begin{align}
	\nonumber \mathcal{H}(&\mathcal{C}^i_{t}(\delta))=Q^i+K^\top (\mathcal{C}^i_{t}, Q^i, R^i)R^i K(\mathcal{C}^i_{t}, Q^i, R^i)\\
 \nonumber &-2\mu^{i}_{n_i(t)}\|P(\mathcal{C}^i_{t}, Q^i, R^i)\|_*\times\\
&\begin{pmatrix} I \\ K(\mathcal{C}^i_{t}, Q^i, R^i) \end{pmatrix}{V^i}^{-1}_{n_i(t)}\begin{pmatrix} I \\ K(\mathcal{C}^i_{t}, Q^i, R^i) \end{pmatrix}^\top.
 \label{eq:HDef3p}
	\end{align}
%  \begin{align}
%    \mathcal{G}(\mathcal{C}_{\tau_q}^j):=P(\hat{\Theta}^j_{\tau_{q}}, Q^j, R^j)+\chi_{\tau_q}^j\label{eq:GDef}
% \end{align}
% and
    %  \begin{align}
    %   \nonumber  \chi_{\tau_q}^i&:=\frac{2\kappa_i^2\mu_i^{\tau_q}}{\gamma}\|P(\hat{\Theta}^i_{\tau_q}, Q^i, R^i)\|_* \times\\
    %    &\bigg\|\begin{pmatrix} I \\ K(\hat{\Theta}^i_{\tau_q}, Q^i, R^i) \end{pmatrix}{V^i}^{-1}_{n_i(\tau_q)}\begin{pmatrix} I \\ K(\hat{\Theta}^i_{\tau_q}, Q^i, R^i) \end{pmatrix}^\top\bigg\|I \label{eq:chiddef_th}
    % \end{align}
\end{theorem}
Note that in the proof of Theorem \ref{thm:minimum_average_dwell_revised} it is shown that $0< \eta \big(\mathcal{C}^i_{t}(\delta)\big)< 1$. 

% Furthermore, when the mode $i$ is suppressing comparing to the mode $j$ at epoch starting at time $\tau_q$ (See Definition), then $\rho (\mathcal{C}^i_{\tau_q}(\delta), \mathcal{C}^j_{\tau_{q-1}}(\delta))<0$, meaning that there is no need to stay on the mode $i$ and the subsequent switch can happen quickly.

The following lemma gives an upper-bound for the state norm for the switched system.

\begin{lemma}\label{thm:res_sequential_stability_resp}
Algorithm \ref{alg:OSL3} guarantees
\begin{align}
&\mathbb{E}[x^\top_{\mathfrak{T}^{{es}^+}_{i_{k+1}}}x_{\mathfrak{T}^{{es}^+}_{i_{k+1}}}| \mathcal{F}_{\tau_{q}-1}]\leq \mathbb{E}[x^\top_{\mathfrak{T}^{{es}^+}_{i_{k}}}x_{\mathfrak{T}^{{es}^+}_{i_{k}}}|\mathcal{F}_{\tau_{q}-1} ]+\kappa_*^4 \frac{\alpha_1^*}{\alpha_0^*} \sigma_{\omega}^2 \label{eq:moghayese}
\end{align}
 for the state norm right after two subsequent switches occurring at $\mathfrak{T}^{{es}}_{i_{k}}$ and $\mathfrak{T}^{{es}}_{i_{k+1}}$.

 Furthermore, for any $\mathfrak{T}^{{es}}_{i_{k}}\leq t<\mathfrak{T}^{{es}}_{i_{k+1}}$
\begin{align}
	\mathbb{E}[x_t^\top x_t| \mathcal{F}_{t-1}]\leq \bar{\alpha}^k\kappa_*^2 x_0^\top x_0+ \frac{2-\bar{\alpha}}{1-\bar{\alpha}} \frac{\mathcal{X}_*^2}{\eta_*}\sigma_{\omega}^2:=\bar{X}_k \label{eq:overalbound}
\end{align}
 or overlay for any $t>0$
\begin{align}
	\mathbb{E}[x_t^\top x_t| \mathcal{F}_{t-1}]\leq \kappa_*^2 x_0^\top x_0+ \frac{2-\bar{\alpha}}{1-\bar{\alpha}} \frac{\mathcal{X}_*^2}{\eta_*}\sigma_{\omega}^2:=\tilde{X} \label{eq:overalbound2}
\end{align}
 where $\eta_*=1/\kappa_*^2$ and $\mathcal{X}_*=\kappa_*^2 {\alpha_1^*}/{\alpha_0^*}$.
\end{lemma}

\begin{corollary}
The policy employed in Algorithm 1 falls within the category of candidate policies as defined by (\ref{eq:policyClass}), i.e., $\kappa_i\leq \kappa^i_{c}$ for all $i\in |\mathcal{M}|$. Now, noting that $\gamma^i_{c}<1$, consequently, upon comparing (\ref{eq:moghayese}) from Lemma \ref{thm:res_sequential_stability_resp} with (\ref{eq:stategrowthpp}), it becomes evident that Algorithm 1 ensures the containment of state norm within bounds, aligning with the conditions specified in Definition \ref{def:underControl}. 
\end{corollary}

The following theorem establishes an upper bound on the estimation error of the dwell time. This bound is highly useful for regret bound analysis of the algorithm. 

In the $M\bar{S}$ setting, the problem of minimizing costs while switching according to $\mathcal{I}$ can be framed as minimizing the expected time, as the objective is to spend as little time as possible in each mode while ensuring the state remains bounded, as defined in Definition 1. In this context, the expected cumulative time required to achieve the goal is expressed as:
\begin{align}
T_*=\sum_{k=0}^{n-1}\tau_*^{i_ki_{k+1}}.
\end{align}
where $\tau_*^{i_ki_{k+1}}$'s  are determined by (\ref{eq:tau_a_known}). 

In the $\bar{M}\bar{S}$ setting, the duration, denoted as
\begin{align}
T_{es}=\sum_{k=0}^{n-1}\tau_{es}^{i_ki_{k+1}},
\end{align}
is not explicitly known. Theorem \ref{Thm:dwellTimeError} also offers insights into the order of the total expected time for switching according to $\mathcal{I}$ under the mode index announcement strategy, particularly when ensuring the boundedness of the state is a concern.

\begin{theorem} (Dwell-Time Estimation Error Upper-Bound)
\label{Thm:dwellTimeError}
Suppose that at time $t$, the system initiates actuation in mode $i$ and immediately receives to mode $j$. In this scenario, the following statements hold.

$(a)$ The error of the estimated minimum mode-dependant dwell time $\tau_{es}^{ij}$  has the following upper-bound
\begin{align}
    \tau_{es}^{ij}-\tau_*^{ij}\leq \frac{\ln \bigg(1+ \bar{\beta}_0\overline{\lambda}\big(\chi^i_{t}\big)\bigg)}{-\ln \big(1-\kappa_i^{-2}\big)} 
\end{align}
where 
\begin{align}
  \nonumber \chi_{t}^i:=&\frac{2\kappa_i^2\mu^i_{t}}{\gamma}\|P(\mathcal{C}_t^i, Q^i, R^i)\|_* \bigg\|\begin{pmatrix} I \\ K(\mathcal{C}_t^i, Q^i, R^i) \end{pmatrix}\times\\
  &{V_{n_i(t)}^{i^{-1}}}\begin{pmatrix} I \\ K(\mathcal{C}_t^i, Q^i, R^i) \end{pmatrix}^\top\bigg\|I ,\label{eq:chiddef3}
\end{align}

\begin{align}
    \alpha_0^i\geq \bar{\alpha}^i_0:=\max_{j}\frac{\sqrt{2\nu_i \nu_j}}{\sigma_{\omega}^2},\;\;  \bar{\beta}^{ij}_0=\frac{2\nu_j}{\sigma_{\omega}^2 {\bar{\alpha}^{i^{2}}_0}}.
\end{align}
$(b)$ The expected cumulative time $T_{es}$ required to accomplish switching according to $\mathcal{I}$ with state norm growth under control, as defined in Definition 1, is of $\mathcal{O}(|\mathcal{M}|\sqrt{n})$, where $n$ is the total number of switches.
\end{theorem}

\subsection{Regret Bound Analysis} \label{sec:regban}
% In this subsection, we carry out the regret bound analysis for the proposed Algorithm \ref{alg:OSL3}. We leave the regret bound analysis of warm-up phase (Algorithm \ref{alg:IExp}) a side, as by adapting the analysis provided by \cite{cohen2019learning}, it is straight forward to see that for each mode the algorithm has regret bound of $\sum_{i=1}^{|\mathcal{M}|}\mathcal{O}(\sqrt{T_0^i})$ where $\mathcal{O}$ absorbs system ambient dimension dependent and logarithmic $T_0^i$ terms.

The regret, defined in (\ref{eq:Reg}), can be partitioned into two components, denoted as $R_{\mathcal{A}}^1$ and $R_{\mathcal{A}}^2$, which are outlined as follows:
\begin{align}
	&\nonumber R_{\mathcal{A}}(\mathcal{I}) =\overbrace{\mathbb {E} [\sum _{k=0}^{n-1}\sum _{t=\mathfrak{T}^{{es}}_{i_{k}}}^{\mathfrak{T}^{{*}}_{i_{k+1}}-1} (c_{t}^{\sigma(t)}-J_*(\Theta_*^{i_k}, Q^{i_k}, R^{i_k}))]}^{R_{\mathcal{A}}^1}\\
 &+\underbrace{\mathbb {E} [\sum _{k=0}^{ns-1}\sum _{t=\mathfrak{T}^{{*}}_{i_{k+1}}}^{\mathfrak{T}^{{es}}_{i_{k+1}}-1} c_{t}^{\sigma(t)}]}_{R_{\mathcal{A}}^2} \label{eq:RegretSwitchDec}
\end{align}
Here, $\sigma(t)=i_k$ for $\mathfrak{T}^{{es}}_{i_{k}}\leq t<\mathfrak{T}^{{es}}_{i_{k+1}}$.

The component $R_{\mathcal{A}}^1$ characterizes the regret arising from the sub-optimality of the feedback policy generated by the algorithm. In contrast, the term $R_{\mathcal{A}}^2$ gauges the regret attributed to estimation error of the minimum mode-dependent dwell time. To establish an upper bound for $R^1_{\mathcal{A}}$, we adopt a decomposition approach similar to that found in \cite{cohen2019learning}, while also incorporating considerations for controlling state norm growth. For the latter term, we rely on the insights derived from Theorem \ref{Thm:dwellTimeError}.

To establish an upper bound for the regret, we introduce the event $\mathcal{E}_{\mathfrak{T}^{{es}}_{i_{k}}}(\mathcal{I})$ which is defined as follows:

\begin{align}
\nonumber \mathcal{E}_{\mathfrak{T}^{{es}}_{i_{k}}}(\mathcal{I})=\big\{&\Pi_{\mathfrak{T}^{{es}}_{i_{k}}}\; \textit{and}\\
&E[x_t^\top x_t|\mathcal{F}_{t-1}]\leq \bar{X}_k\;\; \textit{for} \;\; \mathfrak{T}^{{es}}_{i_{k}}\leq t<\mathfrak{T}^{{es}}_{i_{k+1}}\big\}.
\end{align}

This event contains the collection of confidence ellipsoids updated until time $\mathfrak{T}^{{es}}_{i_{k}}$ using which we designed the feedback gain and the minimum mode-dependent dwell time, subsequently upper-bounding the expected state norm, which is also included in the event. The value of $\bar{X}_k$ is determined according to (\ref{eq:overalbound}) during the execution of the mission following the prescribed sequence $\mathcal{I}$.

\begin{theorem}
\label{thm:RegretBound}
 Suppose Assumptions \ref{Assumption 1} and \ref{asum:Side} holds, then Algorithm \ref{alg:OSL3} achieves the expected regret of $\mathcal{O}(|\mathcal{M}|\sqrt{ns})$ with probability at least $1-\delta$.
\end{theorem}

\section{Conclusion}
In this paper, we propose an algorithm that guarantees safe switching with minimum cost in a scenario where switch sequences are revealed in an online fashion. Our algorithm is based on the Optimism in the Face of Uncertainty (OFU) principle, which integrates system identification, control, and dwell-time design procedures. By utilizing constructed confidence sets for parameter estimates, we develop a strategy to accurately estimate the minimum dwell time. 
% Additionally, we demonstrate that the expected time to complete the task of safe switching according to a finite switch sequence with minimum cost is of the order $\mathcal{O}(|\mathcal{M}|\sqrt{ns})$, where $ns$ represents the total number of switches and $|\mathcal{M}|$ is the number of mode candidates. 
Furthermore, we prove that our proposed algorithm achieves an expected regret of $\mathcal{O}(|\mathcal{M}|\sqrt{ns})$, compared to the case when all parameters are known, where $ns$ represents the total number of switches and $|\mathcal{M}|$ is the number of mode candidates. A possible extension for this work involves a setup in which the decision of the next mode to switch becomes an integral part of the strategy design. For instance, the LQ formulation of the problem discussed in \cite{li2023online} has the potential to jointly design both the timing and selection of the next mode, aiming to identify the most effective actuation mode with minimal regret.

\bibliographystyle{IEEEtran}
\bibliography{ref}

\appendices

% Appendixes, if needed, appear before the acknowledgment.

% \section*{Acknowledgment}

% The preferred spelling of the word ``acknowledgment'' in American English is 
% without an ``e'' after the ``g.'' Use the singular heading even if you have 
% many acknowledgments. Avoid expressions such as ``One of us (S.B.A.) would 
% like to thank $\ldots$ .'' Instead, write ``F. A. Author thanks $\ldots$ .'' In most 
% cases, sponsor and financial support acknowledgments are placed in the 
% unnumbered footnote on the first page, not here.

\section{Further Analysis}
\label{partD}
In this section, we dig further and provide proofs, rigorous analysis of the algorithms, properties of the closed-loop system, and regret bounds.

\subsection{Notes on $(\kappa, \gamma)$- stability}

To get a better sense for $(\kappa, \gamma)$- stability, we provide illustrations of the proof of Lemma B.1 by \cite{cohen2018online}. For any linear system traditionally defined by pair $(A, B)$, the stabilizing policy $K$ is $(\kappa, \gamma)$- stabilizing. The closed loop system is stable if $\rho(A+BK)= 1-\gamma$ where $0<\gamma<1$. If we define matrix $Q=(1-\gamma)^{-1}(A+BK)$, then $Q$ is stable if there exists a positive definite matrix $P$ such that

\begin{align}
Q^\top P Q\preceq P. \label{komak}
\end{align}

It is also trivial that when $Q$ is stable, $(A+BK)$ is also stable.

From (\ref{komak}) one can write

\begin{align}
(A+BK)^\top P(A+BK)\preceq (1-\gamma)^2 P. \label{eq:komak2}
\end{align}

Pre and post multiplying both sides of (\ref{eq:komak2}) by $P^{-\frac{1}{2}}$ and $P^{\frac{1}{2}}$ yields 

\begin{align}
L^\top L\preceq (1-\gamma)^2 . \label{eq:komak3}
\end{align}
where $L=P^{\frac{1}{2}}(A+BK)P^{-\frac{1}{2}}$. Then defining $H=P^{-\frac{1}{2}}$, one can write

\begin{align}
HLH^{-1}=A+BK.
\end{align}

Setting the condition number of $P^{-\frac{1}{2}}$ to $\kappa$, i.e., $\|H\|\|H^{-1}\|\leq \kappa$ and having $\|L\|\leq 1-\gamma$ from  (\ref{eq:komak3}) complete proof.

It is worthy to note that this definition directly results $\rho (A+BK)\leq \|L\|\leq 1-\gamma$.
\begin{lemma} (Lemma 25 of \cite{cohen2019learning}) \label{lem:komaki}
    Let $X$ and $Z$ denote matrices of equal size and let $Y$ denote a $(\kappa, \gamma)$ stable matrix such that $X\preceq Y^\top X Y+Z$, then $X\preceq \frac{\kappa^2}{\gamma}\|Z\| I$.
\end{lemma}

\subsection{Proofs of Section \ref{eq:knownsetting}}

\textbf{Proof of Corollary \ref{cor:mintimecond}}

\begin{proof}

Considering that Program 1 is a minimization problem linear in $\tau_{i_k}$, it attains its minimum at the equality condition of (\ref{eq:conststab}). Utilizing a similar, albeit simpler, approach as shown in the proof of Theorem \ref{thm:minimum_average_dwell_revised}, we can establish that when a system at time $t$ starts actuation in mode $i_k$ using a stabilizing policy $K_{i^k}$ and is subsequently informed of the next mode-to-switch, denoted as $i_{k+1}$ to actuate with another stabilizing policy $K_{i_{k+1}}$, the duration of actuation in mode $i_k$, denoted as $\tau_{i_k}$, must satisfy the following condition:

\begin{align*}
     & \ln \bar{\rho}(K_{i_k}, K_{i_{k+1}})+ {\tau_{i_k}} \ln \bigg(1-\bar{\eta} \big(K_{i_k}\big)\bigg) \\
& +\ln \bar{\mathcal{X}}(K_{i_k}, K_{i_{k+1}})\leq  \ln \mathcal{Y}
\end{align*}
or equivalently 

\begin{align}
     {\tau}_{i_k} \geq -\frac{\ln \bar{\rho}(K_{i_k}, K_{i_{k+1}})+\ln \mathcal{X}(K_{i_k}, K_{i_{k+1}})-\ln \mathcal{Y}}{\ln \big(1-\bar{\eta} \big(K_{i_k})\big)}. \label{eq:cond}
\end{align}
where
 \begin{align}
	&\nonumber \bar{\eta} \big(K_{i_k}) := \frac{\underline{\lambda}\big(H_{K_{i_k}}\big)}{\overline{\lambda}\big(P_{K_{i_k}}\big)}, \; \bar{\rho}(K_{i_k}, K_{i_{k+1}}):=\frac{\overline{\lambda}\big(P_{K_{i_{k+1}}}\big)}{\underline{\lambda}\big(P_{K_{i_k}})}\\
& \bar{\mathcal{X}}(K_{i_k}, K_{i_{k+1}}):= \frac{\overline{\lambda}\big(P_{K_{i_k}}\big)}{\underline{\lambda}\big(P_{K_{i_{k+1}}})}, H_{K_{i_k}}=Q^{i_k}+K_{i_k}^\top R^{i_k} K_{i_k}	
\end{align}  

then again similar to the proof of Theorem we have

 	\begin{align}
\nonumber &\mathbb{E}[\mathcal{V}_{i_{k+1}}((t+\tau_{i_k})^+)]\leq\\
\nonumber  &\bar{\rho}(K_{i_k}, K_{i_{k+1}})\bigg(1-\bar{\eta} \big(K_{i_k})\bigg)^{\tau_{i_k}}\mathbb{E}[\mathcal{V}_{i_k}(t^+)]+ \\
  &\big(\sum_{k=0}^{\tau_{i_k}} (1-\bar{\eta} \big(K_{i_k}))\big) \overline{\lambda} (P_{K_{i_k}})\sigma_{\omega}^2 \label{eq:joonesh]}
	\end{align}

where by definition of $\mathcal{V}_{k+1}((t+\tau_{i_k})^+)=x^\top _{(t+{\tau}_{i_k})^+}P_{K_{i+1}}x_{(t+{\tau}_{i_k})^+}$ and $\mathcal{V}_{k+1}((t+\tau_{i_k})^+)=x^\top _{t^+}P_{K_{i}}x_{t^+}$, and applying Rayleigh-Ritz inequality it yields

\begin{align*}
    &\mathbb{E}[x^\top _{(t+{\tau}_{i_k})^+}x_{(t+{\tau}_{i_k})^+}]\leq\\
    &\bar{\mathcal{X}}(K_{i_k}, K_{i_{k+1}}) \bar{\rho}(K_{i_k}, K_{i_{k+1}})\bigg(1-\bar{\eta} \big(K_{i_k})\bigg)^{\tau_{i_k}}  \mathbb{E}[x^\top _{t^+}x_{\tau_{t^+}}]+ \\
    &\frac{\bar{\mathcal{X}}(K_{i_k}, K_{i_{k+1}})}{\bar{\eta} \big(K_{i_k})} \sigma_{\omega}^2
\end{align*}
 where we applied an upper-bound on the geometric series in (\ref{eq:joonesh]}).

By definition of $\tau_{i_k}$ we have

\begin{align}
    &\bar{\mathcal{X}}(K_{i_k}, K_{i_{k+1}}) \bar{\rho}(K_{i_k}, K_{i_{k+1}})\bigg(1-\bar{\eta} \big(K_{i_k})\bigg)^{\tau_{i_k}} \leq \mathcal{Y}<1
\end{align}

Furthermore by applying Lemma \ref{lem:komaki}, we have

\begin{align*}
    \|P_{K_{i_k}}\|_*\leq \alpha_1 (1+\kappa^2)\frac{\kappa^2}{\gamma},
\end{align*}
 using which one can show

\begin{align}
    \frac{\bar{\mathcal{X}}(K_{i_k}, K_{i_{k+1}})}{\bar{\eta} \big(K_{i_k})} \sigma_{\omega}^2 \leq \frac{\alpha_1^2}{\alpha_0^2}\frac{\kappa^4(1+\kappa^2)^2}{\gamma^2}\sigma_{\omega}^2
\end{align}

in which we applied the fact that $P_{K_{i}}\succeq \alpha_0 I$ which holds for $i=i_k, i_{k+1}$.

Observing that the policy $K_{i_k}\in \mathcal{S}(\Theta_*^{i_k})$, we have $\kappa \leq \kappa_c^*$ and $\gamma\geq \gamma_c^*$. This indicates that we indeed fulfill the boundedness of state immediately after a switch, as defined in accordance with Definition \ref{def:underControl}.

Furthermore, in cases where $\ln \bar{\rho}(K_{i_k}, K_{i_{k+1}})+\ln \mathcal{X}(K_{i_k}, K_{i_{k+1}})<0$, it signifies that prolonging mode $i_k$ is needless. However, our proposed scenario requires a mandatory presence of at least one step in mode $i_k$. Thus, we include this condition within the definition of minimum dwell-time.
    
\end{proof}

% \textbf{Proof of Corollary \ref{col_lowebound}}

% \begin{proof}
%   The proof directly emerges from our stipulation that the system must operate in the mode with the least cost for a duration of $ns$ time steps. This choice provides a conservative yet robust lower bound applicable to any arbitrary switch sequence $\mathcal{I}$.
% \end{proof}

\textbf{Proof of Corollary \ref{coro: compare}}

\begin{proof}
  The evidence is a straightforward outcome of a situation where the system operates in a mode with the highest cost and longest duration. This scenario offers a robust upper-bound for cost that can be used for any given sequence of switches, denoted as $\mathcal{I}$.
\end{proof}
\subsection{Primal-Dual SDP formulation of LQR}
 \label{sec:primaldualSDP}

For any system $\Theta_*$ with quadratic cost specified with $Q$ and $R$, the standard LQR control problem can be reformulated into the Semi-Definite Programming (SDP) form as follows:

\begin{align}
\nonumber\textrm{minimize}\; \; \; \begin{pmatrix}
Q & 0 \\
0 & R
\end{pmatrix}\bullet \Sigma(\Theta_*, Q,R)\\
\nonumber\textrm{S.t.}\; \; \; \Sigma_{xx}(\Theta_*, Q,R)={\Theta_*}^\top\Sigma(\Theta_*, Q,R)\Theta_*+W\\
\Sigma(\Theta_*, Q,R)>0\label{eq:SDPKhali}
\end{align}
where for a given $\Theta_*$, $\Sigma(\Theta_*, Q,R)$ is the covariance matrix of joint distribution $(x,u)$ in steady-state with
\begin{align*}
\Sigma(\Theta_*, Q,R)= \begin{pmatrix}
\Sigma_{xx} & \Sigma_{xu} \\
\Sigma_{ux} & \Sigma_{uu}
\end{pmatrix}
\end{align*}
in which $\Sigma_{xx}(\Theta_*, Q,R) \in  \mathbb{R}^{n \times n}$, $\Sigma_{uu}(\Theta_*, Q,R) \in  \mathbb{R}^{m \times m}$, and $\Sigma_{ux}(\Theta_*, Q,R)=\Sigma_{xu}(\Theta_*, Q,R) \in  \mathbb{R}^{m \times n}$ for $Q\in \mathbb{R}^{n\times n}$ and $R\in \mathbb{R}^{m\times m}$. The optimal value is exactly average expected cost $J_*(\Theta_*, Q,R)$ and for $W>0$ (which guarantees $\Sigma_{xx}>0$) the optimal policy of system $K_*(\Theta_*, Q,R)$ can be obtained from $K_*(\Theta_*, Q,R)=\mathcal{K} (\Sigma(\Theta_*, Q,R))=\Sigma_{ux}(\Theta_*, Q,R)\Sigma^{-1}_{xx}(\Theta_*, Q,R)$. In other words, the matrix
\begin{align*}
\mathcal{E}(K)= \begin{pmatrix}
X & XK^T \\
KX & KXK^T
\end{pmatrix}
\end{align*}
is a feasible solution for the SDP. With stabilizing $K_*(\Theta_*, Q,R)$ the system state converges to the steady state distribution with covariance of $X=\mathbb{E}[xx^T]$.

The dual of this program is written as follows:

\begin{align}
\begin{array}{rrclcl}
\displaystyle \max & \multicolumn{1}{l}{P(\Theta_*, Q, R)\bullet W}\\
\textrm{s.t.} & \begin{pmatrix}
Q-P(\Theta_*, Q, R) & 0 \\
0 & R
\end{pmatrix}+\Theta_* P(\Theta_*, Q, R) {\Theta_*}^\top=0\\
&P(\Theta_*, Q, R)\geq 0 
\end{array}.
\end{align}
using its optimal solution $P_*(\Theta_*, Q, R)$ we define the optimal average expected cost as $J_*(\Theta_*, Q, R)=P_*(\Theta_*, Q, R) \bullet W$.

\subsection{Initialization Algorithm}
\label{sec:initializationa}

Given a $(\kappa_0,\gamma_0)-$stabilizing policy $K_0$, Algorithm \ref{alg:IExp} guarantees reaching to an appropriate $\epsilon_i$ for all $i\in \mathcal{M}$. The only remaining question is the specification of the algorithm input $T_0^i$ for each mode, to address which we first explicate the main steps of the warm-up phase.

\begin{algorithm} 
	\caption{\small  \cite{chekan2022learn} Warm-Up Phase \normalsize} \label{alg:IExp}
		\begin{algorithmic}[1]
		\STATE \textbf{Inputs:} $K^i_{0},$ $T^i_{0},$ $\vartheta^i$, $\delta$, $\sigma_{\omega},\, \eta_t \sim \mathcal{N}(0, 2\sigma_{\omega}^2{\kappa^i_0}^2 I)$, $\forall i$
		
		\STATE set $\lambda_i=\sigma_{\omega}^2\vartheta_i^{-2}$, $V^i_0=\lambda_i I$ $\forall i \in \mathcal{M}$
		\FOR {$i = 1: |\mathcal{M}|$}
		\FOR {$t = 0: T^i_0$}
		\STATE Observe $x_t$
		\STATE Play $u_t=K^i_0 x_t+\eta_t$
		\STATE observe $x_{t+1}$ 
		\STATE using $u_t$ and $x_t$ form $z_t$ and save $(z_t, x_{t+1})$ and update the ellipsoid according to the procedure given (\ref{eq:thetaHatWarmup}-\ref{radius_centralEl_warmUp})
		\ENDFOR
		\ENDFOR
		\STATE \textbf{Output:} $\Theta_0$  the center of constructed ellipsoid
		\end{algorithmic}
	\end{algorithm}

\subsubsection{Main Steps of Algorithm 1}

In Algorithm \ref{alg:IExp}, using data obtained by executing control input $u_t=K^i_0 x_t+\eta_t$ (which consists of linear strongly stabilizing term and an extra exploratory sub-gaussian noise $\eta_t$), a confidence set is constructed around the true parameters of the system by following the steps of Section \ref{sec:ident}. For warm-up phase We set  $\bar{\lambda}_i=\sigma_{\omega}^2{\vartheta^i}^{-1}$ and $\Theta^i_0=0$ in (\ref{eq:LSE_pr}) that results in by-time-$t$ estimates of
\begin{align}
 \hat{\Theta}^i_{t} &={\bar{V}_t^{i^{-1}}}{Z_{t}^\top} X_{t}. \label{eq:thetaHatWarmup}
 \end{align}
where the covariance matrix $\bar{V}^i_t$ is constructed similar to Section \ref{sec:ident}. Recalling $\|\Theta^i_*\|_*\leq \vartheta^i$ by assumption, a high probability confidence set around true but unknown parameters of system is constructed as
 \begin{align}
\nonumber \bar{\mathcal{C}}^i_t(\delta):=\big\{&\Theta^i \in \mathbb{R}^{(n+m_i)\times n}|\\
&\operatorname{Tr}\big((\Theta^i-\hat{\Theta}^i_t)^\top \bar{V}^i_t(\Theta^i-\hat{\Theta}^i_t)\big) \leq \bar{r}^i_t\big\} \label{eq:confSet1_warmup}
\end{align}
where 
\begin{align}
     \bar{r}^i_t=\bigg( \sigma_{\omega} \sqrt{2n \log\frac{n\det(\bar{V}^i_t) }{\delta \det(\bar{\lambda}_iI)}}+\sqrt{\bar{\lambda}_i\vartheta^i } \bigg)^2. \label{radius_centralEl_warmUp}
\end{align}
We need following definition to specify warm-up phase duration.

 \begin{definition} \cite{chekan2022learn} 
 The set $\mathcal{N}_s(\Theta_*)=\{\Theta\in \mathbb{R}^{(n+m)\times n}:\; \|\Theta-\Theta_*\|\leq \epsilon\}$ for some $\epsilon>0$ is $(\kappa,\gamma)-$stabilizing neighborhood of the system (\ref{eq:dyn_atttt}-\ref{eq:obs}),  if for any $\Theta^{\prime}\in \mathcal{N}_s(\Theta_*)$,  $K(\Theta^{\prime}, Q, R)$  is $(\kappa^{\prime},\gamma^{\prime})-$strongly stabilizing on the system (\ref{eq:dyn_atttt}) with $\kappa^{\prime}\leq \kappa$.
 \end{definition}

 Initialization phase, with an given initial policy $K_0^i$ is carried out for all modes to achieve an estimate $\|\Theta_0^i-\Theta_*^i\|\leq \epsilon_i$ such that it fulfills two goals. 
 
 First, it resides in $(\kappa^i_0, \gamma^i_0)-$stabilizing neighborhood. In other words, the confidence ellipsoid $\bar{\mathcal{C}}^i_t(\delta)$ is guaranteed to be a subset of $(\kappa^i_0,\gamma^i_0)-$strong stabilizing neighborhood. This, indeed, guarantees that Algorithm \ref{alg:OSL3} starts with a policy as good as the initial $(\kappa_0^i, \gamma_0^i)$-stabilizing one, $K_0^i$.

 The following theorem is adopted from \cite{chekan2022learn} that gives the required duration $T_0^i$ to fulfill the first goal.

\begin{theorem} \cite{chekan2022learn}
\label{thm:WarmUp_Duration}
For any mode $i$ there are explicit constants $C^i_0$ and $\epsilon^i_0=poly({\alpha^i_0}^{-1},\alpha^i_1,\vartheta_i, \nu_i, n, m)$ such that for a given $(\kappa^i_0, \gamma^i_0)-$ strongly stabilizing input $K^i_0$, the smallest time $T^i_0$ that satisfies
\begin{align}
\nonumber\frac{80}{T^i_0 \sigma^2_{\omega}}\bigg( &\sigma_{\omega} \bigg[(2n (\log \frac{n}{\delta} + \log (1+\\
\nonumber &\frac{300\sigma_{\omega}^2{\kappa^i_0}^4}{{\gamma_0^i}^2}(n+\vartheta^2\kappa^2_0)\log \frac{T^i_0}{\delta}\bigg]^{1/2}+\sqrt{\lambda} \vartheta^i \bigg)^2\\
&\leq \min \bigg\{{\kappa^i_0}^2\frac{ \alpha^i_0 \sigma^2_{\omega}}{C^i_0}, {\epsilon_0^i}^2 \bigg\}:=\tilde{\epsilon}_i^2 \label{firstcond}
\end{align}
guarantees that $K(\hat{\Theta}^i_{T^i_0}, Q^i, R^i)$ applied on the system $\Theta^i_*$ produces $(\kappa^i_0,\gamma^i_0)$ closed-loop system where $\hat{\Theta}^i_{T^i_0}$ is the center of confidence ellipsoid $\mathcal{C}^i_{T^i_0}(\delta)$ constructed by Algorithm \ref{alg:IExp}. 
\end{theorem}

 Before, describing the second goal, we need to find an upper-bound on the duration of the posted problem, given number of switches $n_s$. The following lemma gives an upper bound on duration of an epoch.

  \begin{lemma}\label{thm:minimum_average_dwell} 
 	Maximum duration for an epoch is  
\begin{align}
\tau^{dw}_{max}:= -\frac{\ln \kappa_*}{\ln(1-\gamma_*)}
\label{eq:tau_a_slow} 
\end{align}
 	where with the definition $\kappa_i:=\sqrt{2\nu_i/\sigma_{\omega}^2\alpha_0^i}$,
 	\begin{align}
 	\kappa^*=\max_{i\in \{1,..., |\mathcal{B}^*|\}} \kappa_i,\;\;\; and\;\;
 	 \gamma^*=\frac{1}{2\kappa_{*}^2} \label {eq:def_ruand_eta} 
 	\end{align}
% 	Furthermore, the state norm is upper-bounded by:
% \begin{align}
% \|x(t)\|&\leq \|x_0\|+U_{\Omega}:=\bar{X}_{max},\;\; U_{\Omega}:=\frac{n_s\sigma_{\omega} \kappa^*}{\gamma^*}\sqrt{n\log\frac{\tau_{mam}^{dw}}{\delta}}
% \end{align}
 \end{lemma}
 \begin{proof}
The proof is straightforward. 
Applying strongly stabilizing $K$ on the system with initial state $x_{t_0}$ guarantees the boundedness of state trajectory as follows:
 \begin{align}
\|x_t\|\leq \kappa_i (1-\gamma_i)^{t-t_0}\|x_{t_0}\|+\frac{\kappa_i}{\gamma_i} \max\limits_{t_0\leq s\leq t} \|w_s\|.\label{eq:rsp_seq_st} 
 \end{align}
The proof is completed by noting that the worst mode after a switch is the less suppressive one, i.e., the one with $\kappa_i=\kappa^*$.
 \end{proof}

% The second goal directly corresponds to the regret of Algorithm \ref{alg:OSL3}. We need the initial estimate $\|\Theta_0^i-\Theta_*^i\|\leq \bar{\epsilon}_i$ such that $\bar{\epsilon}_i\bar{\lambda}^i_{max}<1$ where $\bar{\lambda}^i_{max}=4\bar{\mu}^i_{max}\nu_i/\alpha_0^i\sigma_{\omega}^2$ and

% \begin{align*}
%      \big(1+2\vartheta_i (n\tau_{max}^{dw}+n\tau_{max}^{dw}(1+{\kappa_0^i}^2)^2\bar{X}^2_{max})^{0.5} \big):=\bar{\mu}_{max}^i
% \end{align*}

The second goal is that the estimate $\Theta_0^i$ satisfies $\|\Theta^i_0-\Theta^i_*\|\leq \bar{\epsilon}_i$ where $\bar{\epsilon}_i\bar{\lambda}_{\max}^i<1$ and $ \bar{\lambda}_{\max}^i=\max_{\tau_q} \tau_{\tau_q}^i$ where $\tau_q$'s are start of epochs. In Corollary 5 of \cite{cohen2018online} it is shown that if $n_s \tau^{dw}_{max}>poly(n, \nu, \vartheta^i, \alpha_0^{-1}, \sigma_{\omega}^{-1}, \kappa_0, \gamma_0^{-1}, \log (\delta^{-1}))$ with 

\begin{align}
    \tilde{T}_0=\mathcal{O}\big(\frac{n^2\nu_i^2\vartheta_i}{\alpha_0^5\sigma_{\omega}^{10}}\sqrt{(n_s\tau_{max}^{dw})\log^2\frac{n_s\tau_{max}^{dw}}{\delta}}\big) \label{secondCondition}
\end{align}
time steps we fulfill the second goal.

To put all in nutshell, to achieve an initial estimate $\Theta_0^i$ such that $\|\Theta_0^i-\Theta_*^i\|\leq min \{\bar{\epsilon}_i, \tilde{\epsilon}_i\}:=\epsilon_i$ we need to run warm-up algorithm for $T_0$ time steps that satisfies (\ref{firstcond}) and (\ref{secondCondition}).

 We can relax the assumption of having an initial stabilizing $K_0^i$ by adapting the results of \cite{lale2022reinforcement} and \cite{chekan2022joint}; however we prefer to stick to this assumption and be more focused on presenting the core contribution of the paper, learning mode-dependent minimum dwell time. 

\subsection{Stability Analysis}
\label{perturbationLemma_Stab}

We first need the following lemmas as required ingredients to prove Theorem \ref{Stability_thm17} and afterward Theorem \ref{thm:minimum_average_dwell_revised}. 

\begin{lemma} (Perturbation Lemma \cite{cohen2019learning})
\label{matix_perturbation_Bound1}
	Let $X$ and $\Delta$ be matrices of matching sizes and let $\Delta\Delta^T\leq r V^{-1}$ for positive definite matrix $V$. Then for any $P\geq 0$ and $\mu\geq r(1+2\|X\|\|V\|^{1/2})$, we have
	\begin{align*}
-\mu \|P\|_*V^{-1}\leq (X+\Delta)^\top P(X+\Delta)\leq \mu \|P\|_*V^{-1}.
	\end{align*}
\end{lemma}
\begin {proof}
The proof follows the same steps of Lemma 24 in \cite{cohen2019learning} with only the difference in $\mu$ definition, which is slightly different than the one of \cite{cohen2019learning}. This is due to the fact that unlike \cite{cohen2019learning}, we have not normalized our confidence ellipsoid.
\end{proof}

\begin{lemma} (Lemma 16 in \cite{cohen2019learning})
\label{matix_policy}
	Assume at time $\tau_q$ a switch to mode $i$ occurs. Let $\mu_i\geq r^i_{\tau_q}(1+2\vartheta_i\|V^i_{n_i(\tau_q)}\|^{1/2})$ and $V^i_{n_i(\tau_q)}\geq (4\nu_i\mu_i/\alpha^i_0\sigma_{\omega}^2)I$ (This holds true by rationale illustrated in Theorem \ref{Stability_thm17}). Let $\Sigma(\hat{\Theta}^i_{\tau_q}, Q^i, R^i)$ be the primal relaxed SDP solution of (\ref{eq:RedSDP_DUAL}) and let $P(\hat{\Theta}^i_{\tau_q}, Q^i, R^i)$ denote the solution of dual relaxed SDP (\ref{eq:RedSDP_DUAL}). Then $\Sigma_{xx}(\hat{\Theta}^i_{\tau_q}, Q^i, R^i)$ is invertible and
	\begin{align}
\nonumber &P(\hat{\Theta}^i_{\tau_q}, Q^i, R^i)=Q^i+K^\top (\hat{\Theta}^i_{\tau_q}, Q^i, R^i)R^i K(\hat{\Theta}^i_{\tau_q}, Q^i, R^i)\\
 \nonumber &+(\hat{A}^i_t+\hat{B}_t^iK(\hat{\Theta}^i_{\tau_q}, Q^i, R^i))^\top P(\hat{\Theta}^i_{\tau_q}, Q^i, R^i)\times\\
 \nonumber &(\hat{A}^i_{\tau_q}+\hat{B}_{\tau_q}^iK(\hat{\Theta}^i_{\tau_q}, Q^i, R^i))-\mu_i\|P(\hat{\Theta}^i_{\tau_q}, Q^i, R^i)\|_*\times\\
 &\begin{pmatrix} I \\ K(\hat{\Theta}^i_{\tau_q}, Q^i, R^i) \end{pmatrix}{V^i}^{-1}_{n_i(\tau_q)}\begin{pmatrix} I \\ K(\hat{\Theta}^i_{\tau_q}, Q^i, R^i) \end{pmatrix}^\top
\label{eq:Lemma8_res}
	\end{align}
	where $(\hat{A}^i_{\tau_q}, \hat{B}^i_{\tau_q})^\top=\hat{\Theta}^i_{\tau_q}$ and $K(\hat{\Theta}^i_{\tau_q}, Q^i, R^i)=\Sigma_{ux}(\hat{\Theta}^i_{\tau_q}, Q^i, R^i){\Sigma}^{-1}_{xx}(\hat{\Theta}^i_{\tau_q}, Q^i, R^i)$.
\end{lemma}

\textbf{Proof of Theorem \ref{Stability_thm17}}

% The proof directly follows by Lemma 18 in \cite{cohen2019learning}. However for the sake of completeness and avoiding possible misunderstandings due to the notation difference we provide it here. The proof also includes the definition of $\lambda_i$'s and upper-bound of $\mu_i$'s mentioned in Theorem \ref{Stability_thm17}.

\begin{proof} 

First we need to find a $\mu_i$ such that $\mu_i\geq r^i_{\tau_q}(1+2\vartheta_i\|V^i_{n_i(\tau_q)}\|^{1/2})$, we denote it by $\bar{\mu}^{\tau_q}_i$ which is defined as follows:
 \begin{align}
\nonumber &r^i_{\tau_q}(1+2\vartheta_i\|V^i_{n_i(\tau_q)}\|^{1/2})\leq \\
&r^i_{\tau_q}\bigg(1+2\vartheta_i \big(n_i(\tau_q)+\|\sum_{k=1}^{n_i(\tau_q)} z^i_k {z^i_k}^\top\|\big)^{0.5}\bigg):=\bar{\mu}^{\tau_q}_i \label{eq:mofid}
	\end{align}
where $r^i_{\tau_q}$ is $\mathcal{O}\big(\log (n_i(\tau_q))\big)$ considering the definition of $\epsilon_i$. In the second inequality we apply the fact that $\lambda_i^{\tau_q}< n_i(\tau_q)$ considering the definition given for $\lambda_{\tau_q}^i$ below.

Applying the perturbation Lemma \ref{matix_perturbation_Bound1} in (\ref{eq:Lemma8_res}) of Lemma \ref{matix_policy} yields	
\begin{align}
\nonumber &P(\hat{\Theta}^i_{\tau_q}, Q^i, R^i)\succeq Q^i+K^\top (\hat{\Theta}^i_{\tau_q}, Q^i, R^i)R^i K(\hat{\Theta}^i_{\tau_q}, Q^i, R^i)\\
\nonumber &+(A_*+B_*^iK(\hat{\Theta}^i_{\tau_q}, Q^i, R^i))^\top P(\hat{\Theta}^i_{\tau_q}, Q^i, R^i)\times \\
\nonumber &(A_*+B_*^iK(\hat{\Theta}^i_{\tau_q}, Q^i, R^i))-2\bar{\mu}_i^{\tau_q}\|P(\hat{\Theta}^i_{\tau_q}, Q^i, R^i)\|_*\times\\
&\begin{pmatrix} I \\ K(\hat{\Theta}^i_{\tau_q}, Q^i, R^i) \end{pmatrix}{V^i}^{-1}_{n_i(\tau_q)}\begin{pmatrix} I \\ K(\hat{\Theta}^i_{\tau_q}, Q^i, R^i) \end{pmatrix}^\top.
\label{eq:Lemma8_res22}
\end{align}

Recalling the dual formulation (\ref{eq:RedSDP_DUAL}) and the fact $\nu_i\geq J_*^i$ one can write 
	\begin{align*}
\nu_i\geq J_*^i\geq P(\hat{\Theta}^i_{\tau_q}, Q^i, R^i)\bullet W\geq \begin{pmatrix}
Q^i & 0 \\
0 & R^i
\end{pmatrix}\bullet W\geq \alpha^i_0\sigma_{\omega}^2 
	\end{align*}
	which clearly shows that $\nu_i/\alpha^i_0\sigma_{\omega}^2\geq1$. 
	
Now we let $\lambda^{\tau_q}_i:= 4\nu_i\bar{\mu}^{\tau_q}_i/\alpha^i_0\sigma_{\omega}^2$  and together with $\|P(\hat{\Theta}^i_{\tau_q}, Q^i, R^i)\|_*\leq\nu_i/\sigma_{\omega}^2$, and the fact that $V^i_{n_i(t)}\geq \lambda^{\tau_q}_i I$ $\forall t$ it yields	
\begin{align}
\nonumber \bar{\mu}_i\|P(\hat{\Theta}^i_{\tau_q}, Q^i, R^i)\|_*{V^i}_{n_i({\tau_q})}^{-1}&\leq \bar{\mu}_i\frac{\nu_i}{\sigma_{\omega}^2}\frac{\alpha^i_0\sigma_{\omega}^2}{4\nu_i\bar{\mu}_i}I\\
&\leq\bar{\mu}_i\frac{\nu_i}{\sigma_{\omega}^2}\frac{\alpha^i_0\sigma_{\omega}^2}{4\nu_i\bar{\mu}_i}I=\frac{\alpha^i_0}{4}I \label{eq:stabCon2}
\end{align}
plugging (\ref{eq:stabCon2}) together with the assumptions $Q^i, R^i\geq\alpha^i_0I$ into (\ref{eq:Lemma8_res22}) yields,

\begin{align}
\nonumber P(\hat{\Theta}^i_{\tau_q}, Q^i, R^i)& \geq \frac{\alpha^i_0}{2} I+\frac{\alpha^i_0}{2}K^\top (\hat{\Theta}^i_{\tau_q}, Q^i, R^i)K(\hat{\Theta}^i_{\tau_q}, Q^i, R^i)\\
 \nonumber &+(A_*+B^i_*K(\hat{\Theta}^i_{\tau_q}, Q^i, R^i))^\top P(\hat{\Theta}^i_{\tau_q}, Q^i, R^i)\\
 &\times (A_*+B^i_*K(\hat{\Theta}^i_{\tau_q}, Q^i, R^i)). \label{eq:inequality_of_P}
\end{align} 

Defining $\kappa_i=\sqrt{\frac{2\nu_i}{\alpha_0^i\sigma_{\omega}^2}}$ and using the fact that $\|P(\hat{\Theta}^i_{\tau_q}, Q^i, R^i)\|\leq\frac{\nu_i}{\sigma_{\omega}^2}$ (which implies $\kappa_i^{-2} P(\hat{\Theta}^i_{\tau_q}, Q^i, R^i)\leq \frac{\nu_i\kappa_i^{-2}}{\sigma_{\omega}^{2}}I$),

it holds true that
\begin{align}
P(\hat{\Theta}^i_{\tau_q}, Q^i, R^i)-\frac{1}{2}\alpha^i_0I\leq(1-\kappa_i^{-2}) P(\hat{\Theta}^i_{\tau_q}, Q^i, R^i) \label{eq:kappa_in_inequality}
\end{align}
From (\ref{eq:inequality_of_P}) and (\ref{eq:kappa_in_inequality}) one has:
	\begin{align}
\nonumber &{P}^{-1/2}(\hat{\Theta}^i_{\tau_q}, Q^i, R^i)(A_*+B^i_*K(\hat{\Theta}^i_{\tau_q}, Q^i, R^i))^\top \times \\
\nonumber &P(\hat{\Theta}^i_{\tau_q}, Q^i, R^i)(A_*+B^i_*K(\hat{\Theta}^i_{\tau_q}, Q^i, R^i)){P}^{-1/2}(\hat{\Theta}^i_{\tau_q}, Q^i, R^i)\\
&\leq (1-\kappa_i^{-2})I. \label{eq:uppBound_L}
\end{align}
Denoting 
\begin{align*}
    H_{\tau_q}^i=&{P}^{1/2}(\hat{\Theta}^i_{\tau_q}, Q^i, R^i)\\
    L_{\tau_q}t^i=&{P}^{-1/2}(\hat{\Theta}^i_{\tau_q}, Q^i, R^i)(A_*+B^i_*K(\hat{\Theta}^i_{\tau_q}, Q^i, R^i))\\
    &{P}^{1/2}(\hat{\Theta}^i_{\tau_q}, Q^i, R^i)
\end{align*}

 (\ref{eq:uppBound_L}) leads $\|L^i_{\tau_q}\|\leq\sqrt{1-\kappa_i^{-2}}\leq 1-1/2\kappa_i^{-2}$. 

Furthermore, (\ref{eq:inequality_of_P}) results in $P(\hat{\Theta}^i_{\tau_q}, Q^i, R^i)\geq \frac{\alpha^i_0}{2}{K}^\top (\hat{\Theta}^i_{\tau_q}, Q^i, R^i)K(\hat{\Theta}^i_{\tau_q}, Q^i, R^i)$ which together with $\|P(\hat{\Theta}^i_{\tau_q}, Q^i, R^i)\|\leq\frac{\nu_i}{\sigma_{\omega}^2}$  imply $\|K(\hat{\Theta}^i_{\tau_q}, Q^i, R^i)\|\leq \kappa_i$. No need to mention that in an epoch, for any  $t\in [\tau_q,\; \tau_{q+1}]$, $K(\hat{\Theta}^i_{t}, Q^i, R^i)=K(\hat{\Theta}^i_{\tau_q}, Q^i, R^i)$. 

This automatically yields $\|H^i_{\tau_q}\|\leq\sqrt{\nu_i/\sigma_{\omega}^2}$ and $\|{H^i_{\tau_q}}^{-1}\|\leq\sqrt{\frac{2}{\alpha^i_0}}$.
So, our proof is complete and $K(\hat{\Theta}^i_{\tau_q}, Q^i, R^i)$'s are $(\kappa_i,\gamma_i)-$ strongly stabilizing.
\end{proof}

\subsection{Minimum Dwell Time}

Before providing proof of Theorem \ref{thm:minimum_average_dwell_revised} we need the following Lemma. 

\begin{lemma} \cite{cohen2019learning}
\label{lem:tightnessSol}
    For mode $i$ starting $q-$th epoch at $\tau_q$ 
    \begin{align}
      P(\hat{\Theta}_{\tau_q}^i, Q^i, R^i)\preceq P({\Theta}_{*}^i, Q^i, R^i)\preceq &P(\hat{\Theta}_{\tau_q}^i, Q^i, R^i)+ \chi_{\tau_q}^i\label{eq:closness ofSolu}
    \end{align}
    where
    \begin{align}
 \nonumber  &\chi_{\tau_q}^i:=\frac{2\kappa_i^2\mu_i^{\tau_q}}{\gamma}\|P(\hat{\Theta}^i_{\tau_q}, Q^i, R^i)\|_* \\
  &\bigg\|\begin{pmatrix} I \\ K(\hat{\Theta}^i_{\tau_q}, Q^i, R^i) \end{pmatrix}{V^i}^{-1}_{n_i(\tau_q)}\begin{pmatrix} I \\ K(\hat{\Theta}^i_{\tau_q}, Q^i, R^i) \end{pmatrix}^\top\bigg\|I \label{eq:chiddef}
    \end{align}
    holds true.
\end{lemma}

\textbf{Proof of Theorem \ref{thm:minimum_average_dwell_revised}}

\begin{proof}
	Let the system actuate in the mode $i$ within an epoch starting at time $\tau_q$ and ending at $\tau_{q+1}$. Then for any $t\in [\tau_q, \; \tau_{q+1}]$, we designate the potential Lyapunov function as follows:
	\begin{align}
	\mathcal{V}_i(t)=x_t^\top P(\hat{\Theta}^i_{\tau_q}, Q^i, R^i)x_t
	\end{align}
	where $P(\hat{\Theta}^i_{\tau_q}, Q^i, R^i)$ is the solution of relaxed dual program (\ref{eq:RedSDP_DUAL}) for mode $i$, computed by $n_i(\tau_q)$ number of data. Notably, this solution is both positive definite and radially unbounded. Also, by applying Rayleigh-Ritz inequality one can write
 
	\begin{align}
	\nonumber &\underline{\lambda}\big(P(\hat{\Theta}^i_{\tau_q}, Q^i, R^i)\big)\mathbb{E}[x_t^\top x_t| \mathcal{F}_{t-1}]\leq\\
 &\mathbb{E}[\mathcal{V}_i(t)| \mathcal{F}_{t-1}]\leq \overline{\lambda}\big(P(\hat{\Theta}^i_{\tau_q}, Q^i, R^i)\big)\mathbb{E}[x_t^\top x_t| \mathcal{F}_{t-1}].\label{Ray-Rit}
	\end{align}

by martingale difference properties of process noise we have

	\begin{align}
	\nonumber &\mathbb{E}[\mathcal{V}_i(t+1)| \mathcal{F}_t]-\mathcal{V}_i(t)=\\
 \nonumber &x^\top_t(A^i_*+B^i_*K(\hat{\Theta}^i_{\tau_q}, Q^i, R^i))^\top P(\hat{\Theta}^i_{\tau_q}, Q^i, R^i)\times \\
 \nonumber &(A^i_*+B^i_*K(\hat{\Theta}^i_{\tau_q}, Q^i, R^i))x_t+\\
	& \mathbb{E}[\omega^\top_{t+1}P(\hat{\Theta}^i_{\tau_q}, Q^i, R^i)\omega_{t+1}|\mathcal{F}_t]-x^T_tP(\hat{\Theta}^i_{\tau_q}, Q^i, R^i)x_t \label{eq:LyapunocIneq}
	\end{align}
	Now, using (\ref{eq:Lemma8_res22}) one can write,

\begin{align}
\nonumber &P(\hat{\Theta}^i_{\tau_q}, Q^i, R^i)- Q^i-K^\top (\hat{\Theta}^i_{\tau_q}, Q^i, R^i)R^i K(\hat{\Theta}^i_{\tau_q}, Q^i, R^i)+\\
\nonumber &2\mu_i\|P(\hat{\Theta}^i_{\tau_q}, Q^i, R^i)\|_*\times \\
\nonumber &\begin{pmatrix} I \\ K(\hat{\Theta}^i_{\tau_q}, Q^i, R^i) \end{pmatrix}{V^i}^{-1}_{n_i(\tau_q)}\begin{pmatrix} I \\ K(\hat{\Theta}^i_{\tau_q}, Q^i, R^i) \end{pmatrix}^\top \succeq \\
\nonumber &(A^i_*+B_*^iK({\Theta}^i_{\tau_q}, Q^i, R^i))^\top P(\hat{\Theta}^i_{\tau_q}, Q^i, R^i)\times \\
&(A^i_*+B_*^iK(\hat{\Theta}^i_{\tau_q}, Q^i, R^i)). \label{eq:useful3}
\end{align}

Substituting which in (\ref{eq:LyapunocIneq}) yields
 
	\begin{align}
\nonumber &\mathbb{E}[\mathcal{V}_i(t+1)|\mathcal{F}_t]-\mathcal{V}_i(t)\leq \\
&-x^\top_t\mathcal{H}(\mathcal{C}^i_{\tau_q}(\delta))x_t+\overline{\lambda}\big(P(\hat{\Theta}^i_{\tau_q}, Q^i, R^i)\big)W \label{eq:simplifiedLyap}
	\end{align}
	in which,
	\begin{align}
	\nonumber &\mathcal{H}(\mathcal{C}^i_{\tau_q}(\delta))=\\
\nonumber &Q^i+K^\top (\hat{\Theta}^i_{\tau_q}, Q^i, R^i)R^i K(\hat{\Theta}^i_{\tau_q}, Q^i, R^i)-\\
\nonumber&2\mu_i^{\tau_q}\|P(\hat{\Theta}^i_{\tau_q}, Q^i, R^i)\|_*\\
 &\begin{pmatrix} I \\ K(\hat{\Theta}^i_{\tau_q}, Q^i, R^i) \end{pmatrix}{V^i}^{-1}_{n_i(\tau_q)}\begin{pmatrix} I \\ K(\hat{\Theta}^i_{\tau_q}, Q^i, R^i) \end{pmatrix}^\top .\label{eq:Hdefinition}
	\end{align} 

 % By choosing $\mu_{i}^{\tau_q}$, as already specified in Theorem \ref{Stability_thm17}, and $\lambda_i^{\tau_q}$ accordingly it is straightforward to show that $\mathcal{H}(\mathcal{C}^i_{\tau_q}(\delta))\succeq 0$ which guarantees strong-stability of closed-loop, i.e., $(A+BK)^\top P(A+BK)\preceq (1-\kappa_i^{-2})P$ (see (\ref{eq:komak2}) in Appendix).

It is noteworthy that (\ref{eq:Hdefinition}) can be easily computed using confidence ellipsoid $\mathcal{C}^i_{\tau_q}(\delta)$.

 Observing that 

\begin{align*}
    \mathbb{E}[\mathbb{E}[\mathcal{V}_i(t+1)| \mathcal{F}_t]|\mathcal{F}_{t-1}]=\mathbb{E}[\mathcal{V}_i(t+1)| \mathcal{F}_t]
\end{align*}
holds true and the fact that
\begin{align*}
    \mathbb{E}[\mathcal{H}(\mathcal{C}^i_{\tau_q}(\delta))| \mathcal{F}_{t-1}]=\mathcal{H}(\mathcal{C}^i_{\tau_q}(\delta))
\end{align*}
and 
\begin{align*}
    \mathbb{E}[P(\hat{\Theta}^i_{\tau_q}, Q^i, R^i)| \mathcal{F}_{t-1}]=P(\hat{\Theta}^i_{\tau_q}, Q^i, R^i).
\end{align*}
are true for all $t\geq \tau_q$ we can rewrite (\ref{eq:simplifiedLyap}) as follows
	\begin{align}
	\nonumber &\mathbb{E}[\mathcal{V}_i(t+1)| \mathcal{F}_t]-\mathbb{E}[\mathcal{V}_i(t)| \mathcal{F}_{t-1}] \leq \\
 &-\underline{\lambda}\big(\mathcal{H}(\mathcal{C}^i_{\tau_q}(\delta))\big)\mathbb{E}[x_t^\top x_t| \mathcal{F}_{t-1}]+\overline{\lambda}\big(P(\hat{\Theta}^i_{\tau_q}, Q^i, R^i)\big)\sigma_{\omega}^2
	\end{align}
	combining which with (\ref{Ray-Rit}) gives
	\begin{align}
	&\nonumber \mathbb{E}[\mathcal{V}_i(t+1)| \mathcal{F}_t]-\mathbb{E}[\mathcal{V}_i(t)| \mathcal{F}_{t-1}]\leq\\
 & -\frac{\underline{\lambda}\big(\mathcal{H}(\mathcal{C}^i_{\tau_q}(\delta)\big)}{\overline{\lambda}\big(P(\hat{\Theta}^i_{\tau_q}, Q^i, R^i)\big)}\mathbb{E}[\mathcal{V}_i(t)| \mathcal{F}_{t-1}]+\overline{\lambda}\big(P(\hat{\Theta}^i_{\tau_q}, Q^i, R^i)\big)\sigma_{\omega}^2
	\end{align}
which can be briefly rewritten as follows:
	\begin{align}
	\nonumber &\mathbb{E}[\mathcal{V}_i(t+1)| \mathcal{F}_t]\leq\\
 &\bigg(1-\eta \big(\mathcal{C}^i_{\tau_q}(\delta)\big)\bigg)\mathbb{E}[\mathcal{V}_i(t)| \mathcal{F}_{t-1}]+\overline{\lambda}\big(P(\hat{\Theta}^i_{\tau_q}, Q^i, R^i)\big)\sigma_{\omega}^2 \label{ineq_epoch}
	\end{align}
 where 
	\begin{align}
	\eta \big(\mathcal{C}^i_{\tau_q}(\delta)\big):= \frac{\underline{\lambda}\big(\mathcal{H}(\mathcal{C}^i_{\tau_q}(\delta)\big)}{\overline{\lambda}\big(P(\hat{\Theta}^i_{\tau_q}, Q^i, R^i)\big)}. \label{eq:etadefApp}
	\end{align}

By carrying the inequality (\ref{ineq_epoch}) through an epoch cycle, starting from $\tau_{q+1}^-$ (just prior to the next switch time) and propagating back to $\tau_{q}^+$ (immediately after the current switch time), we arrive at:
	% \begin{align}
	% E[\mathcal{V}_i(t)| \mathcal{F}_{t-1}]\leq \bigg(1-\eta \big(\mathcal{C}^i_{\tau_q}(\delta)\big)\bigg)^{t-\tau_q}\mathcal{V}_i(\tau_{q^+})+\bigg(1-\eta \big(\mathcal{C}^i_{\tau_q}(\delta)\big)\bigg)^{t-\tau_q-1}\overline{\lambda}\big(P(\hat{\Theta}^i_{\tau_q}, Q^i, R^i)\big)\max_{\tau_q\leq s\leq t}\|\omega_{s}\|^2 \label{eq:inmode}
	% \end{align}

 	\begin{align}
	\nonumber &\mathbb{E}[\mathcal{V}_i(\tau_{q+1}^-)| \mathcal{F}_{\tau_{q+1}-1}]\leq\\
 \nonumber &\bigg(1-\eta \big(\mathcal{C}^i_{\tau_q}(\delta)\big)\bigg)^{\tau_{q+1}-\tau_q}\mathbb{E}[\mathcal{V}_i(\tau_{q}^+)|\mathcal{F}_{\tau_{q}-1} ]+\\
 &\sum_{k=\tau_q+1}^{\tau_{q+1}}\bigg(1-\eta \big(\mathcal{C}^i_{\tau_q}(\delta)\big)\bigg)^{k-\tau_q-1}\overline{\lambda}\big(P(\hat{\Theta}^i_{\tau_q}, Q^i, R^i)\big)\sigma_{\omega}^2 \label{eq:inmode}
	\end{align}
 
Suppose the system, after being in mode $i$ at time $\tau_{q+1}$, switches to mode $j$. In this context, we introduce teh following candidate Lyapunov function for mode $j$:
    \begin{align}
	\mathcal{V}_j(t)=x_t^\top P(\hat{\Theta}^j_{\tau_q}, Q^j, R^j)x_t,
	\end{align} 
Then one can establish

   	\begin{align}
\nonumber &\mathbb{E}[\mathcal{V}_i(\tau_{q+1}^-)| \mathcal{F}_{\tau_{q+1}-1}]\geq \\
\nonumber &\underline{\lambda}\big(P(\hat{\Theta}^i_{\tau_q}, Q^i, R^i)\big)\mathbb{E}[x_{\tau_{q+1}}^\top x_{\tau_{q+1}}| \mathcal{F}_{\tau_{q+1}-1}]\geq \\
&\frac{\underline{\lambda}\big(P(\hat{\Theta}^i_{\tau_q}, Q^i, R^i)}{\overline{\lambda}\big(P({\Theta}^j_{*}, Q^j, R^j)\big)}\mathbb{E}[\mathcal{V}_j(\tau_{q+1}^+)| \mathcal{F}_{\tau_{q+1}-1}]
\label{eq:jumpLyap}
	\end{align}
 where in the first and second inequalities we applied (\ref{Ray-Rit}). 

By defining

\begin{align}
   \rho (\mathcal{C}^i_{\tau_q}(\delta), \mathcal{C}^j_{\tau_{q}}(\delta)):=\frac{\overline{\lambda}\big(P(\hat{\Theta}^j_{\tau_{q}}, Q^j, R^j))\big)}{\underline{\lambda}\big(P(\hat{\Theta}^i_{\tau_{q}}, Q^i, R^i)}
\end{align}

 and combining (\ref{eq:jumpLyap}) with the left hand side of (\ref{eq:inmode}) yields

 	\begin{align}
	\nonumber &\mathbb{E}[\mathcal{V}_j(\tau_{q+1}^+)| \mathcal{F}_{\tau_{q+1}-1}]\leq\\
\nonumber  &\rho (\mathcal{C}^i_{\tau_q}(\delta), \mathcal{C}^j_{\tau_{q}}(\delta))\bigg(1-\eta \big(\mathcal{C}^i_{\tau_q}(\delta)\big)\bigg)^{\tau_{q+1}-\tau_q}\mathbb{E}[\mathcal{V}_i(\tau_{q}^+)|\mathcal{F}_{\tau_{q}-1} ]+ \\
 &\sum_{k=\tau_q+1}^{\tau_{q+1}}\bigg(1-\eta \big(\mathcal{C}^i_{\tau_q}(\delta)\big)\bigg)^{k-\tau_q-1} \overline{\lambda}\big(P(\hat{\Theta}^i_{\tau_q}, Q^i, R^i)\big)\sigma_{\omega}^2 \label{eq:UppExV0}
	\end{align}

Now applying (\ref{Ray-Rit}) yields 

	\begin{align}
	\nonumber &\mathbb{E}[x^\top_{\tau_{q+1}^+}x_{\tau_{q+1}^+}| \mathcal{F}_{\tau_{q+1}-1}]\leq \rho (\mathcal{C}^i_{\tau_q}(\delta), \mathcal{C}^j_{\tau_{q}}(\delta)) \mathcal{X} (\mathcal{C}^i_{\tau_q}(\delta), \mathcal{C}^j_{\tau_{q}}(\delta)) \\
\nonumber  &\bigg(1-\eta \big(\mathcal{C}^i_{\tau_q}(\delta)\big)\bigg)^{\tau_{q+1}-\tau_q}\mathbb{E}[x^\top_{\tau_{q}^+}x_{\tau_{q}^+}|\mathcal{F}_{\tau_{q}-1} ]+ \\
 &\sum_{k=\tau_q+1}^{\tau_{q+1}}\bigg(1-\eta \big(\mathcal{C}^i_{\tau_q}(\delta)\big)\bigg)^{k-\tau_q-1} \mathcal{X} (\mathcal{C}^i_{\tau_q}(\delta), \mathcal{C}^j_{\tau_{q}}(\delta))\sigma_{\omega}^2 \label{eq:UppExV}
	\end{align}

where

\begin{align*}
    \mathcal{X} (\mathcal{C}^i_{\tau_q}(\delta), \mathcal{C}^j_{\tau_{q}}(\delta)):=\frac{\overline{\lambda}\big(P(\hat{\Theta}^i_{\tau_{q}}, Q^i, R^i)\big)}{\underline{\lambda}\big(P(\hat{\Theta}^j_{\tau_{q}}, Q^j, R^j)}
\end{align*}

The boundedness of the second term on the right hand side of (\ref{eq:UppExV}) follows by following argument. Recalling the definition of $\mathcal{H} (\mathcal{C}^i_{\tau_q}(\delta))$ as provided in (\ref{eq:Hdefinition}), and taking into consideration (\ref{eq:Lemma8_res22}), along with the fact that $P(\hat{\Theta}^i_{\tau_q}, Q^i, R^i)\succeq 0$, this leads to the conclusion that $P(\hat{\Theta}^i_{\tau_q}, Q^i, R^i)\succeq \mathcal{H}(\mathcal{C}^i_{\tau_q}(\delta))$, resulting in $\eta \big(\mathcal{C}^i_{\tau_q}(\delta)\big)<1$. Also by stability condition (\ref{eq:stabCon2}) it is straight forward to notice $\eta \big(\mathcal{C}^i_{\tau_q}(\delta)\big)>0$. This ensures that the sum of the geometric series $1-\eta \big(\mathcal{C}^i_{\tau_q}(\delta)\big)$ remains bounded. Moreover, in conjunction with $\overline{\lambda}\big(P(\hat{\Theta}^i_{\tau_q}, Q^i, R^i)\leq \nu_i/\sigma_{\omega}^2$, the boundedness of the second term on the right-hand side of (\ref{eq:UppExV}) is guaranteed.

In order to manage the growth of the expected state norm, the following condition needs to hold for a user-defined $\mathcal{Y}<1$:

\begin{align}
   \rho (\mathcal{C}^i_{\tau_q}(\delta), \mathcal{C}^j_{\tau_{q}}(\delta))  \mathcal{X} (\mathcal{C}^i_{\tau_q}(\delta), \mathcal{C}^j_{\tau_{q}}(\delta)) \big(1-\eta \big(\mathcal{C}^i_{\tau_q}(\delta)\big)^{\tau_{dw}^q}&\leq \mathcal{Y}\label{eq:contaction}
\end{align}

which is done by requiring $\tau_{dw}^q:=\tau_{q+1}-\tau_{q}$, denoted as the dwell time of the epoch starting at time $\tau_q$, meets the condition:

 \begin{align}
	\tau_{dw}^q\geq -\frac{\ln \rho (\mathcal{C}^i_{\tau_q}(\delta), \mathcal{C}^j_{\tau_{q}}(\delta))+\ln  \mathcal{X} (\mathcal{C}^i_{\tau_q}(\delta), \mathcal{C}^j_{\tau_{q}}(\delta))-\ln \mathcal{Y}}{\ln \big(1-\eta \big(\mathcal{C}^i_{\tau_q}(\delta)\big)}.
	\label{tau_a} 
	\end{align}

Note that when 

\begin{align*}
    \ln \rho (\mathcal{C}^i_{\tau_q}(\delta), \mathcal{C}^j_{\tau_{q}}(\delta))+\ln  \mathcal{X} (\mathcal{C}^i_{\tau_q}(\delta), \mathcal{C}^j_{\tau_{q}}(\delta))< \ln \mathcal{Y}
\end{align*}
 it indicates a benign situation where there is no necessity to prolong mode $i$ and a subsequent switch can occur quickly. However, since it is mandatory to spend a minimum of one second in each mode, we set forth the minimum mode-specific dwell time as follows:

\begin{align}
  \nonumber  & [\tau_{md}^{ij}]_{q}:=\\
    &\max\bigg\{1, -\frac{\ln \rho (\mathcal{C}^i_{\tau_q}(\delta), \mathcal{C}^j_{\tau_{q}}(\delta))+\ln \mathcal{X} (\mathcal{C}^i_{\tau_q}(\delta), \mathcal{C}^j_{\tau_{q}}(\delta))-\ln \mathcal{Y}}{\ln \big(1-\eta \big(\mathcal{C}^i_{\tau_q}(\delta)\big)}\bigg\} \label{eq:komaki3}
\end{align}

which completes the proof.
\end{proof}

\textbf{Proof of Lemma \ref{thm:res_sequential_stability_resp}}

\begin{proof}

Our first step involves determining a lower bound for $\eta \big(\mathcal{C}^i_{\tau_q}(\delta)\big)$. This can be derived as follows: Given the definition of $\mathcal{H}(\mathcal{C}^i_{\tau_q}(\delta))$ according to equation (\ref{eq:Hdefinition}), we can conclude that

\begin{align*}
   \mathcal{H}(\mathcal{C}^i_{\tau_q}(\delta))\geq \alpha_0(1+\kappa^2)-2 \frac{\alpha_0}{4}(1+\kappa^2)=\frac{\alpha_0}{2}(1+\kappa^2)
\end{align*}
in which we applied (\ref{eq:stabCon2}). Consequently, following the definition of $\eta \big(\mathcal{C}^i_{\tau_q}(\delta)\big)$, one can write

\begin{align}
   \nonumber  &\eta \big(\mathcal{C}^i_{\tau_q}(\delta)\big) \geq \frac{\frac{\alpha^i_0}{2}(1+\kappa_i^2)}{\frac{\nu^i}{\sigma_{\omega}^2}}=\frac{(1+\kappa_i^2)}{\kappa_i^2}\geq \frac{1}{\kappa_i^2}\geq \frac{1}{\kappa_*^2}:=\eta_* \\
   &  \mathcal{X} (\mathcal{C}^i_{\tau_q}(\delta), \mathcal{C}^j_{\tau_{q}}(\delta))\leq \frac{\frac{\nu^i}{\sigma_{\omega}^2}}{\frac{\alpha_0^j}{2}}=\kappa_i^2 \frac{\alpha_0^i}{\alpha_0^j}\leq \kappa_*^2 \frac{\alpha_1^*}{\alpha_0^*}:=\mathcal{X}_*
    \label{eq:khoobeaali}
\end{align}
where the last inequality of both expression hold by definitions $\max_{i\in |\mathcal{M}|}\kappa_i=: \kappa_*$ and $\alpha^*_0 I\preceq Q^i, R^i \preceq \alpha^*_1 I$ for all $i\in |\mathcal{M}|$.

Let's begin by providing the proof for the first statement of the lemma. We start with equations (\ref{eq:UppExV}) and (\ref{eq:contaction}), which yield the following result:

	\begin{align}
	\nonumber &\mathbb{E}[x^\top_{\tau_{q+1}^+}x_{\tau_{q+1}^+}| \mathcal{F}_{\tau_{q+1}-1}]\leq \mathcal{Y}\mathbb{E}[x^\top_{\tau_{q}^+}x_{\tau_{q}^+}|\mathcal{F}_{\tau_{q}-1} ]+ \\
 &\sum_{k=\tau_q+1}^{\tau_{q+1}}\bigg(1-\eta \big(\mathcal{C}^i_{\tau_q}(\delta)\big)\bigg)^{k-\tau_q-1} \mathcal{X} (\mathcal{C}^i_{\tau_q}(\delta), \mathcal{C}^j_{\tau_{q}}(\delta))\sigma_{\omega}^2 \\
 &\leq \mathcal{Y}\mathbb{E}[x^\top_{\tau_{q}^+}x_{\tau_{q}^+}|\mathcal{F}_{\tau_{q}-1} ]+ \frac{1}{\eta \big(\mathcal{C}^i_{\tau_q}(\delta)\big)} \frac{\frac{\nu^i}{\sigma_{\omega}^2}}{\frac{\alpha_0^j}{2}}\sigma_{\omega}^2\leq \\
 & \mathcal{Y}\mathbb{E}[x^\top_{\tau_{q}^+}x_{\tau_{q}^+}|\mathcal{F}_{\tau_{q}-1} ]+\kappa_i^4 \frac{\alpha_0^i}{\alpha_0^j} \sigma_{\omega}^2\leq\\
 &\mathcal{Y}\mathbb{E}[x^\top_{\tau_{q}^+}x_{\tau_{q}^+}|\mathcal{F}_{\tau_{q}-1} ]+\kappa_*^4 \frac{\alpha_1^*}{\alpha_0^*} \sigma_{\omega}^2
	\end{align}
inwhich, we have used an upper boundary of $1/\eta \big(\mathcal{C}^i_{\tau_q}(\delta)\big)$ on the summation of a geometric series, considering the fact that $0<\eta \big(\mathcal{C}^i_{\tau_q}(\delta)\big)<1$ then applied (\ref{eq:khoobeaali}). Moreover, we employed the inequalities $\overline{\lambda}\big(P(\hat{\Theta}^i_{\tau_{q}}, Q^i, R^i)\big)\leq \nu^i/\sigma_{\omega}^2$ and $\underline{\lambda}\big(P(\hat{\Theta}^j_{\tau_{q}}, Q^j, R^j)\geq \alpha_0^j/2$. Subsequently, we applied (\ref{eq:khoobeaali}). These steps together constitute the completion of the proof for the first part of the lemma.

To prove the second and third components of the lemma, we proceed as follow

Recalling definition of $k^{th}$ switch time $\tau_k=\sum_{q=0}^{k-1}[\tau_{md}^{i_qi_{q+1}}]_q$ and employing (\ref{eq:UppExV}),  we propagate from $\tau_k$ back to the initial time $t=0$. We use the intermittent variables $\eta_*$ and $\mathcal{X}_*$ defined in (\ref{eq:khoobeaali}), and apply gemoetric series bound on the summation of $(1-\eta_*)$ which results in

\begin{align}
    \mathbb{E}[x^\top_{\tau_{k}^+}x_{\tau_{k}^+}| \mathcal{F}_{\tau_{k}-1}]\leq \mathcal{Y}^k x_0^\top x_0 +\frac{\mathcal{X}_*}{\eta_*} \sigma_{\omega}^2\sum_{m=1}^{k-1} \mathcal{Y}^m \label{eq:onemorest}
\end{align}

Moreover, for any instant $t$ falling within the epoch initiated at $\tau_{k}$ and involving action in mode $j$, the application of (\ref{ineq_epoch}) along with a similar analogy to (\ref{eq:inmode}) enables us to formulate:

 	\begin{align}
	\nonumber &\mathbb{E}[\mathcal{V}_j(t)| \mathcal{F}_{t-1}]\leq\\
 \nonumber &\bigg(1-\eta \big(\mathcal{C}^j_{\tau_k}(\delta)\big)\bigg)^{t-\tau_{k}}\mathbb{E}[\mathcal{V}_j(\tau_{k}^+)|\mathcal{F}_{\tau_{k}-1} ]+\\
 &\sum_{m=\tau_{k}+1}^{t}\bigg(1-\eta \big(\mathcal{C}^j_{\tau_k}(\delta)\big)\bigg)^{k-\tau_k-1}\overline{\lambda}\big(P(\hat{\Theta}^j_{\tau_q}, Q^j, R^j)\big)\sigma_{\omega}^2 \label{eq:good to go}
	\end{align}
which yields to

	\begin{align}
	\nonumber &\mathbb{E}[x_t^\top x_t| \mathcal{F}_{t-1}]\leq \kappa_*^2 \bigg(1-\eta \big(\mathcal{C}^j_{\tau_k}(\delta)\big)\bigg)^{t-\tau_{k}}\mathbb{E}[x^\top_{\tau_{k}^+}x_{\tau_{k}^+}| \mathcal{F}_{\tau_{k}-1}]\\
 &+\frac{\mathcal{X_*}}{\eta_*}  \sigma_{\omega}^2.\label{eq:borokeberim}
	\end{align}
 Now by combining (\ref{eq:onemorest}) and (\ref{eq:borokeberim}) for $\tau_k\leq t<\tau_{k+1}$ we summarize

 \begin{align}
  \nonumber &\mathbb{E}[x_t^\top x_t| \mathcal{F}_{t-1}]\leq \mathcal{Y}^k\kappa_*^2 \bigg(1-\eta \big(\mathcal{C}^j_{\tau_k}(\delta)\big)\bigg)^{t-\tau_{k}}x_0^\top x_0 +\\
  \nonumber&(1+\frac{1}{1-\mathcal{Y}}) \frac{\mathcal{X}_*^2}{\eta_*}\sigma_{\omega}^2\leq\\
 & \mathcal{Y}^k\kappa_*^2 x_0^\top x_0+ \frac{2-\mathcal{Y}}{1-\mathcal{Y}} \frac{\mathcal{X}_*^2}{\eta_*}\sigma_{\omega}^2
\end{align}
which completes proof of the second part. The third part directly follows by the second part.
 
\end{proof}

\subsection{Minimum Mode-dependent Dwell-Time Estimation Error}

In this subsection, we give an upper-bound for dwell time estimation error $[\tau_{md}^{ij}]_{\tau_q}-\tau_*^{ij} $, which is used for the regret bound analysis.

\textbf{Proof of Theorem \ref{Thm:dwellTimeError}}

\begin{proof}

\textbf{proof of part (a)}

The most pessimistic upper limit on the error in estimating the minimum mode-dependent dwell time $[\tau_{md}^{ij}]_{\tau_q}-\tau^{ij}_*$ occurs when the type of switch is benign according to Definition \ref{def:malignant}. In such a case, $\tau_*^{ij}$ reaches its lowest value of one. Consequently, we directly bound the minimum mode-dependent dwell time itself. This result can also be applied to prove part (b) of the proof.

By Lemma \ref{lem:tightnessSol} for mode $i$ we can write

\begin{align}
\nonumber \overline{\lambda}\big(P(\hat{\Theta}^i_{\tau_q}, Q^i, R^i)\big)&\leq \overline{\lambda} (\big(P(\hat{\Theta}^i_{\tau_q}, Q^i, R^i)+ \chi^i_{\tau_q}\big)\\
\nonumber &\leq \overline{\lambda} \big(P(\hat{\Theta}^i_{\tau_q}, Q^i, R^i)\big)+ \overline{\lambda} \big( \chi^i_{\tau_q}\big)\\
& \leq \kappa_i^2 \underline{\lambda} \big(P(\hat{\Theta}^i_{\tau_q}, Q^i, R^i)\big)+ \overline{\lambda} \big( \chi^i_{\tau_q}\big) \label{eq:bekar}
\end{align}

where in the second inequality we applied Weyl's inequality and in the last inequality we applied 

\begin{align}
\frac{\overline{\lambda}\big(P(\hat{\Theta}^i_{\tau_q}, Q^i, R^i)\big)}{\underline{\lambda}\big(P(\hat{\Theta}^i_{\tau_q}, Q^i, R^i)\big)}\leq \frac{\frac{\nu_i}{\sigma_{\omega}^2}}{\alpha_0^i/2}=\kappa_i^2.
\end{align}
By dividing both sides of (\ref{eq:bekar}) by $\underline{\lambda} \big(P(\hat{\Theta}^i_{\tau_q}, Q^i, R^i)\big)$ and noting that $\underline{\lambda} \big(P(\hat{\Theta}^i_{\tau_q}, Q^i, R^i)\big)\geq \alpha_o^i/2$
one can write:
\begin{align}
\frac{\overline{\lambda}\big(P(\hat{\Theta}^j_{\tau_q}, Q^j, R^j)}{\underline{\lambda}\big(P(\hat{\Theta}^i_{\tau_q}, Q^i, R^i)\big)}\leq \kappa_i^2+ \frac{2}{\alpha_0^i}\overline{\lambda}\big(\chi^i_{\tau_q}\big).
\end{align}
Now, we have
\begin{align}
 \nonumber  &\ln \rho (\mathcal{C}^i_{\tau_q}(\delta), \mathcal{C}^j_{\tau_{q}}(\delta))+\ln  \mathcal{X} (\mathcal{C}^i_{\tau_q}(\delta), \mathcal{C}^j_{\tau_{q}}(\delta))=  \\
 \nonumber  &\ln \frac{\overline{\lambda}\big(P(\hat{\Theta}^i_{\tau_q}, Q^i, R^i)\big)}{\underline{\lambda}\big(P(\hat{\Theta}^i_{\tau_q}, Q^i, R^i)\big)} +\ln \frac{\overline{\lambda}\big(P(\hat{\Theta}^j_{\tau_q}, Q^j, R^j)\big)}{\underline{\lambda}\big(P(\hat{\Theta}^j_{\tau_q}, Q^j, R^j)\big)}\leq \\
 \nonumber & \ln \big(\kappa_i^2+ \frac{2}{\alpha_0^i}\overline{\lambda}\big(\chi^i_{\tau_q}\big)\big) + \ln \big(\kappa_j^2+ \frac{2}{\alpha_0^j}\overline{\lambda}\big(\chi^j_{\tau_q}\big)\big)=\\
  & \ln \kappa_i^2 \ln \big(1+ \frac{2}{\kappa_i^2\alpha_0^i}\overline{\lambda}\big(\chi^i_{\tau_q}\big)\big)+\ln \kappa_j^2 \ln \big(1+ \frac{2}{\kappa_j^2\alpha_0^j}\overline{\lambda}\big(\chi^j_{\tau_q}\big)\big)
\end{align}

Now considering the definition of minimum dwell time and the fact that $\eta \big(\mathcal{C}^i_{\tau_q}(\delta)\big)\geq 1/\kappa_i^2$ one can write:

\begin{align}
    -ln \big(1-\eta \big(\mathcal{C}^i_{\tau_q}(\delta)\big)\geq -ln \big(1-\kappa_i^{-2}\big).
\end{align}

which yields to

\begin{align}
 \nonumber &[\tau_{md}^{ij}]_{\tau_q}-\tau_*^{ij}\leq\\
 &[\tau_{md}^{ij}]\leq \\
   \nonumber  &\frac{ \sum_{s=i,j}\ln \kappa_s^2+ \ln \big(1+ \frac{2}{\kappa_s^2\alpha_0^s}\overline{\lambda}\big(\chi^s_{\tau_q}\big)\big)-\ln \mathcal{Y}}{-ln \big(1-\kappa_i^{-2}\big)}
\end{align}

that completes the proof.
\end{proof}

Following lemma gives an statement to be used in the regret bound analysis section.

\begin{lemma}
    The following statement holds true.

    \begin{align}
    [\tau_{md}^{ij}]_{\tau_q}-\tau_*^{ij}\leq \frac{\ln \frac{\kappa_i^2 \kappa_j^2}{\mathcal{Y}}}{-ln \big(1-\kappa_s^{-2}\big)}+\sum_{s=i,j}\frac{\sqrt{\frac{2}{\kappa_s^2 \alpha_0^s}\overline{\lambda}\big(\chi^s_{\tau_q}\big)}}{-ln \big(1-\kappa_s^{-2}\big) }
\end{align}
    
\end{lemma}
\begin{proof}
    
The proof directly comes from the following useful inequality 

\begin{align*}
   ln \big(1+x\big)\leq \frac{x}{\sqrt{1+x}},\;\; \textit{for $x>0$} . 
\end{align*}

\end{proof}

\subsection{Regret Bound Analysis (Proof of Theorem \ref{thm:RegretBound})}

Recalling (\ref{eq:RegretSwitchDec}) to bound the first term $R_T^1$, we start with upper-bounding the regret in the epoch starting with $k^{th}$ switch, i.e.,  $\tau_k\leq t<\mathfrak{T}_k$. For the sake of brevity of notations, we let $s_1:=\tau_k$, $s_2:=\mathfrak{T}_k$, $i:=i_k$, and $j:=i_{k+1}$. Then the regret of actuating in the corresponding epoch is written as follows:

\begin{align*}
    \mathcal{R}_{\tau_k:\mathfrak{T}_k}=\sum _{t=\tau_k}^{\mathfrak{T}_k-1} (c_{t}^{i}-J_*(\Theta_*^{i}, Q^{i}, R^{i}))
\end{align*}
where 

\begin{align}
    c_{t}^i=x_t^\top \big(Q^i+K^\top (\hat{\Theta}^i_{\tau_k}, Q^i, R^i)R^i K(\hat{\Theta}^i_{\tau_k}, Q^i, R^i)\big)x_t
\end{align}
 as the algorithm commits to a fixed policy during the epoch.

Given that $J_*(\Theta_*^{i}, Q^{i}, R^{i})=P(\Theta_*^{i}, Q^{i}, R^{i}) \bullet \sigma_{\omega}^2I$, and  $P(\hat{\Theta}^i_{s_1}, Q^i, R^i)\preceq P({\Theta}^i_{*}, Q^i, R^i)$ (see Lemma 19 in \cite{cohen2019learning}) it yields

\begin{align}
    J_*(\Theta_*^{i}, Q^{i}, R^{i})\geq \sigma_{\omega}^2 \|P(\hat{\Theta}^i_{\tau_k}, Q^i, R^i)\|_* \label{ineqP_tFnorm}
\end{align}

On the other hand, by (\ref{eq:Lemma8_res22}) we have

\begin{align}
\nonumber & Q^i+K^\top (\hat{\Theta}^i_{\tau_k}, Q^i, R^i)R^i K(\hat{\Theta}^i_{\tau_k}, Q^i, R^i) \preceq  \\
\nonumber &P(\hat{\Theta}^i_{\tau_k}, Q^i, R^i)-\\
\nonumber &(A_*+B_*^iK(\hat{\Theta}^i_{\tau_k}, Q^i, R^i))^\top P(\hat{\Theta}^i_{\tau_k}, Q^i, R^i)\times \\ 
\nonumber &(A_*+B_*^iK(\hat{\Theta}^i_{\tau_k}, Q^i, R^i))+\\
\nonumber &2\mu_i\|P(\hat{\Theta}^i_{\tau_k}, Q^i, R^i)\|_*\times \\
&\begin{pmatrix} I \\ K(\hat{\Theta}^i_{\tau_k}, Q^i, R^i) \end{pmatrix}{V^i}^{-1}_{n_i(\tau_k)}\begin{pmatrix} I \\ K(\hat{\Theta}^i_{\tau_k}, Q^i, R^i) \end{pmatrix}^\top.
\end{align}

Combining the above inequalities results in 
\begin{align}
 \nonumber&\mathcal{R}_{\tau_k:\mathfrak{T}_k}\leq  \mathcal{R}_{\tau_k:\mathfrak{T}_k}^{11}+\mathcal{R}_{\tau_k:\mathfrak{T}_k}^{12}+\mathcal{R}_{\tau_k:\mathfrak{T}_k}^{13}+\mathcal{R}_{\tau_k:\mathfrak{T}_k}^{14}\\
 \nonumber &\mathcal{R}_{\tau_k:\mathfrak{T}_k}^{11}=\sum_{t=\tau_k}^{\mathfrak{T}_k-1} \bigg(x_t^\top P(\hat{\Theta}^{i_k}_{\tau_k}, Q^{i_k}, R^{i_k})x_t-\\
 \nonumber & \quad \quad\quad \quad\quad \quad \quad x_{t+1}^\top P(\hat{\Theta}^{i_k}_{\tau_k}, Q^{i_k}, R^{i_k})x_{t+1}\bigg) 1_{\mathcal{E}_t}\\
 \nonumber & \mathcal{R}_{\tau_k:\mathfrak{T}_k}^{12}=\sum_{t=\tau_k}^{\mathfrak{T}_k-1} \big(\omega_{t+1}^\top P(\hat{\Theta}^{i_k}_{\tau_k}, Q^{i_k}, R^{i_k})x_t \big)1_{\mathcal{E}_t}\\
 \nonumber & \mathcal{R}_{\tau_k:\mathfrak{T}_k}^{13}=\sum_{t=\tau_k}^{\mathfrak{T}_k-1} \bigg(\omega_{t+1}^\top P(\hat{\Theta}^{i_k}_{\tau_k}, Q^{i_k}, R^{i_k})\omega_{t+1}-\\
 \nonumber &\quad \quad\quad \quad\quad \quad \quad \sigma_{\omega}^2 \|P(\hat{\Theta}^i_{\tau_k}, Q^{i_k}, R^{i_k})\|_*\bigg)1_{\mathcal{E}_t}\\
 & \mathcal{R}_{\tau_k:\mathfrak{T}_k}^{14}=\sum_{t=\tau_k}^{\mathfrak{T}_k-1}\frac{4\nu_{i_k}\mu_{i_k}}{\sigma_{\omega}^2}\big(z_t^\top {V_{n_i(\tau_k)}^{i_k}}^{-1}z_t\big)1_{\mathcal{E}_t}.
\end{align}

It is worthy to note that in the terms mentioned above, we do not specify the mode index as it is known from the context,i.e., for an epoch starting at $\tau_k$ the corresponding mode is $i_k\in \mathcal{I}:=\{i_0, i_1, ..., i_{n-1}\}$.

Furthermore, recalling $R_{\mathcal{A}}^2$ definition we can write

\begin{align}
 \nonumber &\sum _{t=\mathfrak{T}_k}^{\tau_{k+1}-1} c_{t}^{i_k}=\\
\nonumber  &\sum _{t=\mathfrak{T}_k}^{\tau_{k+1}-1} \big( Q^{i_k}+K^\top (\hat{\Theta}^{i_k}_{\tau_k}, Q^{i_k}, R^{i_k})R^{i_k} K(\hat{\Theta}^{i_k}_{\tau_k}, Q^{i_k}, R^{i_k})\big)\leq \\
\nonumber & \mathcal{R}^{21}_{\mathfrak{T}_k: \tau_{k+1}}+\mathcal{R}^{22}_{\mathfrak{T}_k: \tau_{k+1}}\\
 \nonumber &\mathcal{R}^{21}_{\mathfrak{T}_k: \tau_{k+1}}=\sum _{t=\mathfrak{T}_k}^{\tau_{k+1}-1} \bigg (x_t^\top P(\hat{\Theta}^{i_k}_{\tau_k}, Q^{i_k}, R^{i_k}) x_t-\\
 \nonumber & x_t^\top (A_*^{i_k}+B_*^{i_k}K(\hat{\Theta}^{i_k}_{\tau_q}, Q^{i_k}, R^{i_k}))^\top P(\hat{\Theta}^i_{\tau_q}, Q^{i_k}, R^{i_k})\times \\&(A_*^{i_k}+B_*^iK(\hat{\Theta}^{i_k}_{\tau_q}, Q^{i_k}, R^{i_k}))  x_t\bigg) 1_{\mathcal{E}_t}\\
\nonumber & \mathcal{R}^{22}_{\mathfrak{T}_k: \tau_{k+1}}=\sum _{t=\mathfrak{T}_k}^{\tau_{k+1}-1} 2{\mu}_{i_k}^{\tau_k}\|P(\hat{\Theta}^{i_k}_{\tau_q}, Q^{i_k}, R^{i_k})\|_*\\
&\begin{pmatrix} I \\ K(\hat{\Theta}^{i_k}_{\tau_k}, Q^{i_k}, R^{i_k}) \end{pmatrix}{V^{i_k}}^{-1}_{n_i(\tau_k)}\begin{pmatrix} I \\ K(\hat{\Theta}^{i_k}_{\tau_k}, Q^{i_k}, R^{i_k}) \end{pmatrix}^\top 1_{\mathcal{E}_t}.
\end{align}

in which we applied (\ref{eq:Lemma8_res22}).

By this decomposition, now the regret (\ref{eq:RegretSwitchDec}) is upper-bounded as follows:

\begin{align}
	\nonumber R_T \leq & \mathbb {E} [\sum _{k=0}^{ns-1}  \mathcal{R}_{\tau_k:\mathfrak{T}_k}^{11}+\mathcal{R}_{\tau_k:\mathfrak{T}_k}^{12}+\mathcal{R}_{\tau_k:\mathfrak{T}_k}^{13}+\mathcal{R}_{\tau_k:\mathfrak{T}_k}^{14}]+\\
 \nonumber &\mathbb {E} [\sum _{k=0}^{ns-1}\mathcal{R}^{21}_{\mathfrak{T}_k: \tau_{k+1}}+\mathcal{R}^{22}_{\mathfrak{T}_k: \tau_{k+1}}]
\end{align}

Now we need to bound each term individually. First we show that $\mathbb {E} [\sum _{k=0}^{ns-1} \mathcal{R}^{12}_{\tau_k: \mathfrak{T}_k}]=0$ and $\mathbb {E} [\sum _{k=0}^{ns-1} \mathcal{R}^{13}_{\tau_k: \mathfrak{T}_k}]=0$. Lemmas provides for these claims.

\begin{lemma}
On event $\mathcal{E}_t$, 

\begin{align}
   \mathbb{E}[\sum_{k=0}^{ns-1}\mathcal{R}_{s_1:s_2}^2]=0
\end{align}
holds true.
\end{lemma}
\begin{proof}
\begin{align}
    \nonumber &\mathbb{E}[\sum_{k=0}^{ns-1}\sum_{t=s_1}^{s_2-1} \big(\omega_{t+1}^\top P(\hat{\Theta}^{i_k}_{s_1}, Q^{i_k}, R^{i_k})x_t \big)1_{\mathcal{E}_t}]=\\
   \nonumber  &\sum_{k=0}^{ns-1}\mathbb{E}[\sum_{t=s_1}^{s_2-1} \big(\omega_{t+1}^\top P(\hat{\Theta}^{i_k}_{s_1}, Q^{i_k}, R^{i_k})x_t \big)1_{\mathcal{E}_t}|\mathcal{F}_{k-1}]=\\
    \nonumber &\sum_{k=0}^{ns-1}P(\hat{\Theta}^{i_k}_{s_1}, Q^{i_k}, R^{i_k})\bullet \mathbb{E}[\sum_{t=s_1}^{s_2-1} \big(\omega_{t+1}^\top x_t \big)1_{\mathcal{E}_t}|\mathcal{F}_{k-1}]=\\
    &\sum_{k=0}^{ns-1}P(\hat{\Theta}^{i_k}_{s_1}, Q^{i_k}, R^{i_k})\bullet \sum_{t=s_1}^{s_2-1} \mathbb{E}[\big(\omega_{t+1}^\top x_t \big)1_{\mathcal{E}_t}|\mathcal{F}_{t-1}]=0
\end{align}
Note that the second equality holds because $P(\hat{\Theta}^{i_k}_{s_1}, Q^{i_k}, R^{i_k})$ is $\mathcal{F}_{k-1}$ measurable. The last equality holds because $x_t$ and $\omega_{t+1}$ are independent and $\omega_{t+1}$ is martingale difference sequence, i.e., $\mathbb[\omega_{s+1}|\mathcal{F}_s]=0$ for all $s=0,1,...,t$.
\end{proof}

\begin{lemma}
    It holds that
   \begin{align}
   \mathbb{E}[\sum_{k=0}^{ns-1}\mathcal{R}_{s_1:s_2}^3]=0
\end{align}
\end{lemma}
\begin{proof}
\begin{align}
   \nonumber &\mathbb{E}[\sum_{k=0}^{ns-1}\sum_{t=s_1}^{s_2-1} \big(\omega_{t+1}^\top P(\hat{\Theta}^{i_k}_{s_1}, Q^{i_k}, R^{i_k})\omega_{t+1}-\\
  \nonumber & \quad \sigma_{\omega}^2 \|P(\hat{\Theta}^{i_k}_{s_1}, Q^{i_k}, R^{i_k})\|_*\big)1_{\mathcal{E}_t}]=\\
    \nonumber &\sum_{k=0}^{ns-1}\mathbb{E}[\sum_{t=s_1}^{s_2-1} \big(\omega_{t+1}^\top P(\hat{\Theta}^{i_k}_{s_1}, Q^{i_k}, R^{i_k})\omega_{t+1}-\\
    \nonumber &\quad  \quad \quad \sigma_{\omega}^2 \|P(\hat{\Theta}^{i_k}_{s_1}, Q^{i_k}, R^{i_k})\|_*\big)1_{\mathcal{E}_t}| \mathcal{F}_{k-1}]=\\
     \nonumber &\sum_{k=0}^{ns-1}\bigg(\mathbb{E}[\sum_{t=s_1}^{s_2-1} \big(\omega_{t+1}^\top P(\hat{\Theta}^{i_k}_{s_1}, Q^{i_k}, R^{i_k})\omega_{t+1}\big)1_{\mathcal{E}_t}| \mathcal{F}_{k-1}]-\\
     \nonumber &\quad \quad  \quad\mathbb{E}[\sum_{t=s_1}^{s_2-1}\big(\sigma_{\omega}^2 \|P(\hat{\Theta}^{i_k}_{s_1}, Q^{i_k}, R^{i_k})\|_*\big)1_{\mathcal{E}_t}| \mathcal{F}_{k-1}]\bigg)=\\ 
      \nonumber &\sum_{k=0}^{ns-1}\bigg(P(\hat{\Theta}^{i_k}_{s_1}, Q^{i_k}, R^{i_k})\bullet\mathbb{E}[\sum_{t=s_1}^{s_2-1} \big(\omega_{t+1}^\top \omega_{t+1}\big)1_{\mathcal{E}_t}| \mathcal{F}_{k-1}]-\\
     \nonumber  &\quad \quad  \quad \mathbb{E}[\sum_{t=s_1}^{s_2-1}\big(\sigma_{\omega}^2 \|P(\hat{\Theta}^{i_k}_{s_1}, Q^{i_k}, R^{i_k})\|_*\big)1_{\mathcal{E}_t}| \mathcal{F}_{k-1}]\bigg)\leq\\ 
     &\nonumber \sum_{k=0}^{ns-1} \|P(\hat{\Theta}^{i_k}_{s_1}, Q^{i_k}, R^{i_k})\|_*\\
     &\quad\quad \sum_{t=s_1}^{s_2-1} \mathbb{E}[\big(\omega_{t+1}^\top \omega_{t+1}-\sigma_{\omega}^2I\big)1_{\mathcal{E}_t}| \mathcal{F}_{t-1}]=0
\end{align}
   where, we applied $\mathcal{F}_{k-1}$ measurablity of $P(\hat{\Theta}^{i_k}_{s_1}, Q^{i_k}, R^{i_k})$ in the last equality and in the last inequality we applied the property $A\bullet B\leq \|A\|_*\|B\|_*$.
\end{proof}

\begin{lemma}
    It holds that
    \begin{align}
    \mathbb{E}[\sum_{k=0}^{ns-1}\mathcal{R}_{s_1:s_2}^1]\leq \frac{\nu_*}{\sigma_{\omega}^2}\mathbb{E}[x_0^\top x_0]+\frac{\nu_*}{\sigma_{\omega}^2}\kappa_*^4
    \end{align}
% where $X_{\tau^{k-1}_{md}}:=\max_{0\leq s\leq \tau^{k-1}_{md}} \mathbb[x^\top_{\tau^{k-1}_{md}}x_{\tau^{k-1}_{md}}]$ which is simply obtained by (\ref{stateboundsw}).   
\end{lemma}
\begin{proof}
   We have
   \begin{align}
     \nonumber & \mathcal{R}_{s_1:s_2}^1=\sum_{t=s_1}^{s_2-1} \bigg(x_t^\top P(\hat{\Theta}^i_{s_1}, Q^i, R^i)x_t-\\
    \nonumber  &\quad \quad \quad \quad \quad \quad \quad x_{t+1}^\top P(\hat{\Theta}^i_{s_1}, Q^i, R^i)x_{t+1}\bigg) 1_{\mathcal{E}_t}\\
     \nonumber  & \quad  \quad \quad=\bigg[x_{s_1}^\top P(\hat{\Theta}^i_{s_1}, Q^i, R^i)x_{s_1}-x_{s_2}^\top P(\hat{\Theta}^i_{s_1}, Q^i, R^i)x_{s_2}+\\
      &\nonumber  \quad \sum_{t=s_1+1}^{s_2-2}  \big(x_t^\top P(\hat{\Theta}^i_{s_1}, Q^i, R^i)x_t- x_{t}^\top P(\hat{\Theta}^i_{s_1}, Q^i, R^i)x_{t}\big)\bigg] 1_{\mathcal{E}_t}\\
     & \quad \quad \quad =x_{s_1}^\top P(\hat{\Theta}^i_{s_1}, Q^i, R^i)x_{s_1}-x_{s_2}^\top P(\hat{\Theta}^i_{s_1}, Q^i, R^i)x_{s_2} \label{eq:Telesc}
   \end{align}

by telescoping. By applying (\ref{eq:Telesc}) and taking into account the definitions $\mathfrak{T}_k$ and $\tau_k$ it implies

\begin{align}
    \nonumber &\mathbb{E}[\sum_{k=0}^{ns-1}\mathcal{R}_{\tau_k:\mathfrak{T}_k}^1]=\\
     \nonumber & \sum_{k=0}^{ns-1} \bigg(-\mathbb{E}[x_{\mathfrak{T}_k}^\top P(\hat{\Theta}^{i_{k-1}}_{\tau_k}, Q^{i_{k-1}}, R^{i_{k-1}}) x_{\mathfrak{T}_k}|\mathcal{F}_{\mathfrak{T}_k-1}]\\
     &\quad \quad \quad +\mathbb{E}[x_{\tau_k}^\top P(\hat{\Theta}^{i_k}_{\tau_k}, Q^{i_k}, R^{i_k}) x_{\tau_k}|\mathcal{F}_{\tau_k-1}]\bigg)
\end{align}

% \begin{align}
%     \nonumber &\mathbb{E}[\sum_{k=0}^{ns-1}\mathcal{R}_{\tau_k:\mathfrak{T}_k}^1]=\mathbb{E}[x_{0}^\top P(\hat{\Theta}^{i_0}_{\tau_0}, Q^{i_0}, R^{i_0}) x_{0}|\mathcal{F}_{0^-}]-\\
%     \nonumber &\mathbb{E}[ x_{\mathfrak{T}_{ns-1}}^\top P(\hat{\Theta}^{i_{ns-1}}_{\mathfrak{T}_{ns-1}}, Q^{i_{ns-1}}, R^{i_{ns-1}}) x_{\mathfrak{T}_{ns-1}}|\mathcal{F}_{\mathfrak{T}_{ns-1}-1}]+\\
%      \nonumber & \sum_{k=1}^{ns-2} \bigg(-\mathbb{E}[x_{\mathfrak{T}_k}^\top P(\hat{\Theta}^{i_{k-1}}_{\tau_k}, Q^{i_{k-1}}, R^{i_{k-1}}) x_{\mathfrak{T}_k}|\mathcal{F}_{\mathfrak{T}_k-1}]\\
%      &\quad \quad \quad +\mathbb{E}[x_{\tau_k}^\top P(\hat{\Theta}^{i_k}_{\tau_k}, Q^{i_k}, R^{i_k}) x_{\tau_k}|\mathcal{F}_{\tau_k-1}]\bigg)
% \end{align}

\textbf{Case A: Contraction Case $\mathcal{Y}<1$}

For this case we, by applying  (\ref{eq:moghayese}) it yields 

\begin{align}
    \nonumber& \mathbb{E}[\sum_{k=0}^{ns-1}\mathcal{R}_{\tau_k:\mathfrak{T}_k}^1]\leq \sum_{k=0}^{ns-1} \mathbb{E}[x_{\tau_k}^\top P(\hat{\Theta}^{i_k}_{\tau_k}, Q^{i_k}, R^{i_k}) x_{\tau_k}|\mathcal{F}_{\tau_k-1}]\\
    & \leq \frac{\nu_*}{\sigma_{\omega}^2}  \sum_{k=0}^{ns-1} \mathbb{E}[x_{\tau_k}^\top x_{\tau_k}|\mathcal{F}_{\tau_k-1}]
\end{align}

\begin{align}
    \mathbb{E}[\sum_{k=0}^{ns-1}\mathcal{R}_{s_1:s_2}^1]
\end{align}

By (\ref{eq:UppExV}) we have 
	\begin{align}
	\nonumber &\zeta (\mathcal{C}^i_{\tau_q}(\delta), \mathcal{C}^j_{\tau_{q}}(\delta))W\geq \mathbb{E}[\mathcal{V}_j(\tau_{q+1}^+)| \mathcal{F}_{\tau_{q+1}-1}]-\\
 \nonumber &\rho (\mathcal{C}^i_{\tau_q}(\delta), \mathcal{C}^j_{\tau_{q}}(\delta))\bigg(1-\eta \big(\mathcal{C}^i_{\tau_q}(\delta)\big)\bigg)^{\tau_{q+1}-\tau_q}\mathbb{E}[\mathcal{V}_i(\tau_{q}^+)|\mathcal{F}_{\tau_{q}-1} ]
	\end{align}
Based on the dwell-time design the following inequality

\begin{align*}
    \rho (\mathcal{C}^i_{\tau_q}(\delta), \mathcal{C}^j_{\tau_{q}}(\delta))\bigg(1-\eta \big(\mathcal{C}^i_{\tau_q}(\delta)\big)\bigg)^{\tau_{q+1}-\tau_q}\leq 1
\end{align*}
holds if $\tau_{q+1}-\tau_q\geq [\tau_{dw}^{ij}]_{\tau_q}$. Therefore for any $\tau< [\tau_{dw}^{ij}]_{\tau_q}$ 

\begin{align}
 \nonumber &\zeta (\mathcal{C}^i_{\tau_q}(\delta), \mathcal{C}^j_{\tau_{q}}(\delta))W \geq \mathbb{E}[\mathcal{V}_j(\tau_{q+1}^+)| \mathcal{F}_{\tau_{q+1}-1}]-\\
 \nonumber &\rho (\mathcal{C}^i_{\tau_q}(\delta), \mathcal{C}^j_{\tau_{q}}(\delta))\bigg(1-\eta \big(\mathcal{C}^i_{\tau_q}(\delta)\big)\bigg)^{ [\tau_{dw}^{ij}]_{\tau_q}}\mathbb{E}[\mathcal{V}_i(\tau_{q}^+)|\mathcal{F}_{\tau_{q}-1} ]\geq\\
 \nonumber & \mathbb{E}[\mathcal{V}_j(\tau_{q+1}^+)| \mathcal{F}_{\tau_{q+1}-1}]-\\
 &\rho (\mathcal{C}^i_{\tau_q}(\delta), \mathcal{C}^j_{\tau_{q}}(\delta))\bigg(1-\eta \big(\mathcal{C}^i_{\tau_q}(\delta)\big)\bigg)^{\tau}\mathbb{E}[\mathcal{V}_i(\tau_{q}^+)|\mathcal{F}_{\tau_{q}-1}]  \label{eq:mohem2}
\end{align}

Also, by definition we have

\begin{align*}
   \nonumber &\mathbb{E}[\mathcal{V}_{i_{k-1}}(s_2^{k-1})|\mathcal{F}_{s_2^{k-1}-1}]:=\\
   &\mathbb{E}[P(\hat{\Theta}^{i_{k-1}}_{s^{k-1}_1}, Q^{i_{k-1}}, R^{i_{k-1}})\bullet [x_{s^{k-1}_2}^\top x_{s^{k-1}_2}]|\mathcal{F}_{s^{k-1}_2-1}]  
\end{align*}
and similarly 
\begin{align*}
   \mathbb{E}[\mathcal{V}_{i_k}(s_1^k)|\mathcal{F}_{s_1^k-1}]:=\mathbb{E}[P(\hat{\Theta}^{i_k}_{s_1}, Q^{i_k}, R^{i_k})\bullet [x_{s^k_1}^\top x_{s^k_1}]|\mathcal{F}_{s_1^k-1}].
\end{align*}

Now noting that $s_2^k-s_1^k<[\tau_{dw}^{ij}]_{\tau_q}$ (\ref{eq:mohem2}) implies
\begin{align}
   \nonumber  &-\mathbb{E}[P(\hat{\Theta}^{i_{k-1}}_{s^{k-1}_1}, Q^{i_{k-1}}, R^{i_{k-1}})\bullet [x_{s^{k-1}_2}^\top x_{s^{k-1}_2}]|\mathcal{F}_{s^{k-1}_2-1}]+\\
    &\mathbb{E}[P(\hat{\Theta}^{i_k}_{s_1}, Q^{i_k}, R^{i_k})\bullet [x_{s^k_1}^\top x_{s^k_1}]|\mathcal{F}_{s^k-1}]\leq \zeta (\mathcal{C}^i_{\tau_q}(\delta), \mathcal{C}^j_{\tau_{q}}(\delta))W \label{eq:mohem3}
\end{align}
By summing up the right hand side of (\ref{eq:mohem3}) over the switch sequence we have 

\begin{align}
    \zeta (\mathcal{C}^i_{\tau_q}(\delta), \mathcal{C}^j_{\tau_{q}}(\delta))W \leq \frac{\nu_*}{\sigma_{\omega}^2}\kappa_*^4
\end{align}
 which follows by similar analysis, carried our for Lemma \ref{thm:res_sequential_stability_resp}. Now it is straight forward to see

 \begin{align}
     \mathbb{E}[P(\hat{\Theta}^{i_0}_{s_1}, Q^{i_0}, R^{i_0})\bullet [x_{0}^\top x_{0}]|\mathcal{F}_{0^-}]\leq \frac{\nu_*}{\sigma_{\omega}^2}  \mathbb{E}[x_{0}^\top x_{0}|\mathcal{F}_{0^-}]
 \end{align}
 that completes the proof.

\end{proof}

To bound the term $\mathbb{E}[\sum_{k=0}^{ns-1}\mathcal{R}_{s_1:s_2}^4]$ we need the following lemma.
\begin{lemma} 
\label{lem:usefulforbound}
    Let the matrix $M$ be positive definite and let $z_t$ be a $\mathcal{F}_t$-measurable vector, if $\sum_{t=s_1}^{s_2}\mathbb{E}[z_t^\top M^{-1}z_t| \mathcal{F}_{t-1}]\leq 1$ then
      \begin{align}
       \nonumber &\sum_{t=s_1}^{s_2}\mathbb{E}[z_t^\top M^{-1}z_t| \mathcal{F}_{t-1}]\leq \\
       &2\log \bigg(\frac{\det \big(M+\sum_{t=s_1}^{s_2}\mathbb{E}[z_s z_s^\top | \mathcal{F}_{t-1}]\big)}{\det(M)}\bigg) 
     \end{align}
\end{lemma}
\begin{proof}
    Following statement holds true by determinant properties
\begin{align*}
    &\det \big(M+\sum_{t=s_1}^{s_2}\mathbb{E}[z_s z_s^\top | \mathcal{F}_{t-1}]\big)=\\
    &\det(M)det \big(I+M^{-1/2}\sum_{t=s_1}^{s_2} \mathbb{E}[z_tz_t^\top| \mathcal{F}_{t-1}]M^{-1/2}\big)=\\
    &\det(M)(1+\sum_{t=s_1}^{s_2}\mathbb{E}[z_t^\top M^{-1}z_t| \mathcal{F}_{t-1}])
\end{align*}
which results in
\begin{align*}
   &\log (1+\sum_{t=s_1}^{s_2}\mathbb{E}[z_t^\top M^{-1}z_t| \mathcal{F}_{t-1}])=\\
   &\log \bigg(\frac{\det \big(M+\sum_{t=s_1}^{s_2}\mathbb{E}[z_t z_t^\top | \mathcal{F}_{t-1}]\big)}{\det(M)}\bigg)
\end{align*}

If $\sum_{t=s_1}^{s_2}\mathbb{E}[z_t^\top M^{-1}z_t| \mathcal{F}_{t-1}]\leq 1$ then applying the inequality $x\leq 2\log (1+x)$ which holds true for $0\leq x\leq 1$, one can write:
\begin{align*}
   &\sum_{t=s_1}^{s_2}\mathbb{E}[z_t^\top M^{-1}z_t| \mathcal{F}_{t-1}]\leq \\
   &2\log \bigg(\frac{\det \big(M+\sum_{t=s_1}^{s_2}\mathbb{E}[z_t z_t^\top | \mathcal{F}_{t-1}]\big)}{\det(M)}\bigg)
\end{align*}

\end{proof}

\begin{lemma}
    On event $\mathcal{E}_t$ it holds 
 \begin{align}
   \mathbb{E}[\sum_{k=0}^{ns-1}\mathcal{R}_{s_1:s_2}^4]
\end{align}
\end{lemma}
\begin{proof}
Recall the definition of $\bar{\mu}^{\tau_q}_i$ given by (\ref{eq:mofid}) where $\tau_q$ denotes switch time. For the sake of simplicity in notation, we slightly abuse the notation and let the parameter to be defined $\bar{\mu}^{k}_i$ where $k$ is index of the $k$-th switch. Since the algorithm is off-policy and the control is designed by confidence ellipsoid, constructed by data gathered before starting current epoch, then $\mathbb{E}[\bar{\mu}_i^{k^-}|\mathcal{F}_{k^{-}}]=\mathbb{E}[\bar{\mu}_i^{k^-}|\mathcal{F}_{k-1}]=\bar{\mu}_i^{k^{-}}$. The second equality holds because the last update-time of $\bar{\mu}^k_i$ is at least two epochs back. Now, one can write
\begin{align}
  \nonumber & \mathbb{E}[\sum_{k=0}^{ns-1}\mathcal{R}_{s_1:s_2}^4]=\\
  \nonumber &\mathbb{E}[\sum_{k=0}^{ns-1}\sum_{t=s_1}^{s_2-1}\frac{4\nu_{i_k}\bar{\mu}^{k^-}_{i_k}}{\sigma_{\omega}^2}\big(z_t^\top {V^{i_k}}_{n_{i_k}(s_1)}^{-1}z_t\big)1_{\mathcal{E}_t}]=\\
   \nonumber & \mathbb{E}[\mathbb{E}[\sum_{k=0}^{ns-1}\sum_{t=s_1}^{s_2-1}\frac{4\nu_{i_k}\bar{\mu}^{k^-}_{i_k}}{\sigma_{\omega}^2}\big(z_t^\top {V^{i_k}}_{n_{i_k}(s_1)}^{-1}z_t\big)1_{\mathcal{E}_t}]|\bar{\mu}_i^{k^-}]=\\
   \nonumber & \mathbb{E}[\sum_{k=0}^{ns-1}\frac{4\nu_{i_k}\mu^{k^-}_{i_k}}{\sigma_{\omega}^2}{V^{i_k}}_{n_{i_k}(s_1)}^{-1}]\bullet\\
   \nonumber &\quad \quad\quad \mathbb{E}[\sum_{t=s_1}^{s_2-1}\big(z_t^\top z_t\big)1_{\mathcal{E}_t}| \mathcal{F}_{k-1}]|\bar{\mu}_{i_k}^{k^-}]=\\
   \nonumber & \sum_{k=0}^{ns-1}\frac{4\nu_{i_k}\mathbb{E}[\bar{\mu}^{k^-}_{i_k}]}{\sigma_{\omega}^2}\mathbb{E}[{V^{i_k}}_{n_{i_k}(s_1)}^{-1}]\bullet\mathbb{E}[\sum_{t=s_1}^{s_2-1}\big(z_t^\top z_t\big)1_{\mathcal{E}_t}| \mathcal{F}_{t-1}]=\\
  & \sum_{k=0}^{ns-1}\frac{4\nu_{i_k}\mathbb{E}[\bar{\mu}^{k^-}_{i_k}]}{\sigma_{\omega}^2}\sum_{t=s_1}^{s_2-1}\mathbb{E}\big[z_t^\top \mathbb{E}[{V^{i_k}}_{n_{i_k}(s_1)}^{-1}]z_t1_{\mathcal{E}_t}| \mathcal{F}_{t-1}\big] \label{eq:lastineq}
\end{align}
By  ${V^{i_k}}_{n_{i_k}(s_1)}\succsim \lambda^k_{i_k} I$, we have
\begin{align*}
    &\sum_{t=s_1}^{s_2-1}\mathbb{E}\big[z_t^\top \mathbb{E}[{V^{i_k}}_{n_{i_k}(s_1)}^{-1}]z_t1_{\mathcal{E}_t}| \mathcal{F}_{t-1}\big]\leq\\
    &\frac{1}{\lambda^k_{i_k}} \sum_{t=s_1}^{s_2-1}\mathbb{E}\big[z_t^\top z_t1_{\mathcal{E}_t}| \mathcal{F}_{t-1}\big].
\end{align*}
Considering the fact that $\mathbb{E}[z_t^\top z_t\ \mathcal{F}_{t-1}]$ is bounded for $s_1\leq t\leq s_2-1$ then  $\frac{1}{\lambda^k_{i_k}} \sum_{t=s_1}^{s_2-1}\mathbb{E}\big[z_t^\top z_t1_{\mathcal{E}_t}| \mathcal{F}_{t-1}\big]$ is $\mathcal{O}(s_2-s_1)$. However, since $\lambda^k_{i_k}$ is of $\mathcal{O}(\sqrt{n_{i_k}(s_1)})$ by definition, with long enough switch sequence there exists $k^{\prime}-$switch with corresponding time time $t^{\prime}$ such that $\mathcal{O}\big(n_i(t^{\prime})\big)=\mathcal{O}\big((\tau_{\max}^{d\omega})^2\big)$ and for $k\geq k^{\prime}$ the following inequality holds

\begin{align}
   \frac{1}{\lambda^k_{i_k}} \sum_{t=s_1}^{s_2-1}\mathbb{E}\big[z_t^\top z_t1_{\mathcal{E}_t}| \mathcal{F}_{t-1}\big]\leq 1 \label{eq:cond2}
\end{align}
where $s_1\geq t^{\prime}$. The inequality (\ref{eq:cond2}) in deed fulfills the condition for Lemma \ref{lem:usefulforbound}. Therefore for all modes $i\in \mathcal{I}$ and $s_1\geq t^{\prime}$ 

\begin{align}
    \nonumber &\sum_{t=s_1}^{s_2-1}\mathbb{E}\big[z_t^\top \mathbb{E}[{V^i}_{n_i(s_1)}^{-1}]z_t1_{\mathcal{E}_t}| \mathcal{F}_{t-1}\big]\leq \\
    &2\log \bigg(\frac{\det \big(\mathbb{E}[{V^i}_{n_i(s_2)}]\big)}{\det\big(\mathbb{E}[{V^i}_{n_i(s_1)}]\big)}\bigg) \label{eq:upLog1}
\end{align}

holds where
\begin{align*}
    \mathbb{E}[{V^i}_{n_i(s_2)}]=\mathbb{E}[{V^i}_{n_i(s_1)}]+\sum_{t=s_1}^{s_2-1}\mathbb{E}\big[z_t^\top z_t1_{\mathcal{E}_t}| \mathcal{F}_{t-1}\big].
\end{align*}
 The inequality (\ref{eq:upLog1}) implies that for $t>t^{\prime}$ the left hand is not $\mathcal{O}(s_2-s_1)$ anymore which assures achieving better regret. 

Now, we proceed with decomposing (\ref{eq:lastineq}) as follows:

\begin{align}
   \nonumber & \sum_{k=0}^{ns-1}\frac{4\nu_{i_k}\mathbb{E}[\bar{\mu}^{k^-}_{i_k}]}{\sigma_{\omega}^2}\sum_{t=s_1}^{s_2-1}\mathbb{E}\big[z_t^\top \mathbb{E}[{V^{i_k}}_{n_{i_k}(s_1)}^{-1}]z_t1_{\mathcal{E}_t}| \mathcal{F}_{t-1}\big]=\\
\nonumber &\overbrace{\sum_{k=0}^{m^{\prime}-1}\frac{4\nu_{i_k}\mathbb{E}[\bar{\mu}^{k^-}_{i_k}]}{\sigma_{\omega}^2}\sum_{t=s_1}^{s_2-1}\mathbb{E}\big[z_t^\top \mathbb{E}[{V^{i_k}}_{n_{i_k}(s_1)}^{-1}]z_t1_{\mathcal{E}_t}| \mathcal{F}_{t-1}\big]}^{\Gamma_1}+\\
   & \underbrace{\sum_{k=m^{\prime}}^{ns-1}\frac{4\nu_{i_k}\mathbb{E}[\bar{\mu}^{k^-}_{i_k}]}{\sigma_{\omega}^2}\sum_{t=s_1}^{s_2-1}\mathbb{E}\big[z_t^\top \mathbb{E}[{V^{i_k}}_{n_{i_k}(s_1)}^{-1}]z_t1_{\mathcal{E}_t}| \mathcal{F}_{t-1}\big]}_{\Gamma_2}
\end{align}

We upper bound the terms individually. For the term $\Gamma_1$ we have 

\begin{align*}
    &\Gamma_1\leq \sum_{k=0}^{k^{\prime}-1}\alpha_0^{i_k}\sum_{t=s_1}^{s_2-1}\mathbb{E}\big[z_t^\top z_t1_{\mathcal{E}_t}| \mathcal{F}_{t-1}\big]\leq \\
& \sum_{k=0}^{k^{\prime}-1}2\alpha_0^{i_k}\kappa^2_{i_k}\sum_{t=s_1}^{s_2-1}\mathbb{E}\big[x_t^\top x_t1_{\mathcal{E}_t}| \mathcal{F}_{t-1}\big]\leq \\
&\sum_{k=0}^{k^{\prime}-1}2\alpha_0^{i_k}\kappa^2_{i_k}\alpha (s_2-s_1)= \\
    & \sum_{i\in \mathcal{I}}2\alpha_0^{i} \kappa^2_i Xn_i(t^{\prime})\leq 2|\mathcal{M}| X \max_{i\in \mathcal{M}}\alpha_0^{i} \kappa^2_i n_i(t^{\prime})
\end{align*}

where in the first inequality we applied ${V^{i_k}}_{n_{i_k}(s_1)}^{-1}\leq 1/\lambda_{i^k}^{n_{i_k}}I$ and the definition of $\lambda^{i^k}_{n_{i_k}}$. In the second inequality, we used the fact that on the event $\mathcal{E}_t$, $z_t=\begin{pmatrix} I \\ K(\hat{\Theta}_{\tau_q}^i, Q^i, R^i) \end{pmatrix}x_t$ and  $\|\begin{pmatrix} I \\ K(\hat{\Theta}_{\tau_q}^i, Q^i, R^i) \end{pmatrix}\|^2\leq2{\kappa_i}^2$ since $\kappa_i\geq 1$, i.e., 

\begin{align}
  \mathbb{E}\big[z_t^\top z_t1_{\mathcal{E}_t}| \mathcal{F}_{t-1}\big]\leq 2\kappa^2_i \mathbb{E}\big[x_t^\top x_t| \mathcal{F}_{t-1}\big]
\end{align}

In the third inequality we applied Lemma \ref{thm:res_sequential_stability_resp} with $\alpha$ defined as follows:

 \begin{align}
	\alpha:=\kappa_*^4E[x_0^\top x_0|\mathcal{F}_{0^-} ]+\kappa_*^6W.
\end{align}

The equality holds because  
\begin{align}
 \sum_{k=0}^{k^{\prime}-1} (s_2-s_1)=\sum_{i\in \mathcal{I}} n_i(t^{\prime})
\end{align}
by definition.

Now we proceed to upper-bound $\Gamma_2$. For $k\geq k^{\prime}$ thanks to the fulfilment of (\ref{eq:cond2}), which is condition of Lemma \ref{lem:usefulforbound}, it yields

\begin{align}
 \nonumber  \Gamma_2&\leq \sum_{k=k^{\prime}}^{ns-1}\frac{8\nu_{i_k}\mathbb{E}[\mu^{k^-}_{i_k}]}{\sigma_{\omega}^2}\log \bigg(\frac{\det \big(\mathbb{E}[{V^{i_k}}_{n_{i_k}(s_2)}]\big)}{\det\big(\mathbb{E}[{V^{i_k}}_{n_{i_k}(s_1)}]\big)}\bigg) \\
  &\leq\frac{8\nu_{*}\mathbb{E}[\mu^{ns^-}_{*}]}{\sigma_{\omega}^2} \sum_{k=k^{\prime}}^{ns-1}\log \bigg(\frac{\det \big(\mathbb{E}[{V^{i_k}}_{n_{i_k}(s_2)}]\big)}{\det\big(\mathbb{E}[{V^{i_k}}_{n_{i_k}(s_1)}]\big)}\bigg)
\end{align}
where $\mathbb{E}[\mu^{ns^-}_{*}]=\max_{i\in \mathcal{I}}\mathbb{E}[\mu^{ns^-}_{i}]$, which is given by (\ref{eq:boundExpMu}).

Moreover,

\begin{align}
 \nonumber  & \sum_{k=k^{\prime}}^{ns-1}\log \bigg(\frac{\det \big(\mathbb{E}[{V^{i_k}}_{n_{i_k}(s_2)}]\big)}{\det\big(\mathbb{E}[{V^{i_k}}_{n_{i_k}(s_1)}]\big)}\bigg)\leq \\
   &\sum_{i\in \mathcal{M}} \log \bigg(\frac{\det \big(\mathbb{E}[{V^{i}}_{n_{i}(T)}]\big)}{\det\big(\mathbb{E}[{V^{i}}_{n_{i}(t_{m^{\prime}})}]\big)}\bigg)
\end{align}
in which we opened up the summation and re-arranged the terms based on the modes and afterwards, we applied the property $\log \frac{a}{b}+\log \frac{b}{c}=\log \frac{a}{c}$ for each mode's terms.
Conclusively, we have 

\begin{align}
  \Gamma_2\leq \frac{8\nu_{*}\mathbb{E}[\mu^{ns^-}_{*}]}{\sigma_{\omega}^2}\sum_{i\in \mathcal{M}} \log \bigg(\frac{\det \big(\mathbb{E}[{V^{i}}_{n_{i}(T)}]\big)}{\det\big(\mathbb{E}[{V^{i}}_{n_{i}(t_{m^{\prime}})}]\big)}\bigg) \label{eq:eska}
\end{align}
where $T$ is minimum expected time to accomplish mission. Note that in (\ref{eq:eska}), the summation of logarithmic terms are order of $\mathcal{O}(|\mathcal{M}|\log (|\mathcal{M}|\sqrt{ns}))$ given the fact that $T$ is of $\mathcal{O}(|\mathcal{M}|\sqrt{ns})$ by Theorem \ref{Thm:dwellTimeError}.

Furthermore, for any $i\in \mathcal{M}$ one can write:

\begin{align}
\nonumber &\mathbb{E}[\mu^{ns}_{i}]=\\
\nonumber &\quad \quad \quad \mathbb{E}\big[r^i_{n_i(T)}\bigg(1+2\vartheta_i \big(n_i(T)+\|\sum_{k=1}^{n_i(T)} z^i_k {z^i_k}^\top\|\big)^{0.5}\bigg)\big]\leq\\
&\quad \quad \quad \bar{R}_i\big(1+2\nu_i(1+2\kappa_i^2\alpha)^{0.5}\sqrt{n_i(T)}\big)\label{eq:boundExpMu}
\end{align}

By definition $\mathbb{E}[\mu^{ns}_{*}]=\max_{i\in |\mathcal{M}|}\mathbb{E}[\mu^{ns}_{i}]$ where considering the definition of $r^i_{n_i(T)}$, $\bar{R}_i$ is $\mathcal{O}(log |\mathcal{M}|\sqrt{ns})$. $n_i(T)\leq T$ and by Theorem \ref{Thm:dwellTimeError}, $T$ is of $\mathcal{O}(|\mathcal{M}|\sqrt{ns})$.  
This concludes that $\Gamma_2$ is of $\mathcal{O}\big(|\mathcal{M}|^{3/2}{ns}^{\frac{1}{4}}\big)$ where $\mathcal{O}$ absorbs logarithmic orders of $ns$ and $|\mathcal{M}|$.

\end{proof}
It is noteworthy that the regret term $\Gamma_1$ is imposed because of applying off-policy type of OFU-based strategy instead of on-policy one.

\subsection{Upper-bounding $R_T^2$}

Now, need to upper-bound the term $R_T^2$ which is enforced due to the learning error for estimating the mode-dependent minimum dwell-time. Before hand we need the following ingredients.

\begin{align}
  \nonumber  R_T^2=&\mathbb {E} [\sum _{k=0}^{ns-1}\sum _{t=\mathfrak{T}_k}^{\bar{\mathfrak{T}}_{k+1}-1} c_{t}^{\sigma(t)}]=\\
  \nonumber &\mathbb {E} [\sum _{k=0}^{ns-1}\sum _{t=\mathfrak{T}_k}^{\bar{\mathfrak{T}}_{k+1}-1} c_{t}^{i_k}]= \\
   \nonumber  &\mathbb {E} \bigg[\sum _{k=0}^{ns-1}\sum _{t=\mathfrak{T}_k}^{\bar{\mathfrak{T}}_{k+1}-1} x_t^\top \bigg(Q^{i_k}+\\
    &\quad \quad K^\top (\hat{\Theta}^{i_k}_{\bar{\mathfrak{T}}_{k}}, Q^{i_k}, R^{i_k})R^{i_k} K(\hat{\Theta}^{i_k}_{\bar{\mathfrak{T}}_{k}}, Q^{i_k}, R^{i_k})\bigg)x_t\bigg] \label{eq:secReg}
\end{align}

and from (\ref{eq:Lemma8_res22}) we have the following inequality

\begin{align}
  \nonumber  &Q^{i_k}+K^\top (\hat{\Theta}^{i_k}_{\bar{\mathfrak{T}}_{k}}, Q^{i_k}, R^{i_k})R^{i_k} K(\hat{\Theta}^{i_k}_{\bar{\mathfrak{T}}_{k}}, Q^{i_k}, R^{i_k}) \preceq \\
   \nonumber &P(\hat{\Theta}^{i_k}_{\bar{\mathfrak{T}}_{k}}, Q^{i_k}, R^{i_k})+\\
   \nonumber &2\mu^{\bar{\mathfrak{T}}_{k}}_{i_k}\|P(\hat{\Theta}^{i_k}_{\bar{\mathfrak{T}}_{k}}, Q^{i_k}, R^{i_k})\|_*\\
   \nonumber &  \begin{pmatrix} I \\ K(\hat{\Theta}^{i_k}_{\bar{\mathfrak{T}}_{k}}, Q^{i_k}, R^{i_k}) \end{pmatrix}{V^{i_k}}^{-1}_{n_{i_k}(\bar{\mathfrak{T}}_{k})}\begin{pmatrix} I \\ K(\hat{\Theta}^{i_k}_{\bar{\mathfrak{T}}_{k}}, Q^{i_k}, R^{i_k}) \end{pmatrix}^\top\preceq\\
  \nonumber  &\frac{\nu_{i^k}}{\sigma_{\omega}^2}I+\frac{\alpha_0^{i_k}}{2}\big(I+K^\top(\hat{\Theta}^{i_k}_{\bar{\mathfrak{T}}_{k}}, Q^{i_k}, R^{i_k})K(\hat{\Theta}^{i_k}_{\bar{\mathfrak{T}}_{k}}, Q^{i_k}, R^{i_k})\big)\preceq\\
   &\big(\frac{\nu_{i^k}}{\sigma_{\omega}^2}+\frac{\alpha_0^{i_k}}{2}(1+\kappa^2_{i^k})\big)I \label{eq:costupperBou}
\end{align}
where in the second inequality we used $\|P(\hat{\Theta}^{i_k}_{\bar{\mathfrak{T}}_{k}}, Q^{i_k}, R^{i_k})\|\leq \nu_{i^k}/\sigma^2_{\omega}$ and applied ${V^{i_k}}_{n_{i_k}(s_1)}^{-1}\leq 1/\lambda_{i^k}^{n_{i_k}}I$ and used the definition of $\lambda_{i^k}^{n_{i_k}}$ .

Applying the obtained bound (\ref{eq:costupperBou}) on (\ref{eq:secReg}) yields

\begin{align}
   \nonumber  R_T^2&\leq \big(\frac{\nu_{*}}{\sigma_{\omega}^2}+\frac{\alpha_0^{*}}{2}(1+\kappa^2_{*})\big)\alpha  \mathbb {E} \sum _{k=0}^{ns-1} (\tau_{md}^k-\tau_*^k)\\
   \nonumber &\leq \big(\frac{\nu_{*}}{\sigma_{\omega}^2}+\frac{\alpha_0^{*}}{2}(1+\kappa^2_{*})\big)\alpha T
\end{align}
where $T$ is an $\mathcal{O}(|\mathcal{M}|\sqrt{ns})$ term by Theorem \ref{Thm:dwellTimeError} which implies the term $\mathcal{R}_T^2=\mathcal{O}(|\mathcal{M}|\sqrt{ns})$.

 % \newpage
 
%   \begin{IEEEbiography}[{\includegraphics[width=1in,height=1.25in,clip,keepaspectratio]{Jafarpic.png}}]{Jafar Abbaszadeh Chekan} is a PhD candidate at the Aerospace Engineering department of University
% of Illinois at Urbana-Champaign (UIUC). He is also affiliated with the Decision \& Control
% Group at the Coordinated Science Lab (CSL). Prior joining UIUC, he received his master's degree in Engineering Mechanics from Virginia Tech. His
% research interest generally lies in controls, learning, game theory and their intersections.
% \end{IEEEbiography}
\vspace{-17cm}

\end{document}